\documentclass[reqno]{amsart}
\usepackage{amssymb}
\usepackage[dvips]{epsfig}
\usepackage{graphicx}
\usepackage{color}
\usepackage{amsmath}
\usepackage{amssymb}
\usepackage{amsfonts}
\usepackage{amsthm}
\newcommand{\R}{\mathbb R}

\newcommand{\Z}{\mathbb Z}

\newcommand{\C}{\mathbb C}

\newcommand{\N}{\mathbb{N}}
\newcommand{\T}{\mathbb{T}}

\newtheorem{thm}{Theorem}[section]
\newtheorem{lem}[thm]{Lemma}

\newtheorem{rem}{\bf Remark}[section]
\theoremstyle{definition}

\theoremstyle{statement}
\newtheorem{state}[thm]{Statement}

\numberwithin{equation}{section}
\usepackage[colorlinks=true,pdfstartview=FitV,linkcolor=magenta,citecolor=cyan]{hyperref}
\usepackage{bm}
\begin{document}
 \title[Quantitative Green's  function estimates]{Quantitative Green's  function estimates for Lattice Quasi-periodic Schr\"odinger Operators}

\author[Cao]{Hongyi Cao}
\address[H. Cao] {School of Mathematical Sciences,
Peking University,
Beijing 100871,
China}
\email{chyyy@stu.pku.edu.cn}
\author[Shi]{Yunfeng Shi}
\address[Y. Shi] {College of Mathematics,
Sichuan University,
Chengdu 610064,
China}
\email{yunfengshi@scu.edu.cn}

\author[Zhang]{Zhifei Zhang}
\address[Z. Zhang] {School of Mathematical Sciences,
Peking University,
Beijing 100871,
China}
\email{zfzhang@math.pku.edu.cn}

\date{\today}


\begin{abstract}
In this paper,  we establish quantitative Green's function estimates for some higher dimensional  lattice quasi-periodic (QP)  Schr\"odinger operators.  The resonances in the estimates can be described via  a pair of symmetric zeros  of certain  functions   and the estimates  apply to  the sub-exponential type non-resonant conditions.  As the application of quantitative Green's function estimates,  we  prove both the  arithmetic version of Anderson localization and  finite volume version of $(\frac 12-)$-H\"older continuity of the integrated density of states (IDS) for  such QP Schr\"odinger operators. This gives an affirmative answer to  Bourgain's  problem  in   \cite{Bou00}. 
\end{abstract}

\maketitle

\maketitle

\section{Introduction}
Consider  the  QP Schr\"odinger operators
\begin{align}\label{SO}
	H=\Delta+\lambda V(\theta+n\omega)\delta_{n, n'}\   {\rm on}\  \Z^d,
\end{align}
where $\Delta$ is the discrete Laplacian,  $V:\T^d=(\R/\Z)^ d\to\R$ is the potential and $n\omega=(n_1\omega_1,\dots,n_d\omega_d)$. Typically, we call $\theta\in\T^d$ the phase,  $\omega\in[0,1]^d$ the frequency and   $\lambda\in\R$ the coupling .   Particularly, if $V=2\cos2\pi\theta$ and $d=1$, then the operators  \eqref{SO} become  the famous almost Mathieu operators (AMO). \smallskip

Over the past decades,  the study of spectral and dynamical properties of lattice QP Schr\"odinger operators has been one of the central themes in mathematical physics. Of particular importance is the phenomenon  of   Anderson localization (i.e., pure point spectrum with exponentially decaying eigenfunctions).   Determining the nature of the spectrum  and the eigenfunctions properties of \eqref{SO} can be viewed as a small divisor problem, which depends sensitively on  features of $\lambda, V, \omega, \theta$ and $d$.   Then substantial progress has been made following   Green's function estimates  based on  a KAM type multi-scale analysis  (MSA)  of Fr\"ohlich-Spencer \cite{FS83}.   More precisely, Sinai \cite{Sin87} first proved the Anderson localization for a class of $1D$ QP Schr\"odinger operators with a  $C^2$ cosine-like potential assuming the Diophantine frequency \footnote{We say $\omega\in\R$ satisfies the Diophantine condition if there are $\tau>1$ and $\gamma>0$ so that $$\|k\omega\|=\inf_{l\in\Z}|l-k\omega|\geq \frac{\gamma}{|k|^\tau}\ {\rm for}\ \forall\ k\in\Z\setminus\{0\}.$$}.  The proof focuses on  eigenfunctions parametrization and the resonances are overcome via a  KAM iteration scheme. Independently, Fr\"ohlich-Spencer-Wittwer \cite{FSW90} extended the celebrated method of Fr\"ohlich-Spencer  \cite{FS83}  originated from random Schr\"odinger operators case  to the QP one,  and obtained similar Anderson localization result with \cite{Sin87}. The  proof  however uses  estimates of finite volume Green's functions based on the MSA and the eigenvalue variations.  Both  \cite{Sin87} and \cite{FSW90} were  inspired essentially  by  arguments of \cite{FS83}.   Eliasson \cite{Eli97} applied a reducibility method based on KAM iterations to  general Gevrey QP potentials and established the pure point spectrum for corresponding Schr\"odinger operators. All these $1D$  results are perturbative in the sense that the required perturbation strength depends heavily on the Diophantine frequency  (i.e., localization holds for $|\lambda|\geq \lambda_0(V,\omega)>0$). The great breakthrough  was then made by  Jitomirskaya \cite{Jit94, Jit99}, in which the non-perturbative methods for control of  Green's functions (cf. \cite{Jit02})  were developed first  for AMO. The non-perturbative methods can avoid the usage of multi-scale scheme and the eigenvalue variations. This will allow effective (even optimal in many cases) and independent of $\omega$ estimate on $\lambda_0$. In addition, such  methods can provide arithmetic version of Anderson localization which means the removed sets on both  $\omega$ and $\theta$ when obtaining localization have an explicit arithmetic description (cf. \cite{Jit99,JL18} for details).  In contrast, the current perturbation methods seem only providing certain measure  or complexity bounds on these sets.  Later,  Bourgain-Jitomirskaya \cite{BJ02} extended some results of \cite{Jit99} to the exponential long-range hopping case (thus the absence of Lyapunov exponent) and obtained both nonperturbative and arithmetic Anderson localization.  Significantly, Bourgain-Goldstein \cite{BG00} generalized  the non-perturbative Green's function estimates  of Jitomirskaya \cite{Jit99} by introducing the new ingredients of  semi-algebraic sets theory and subharmonic function estimates,  and  established  the non-perturbative Anderson localization\footnote{i.e., Anderson localization assuming the positivity of the Lyapunov exponent. In the present context by nonperturbative Anderson localization we mean localization if  $|\lambda|\geq \lambda_0=\lambda_0(V)>0$ with $\lambda_0$ being independent of $\omega$.}  for general analytic QP potentials.  The localization results of \cite{BG00}  hold for arbitrary $\theta\in\T$ and a.e. Diophantine frequencies (the permitted set of frequencies depends  on $\theta$),  and  there seems  no  arithmetic  version of Anderson localization  results  in this case.    We would  like to mention that the Anderson localization can  also be obtained via reducibility arguments based on Aubry duality \cite{JK16,AYZ17}.\smallskip

If one increases the lattice dimensions of  QP operators,  the Anderson localization proof  becomes significantly difficult. In this setting, Chulaevsky-Dinaburg \cite{CD93} and Dinaburg \cite{Din97}  first extended  results of Sinai \cite{Sin87} to the exponential long-range operator with a $C^2$ cosine type potential on $\Z^d$ for arbitrary $d\geq1.$   However,  in this  case, the localization holds assuming further restrictions on the frequencies (i.e., localization only holds for  frequencies in a set of positive measure, but without explicit arithmetic description).  Subsequently,  the remarkable work of Bourgain-Goldstein-Schlag \cite{BGS02}  established the Anderson localization for the general analytic QP Schr\"odinger operators with  $(n,\theta,\omega)\in\Z^2\times\T^2\times\T^2$ via Green's function estimates. In \cite{BGS02} they first proved the large deviation theorem (LDT) for the finite volume Green's functions by combining MSA, matrix-valued Cartan's estimates and semi-algebraic sets  theory. Then by  using further semi-algebraic arguments together with LDT, they proved the Anderson localization for all $\theta\in\T^2$ and $\omega$ in a set of positive measure (depending on $\theta$). While the restrictions of the frequencies when achieving LDT are purely arithmetic and do not depend on the choice of potentials, in order to obtain the Anderson localization it needs to remove an additional  frequencies set of positive measure.   The proof of \cite{BGS02} is essentially two-dimensional and a generation  of it to higher dimensions  is significantly difficult.  In 2007,   Bourgain \cite{Bou07} successfully extended the  results of  \cite{BGS02} to arbitrary dimensions, and one of his  key ideas is allowing the restrictions of frequencies to depend on the potential  by means of delicate  semi-algebraic sets analysis  when proving LDT for Green's functions.  In other words,  for the proof of LDT in \cite{Bou07} there has already been    additional restrictions on the frequencies,  which depends on the potential $V$ and is thus not arithmetic. The results of \cite{Bou07} have been largely generalized by Jitomirskaya-Liu-Shi  \cite{JLS20} to  the  case of both arbitrarily  dimensional multi-frequencies and exponential long-range hopping. Very recently, Ge-You \cite{GY20}  applied a reducibility argument  to  higher dimensional  long-range QP operators with the cosine potential, and proved the  first arithmetic Anderson localization assuming the  Diophantine  frequency.

Definitely, the LDT type Green's function estimates methods are powerful to deal with higher dimensional QP  Schr\"odinger operators with general analytic potentials.   However, such methods do not provide the detailed information on Green's functions and eigenfunctions that may be extracted by purely perturbative method based on Weierstrass preparation type theorem.  As an evidence, in the celebrated work \cite{Bou00}, Bourgain developed the method of  \cite{Bou97} further to first obtain the finite volume version of  $(\frac12-)$-H\"older continuity of the IDS for AMO. The proof shows that the Green's functions can be controlled via certain quadratic polynomials, and the resonances are completely determined by zeros of these polynomials.  Using this method then yields a surprising quantitative  result on the H\"older exponent of the IDS,  since the celebrated method of Goldstein-Schlag \cite{GS01} which is non-perturbative and works for more general potentials does not seem to provide explicit information on the H\"older exponent. In 2009,  by using KAM reducibility method of Eliasson \cite{Eli92}, Amor \cite{Amo09} obtained the first  $\frac12$-H\"older continuity result of the IDS for $1D$ and multi-frequency QP Schr\"odinger operators with small analytic potentials and Diophantine frequencies.  Later,  the one-frequency result of Amor was largely generalized by Avila-Jitomirskaya \cite{AJ10} to the non-perturbative case via the quantitative almost reducibility and localization method.  In  the regime  of  the positive Lyapunov exponent,   Goldstein-Schlag \cite{GS08} successfully proved the $(\frac{1}{2m}-)$-H\"older continuity of the IDS for $1D$ and one-frequency QP Schr\"odinger operators with  potentials given by analytic perturbations of certain trigonometric polynomials of degree $m\ge1$. This celebrated work provides in fact the  finite volume version of  estimates on the IDS.  We  remark that the  H\"older continuity of the IDS for  $1D$ and multi-frequency QP Schr\"odinger operators with large  potentials is hard to prove. In \cite{GS01}, by using the LDT for transfer matrix  and the avalanche principle, Goldstein-Schlag showed the weak H\"older continuity (cf. \eqref{weakids}) of  the IDS for $1D$ and multi-frequency QP Schr\"odinger operators assuming the positivity of the Lyapunov exponent and strong Diophantine frequencies. The weak H\"older continuity of the IDS for higher dimensional  QP Schr\"odinger operators  has been  established  in \cite{Sch01,Bou07, Liu20}. Very recently, Ge-You-Zhao \cite{GYZ22} proved the  $(\frac{1}{2m}-)$-H\"older continuity of  the  IDS for higher dimensional QP  Schr\"odinger operators with  small exponential long-range hopping and trigonometric polynomial (of degree $m$) potentials via the reducibility argument.  By Aubry duality, they can obtain the  $(\frac{1}{2m}-)$-H\"older continuity of the IDS for $1D$ and multi-frequency QP operators with a  finite range hopping.

Of course,  the  references  mentioned as above  are far from complete and we refer the reader to \cite{Bou05, MJ17, Dam17} for more recent results on the study of both Anderson localization and the H\"older regularity of the IDS for lattice QP Schr\"odinger operators. 
\subsection{Bourgain's problems}

The remarkable Green's function estimates  of \cite{Bou00} should be not restricted   to  the proof of $(\frac12-)$-H\"older regularity of the IDS for AMO only. In fact, in \cite{Bou00} (cf. Page 89),   Bourgain made three comments on the possible extensions of his method: 

\begin{itemize}

\item[(1)] In fact, one may also recover  the Anderson localization results from \cite{Sin87} and \cite{FSW90} in the perturbative case;

\item[ (2)] One may hope that it may be combined with nonperturbative arguments in the spirit of \cite{BG00,GS01} to establish $(\frac12-)$-H\"older regularity assuming  positivity of the Lyapunov exponent only;

\item[(3)] It may also allow progress in the multi-frequency case (perturbative or nonperturbative) where regularity estimates of the form (0.28)\footnote{i.e, a weak H\"older continuity estimate  \begin{align}\label{weakids}|\mathcal{N}(E)-\mathcal{N}(E')|\leq e^{-\left(\log\frac{1}{|E-E'|}\right)^\zeta},\ \zeta\in(0,1),\end{align}  where $\mathcal{N}(\cdot)$ denotes the IDS.} are the best obtained so far.

 \end{itemize}

An  extension of (2) has been accomplished  by Goldstein-Schlag \cite{GS08}.  The answer to the extension of (1) is highly nontrivial due to the following reasons:

\begin{itemize} 

\item The Green's function on {\bf good} sets (cf. Section \ref{GFES} for details) only has a sub-exponential off-diagonal decay estimate rather than an exponential one required by proving Anderson localization;

\item At the $s$-th iteration step ($s\geq1$),  the resonances of \cite{Bou00}  are characterized as
\begin{align*}
	\min\{\|\theta+k\omega-\theta_{s,1}\|, \|\theta+k\omega-\theta_{s,2}\|\}\leq \delta_s\sim \delta_0^{C^s},\ C>1.
\end{align*}
However,  the  symmetry  information  of  $\theta_{s,1}$  and $\theta_{s,2}$ is missing. 
Actually, in \cite{Bou00},  it might be  $\theta_{s,1}+\theta_{s,2}\neq 0$ because of the construction of resonant blocks;

\item If one tries to extend the method of Bourgain \cite{Bou00} to higher lattice dimensions,  there comes new difficulty:  the resonant blocks at each iteration  step could not be the cubes  similar to the intervals appeared in the $1D$ case. 
\end{itemize}

To extend the method of Bourgain \cite{Bou00} to higher lattice dimensions and recover the Anderson localization, one has to address the above issues, which is our main motivation of this paper.

\subsection{Main results}

In this paper, we study  the QP Schr\"odinger operators on $\mathbb{Z}^d$
\begin{align}\label{model}
	H(\theta)=\varepsilon \Delta+\cos2\pi(\theta+ n\cdot{\omega})\delta_{n, n'},\ \varepsilon>0,
\end{align}
where the discrete Laplacian $\Delta$ is defined as
\begin{align*}
	\Delta( n, n')=\delta_{{\|n- n'\|_{1}, 1}},\ \| n\|_{1}:=\sum_{i=1}^{d}\left|n_{i}\right|.
\end{align*}
For the diagonal part of \eqref{model}, we have $\theta\in \mathbb{T}=\mathbb{R}/\mathbb{Z},  \omega\in[0,1]^d$ and $n\cdot\omega=\sum\limits_{i=1}^dn_i \omega_i.$ Throughout the paper,  we assume that $\omega\in \mathcal{R}_{\tau,\gamma}$ for some $0<\tau<1$
and $\gamma>0$   with
\begin{align}\label{bruno}
	\mathcal{R}_{\tau, \gamma}=\left\{\omega\in[0,1]^d:\ \|n\cdot\omega\|=\inf_{l\in\mathbb{Z}}|l- n\cdot\omega|\geq \gamma e^{-\| n\|^{\tau}}\ {\rm for}\ \forall\  n\in\mathbb{Z}^d\setminus\{0\}\right\},
\end{align}
where
$$\| n\|:=\sup _{1\leq i \leq d}\left|n_{i}\right|.$$
 We aim to extend the method of Bourgain \cite{Bou00} to higher lattice dimensions and establish quantitative  Green's function estimates assuming \eqref{bruno}. As the application, we  prove  the arithmetic  version of Anderson localization and  the finite volume version of $(\frac12-)$-H\"older continuity of the IDS for \eqref{model}. 

\subsubsection{Quantitative Green's function estimates}
The first main result of this paper is a quantitative version of Green's function estimates, which will imply both arithmetic Anderson localization and the finite volume version of $(\frac12-)$-H\"older continuity of IDS. The estimates on Green's function are based on multi-scale induction arguments.\smallskip

Let $\Lambda\subset \Z^d$ and denote by $R_\Lambda$ the restriction operator. Given $E\in\R$, the Green's function (if exists) is defined by
$$
T^{-1}_\Lambda(E;\theta)=\left(H_\Lambda(\theta)-E\right)^{-1},\ H_\Lambda(\theta)=R_\Lambda H(\theta)R_\Lambda.
$$

Recall that $\omega\in\mathcal{R}_{\tau,\gamma}$ and $\tau\in(0,1)$. We fix a constant $c>0$ so that
$$
1<c^{20}<\frac1\tau.
$$
At the $s$-th iteration step,  let $\delta_s^{-1}$ (resp. $N_s$)  describe  the  resonance strength (resp. the size of resonant blocks) defined by
$$N_{s+1}=\left[\left|\log\frac{\gamma}{\delta_s}\right|^{\frac{1}{c^5\tau}}\right], \  \left|\log\frac{\gamma}{\delta_{s+1}}\right|=\left|\log\frac{\gamma}{\delta_s}\right|^{c^5},\ \delta_0=\varepsilon^{\frac{1}{10}},
$$
where $[x]$ denotes the integer part of $x\in\R.$  

  If $a\in\R$, let $\|a\|={\rm dist}(a,\Z)=\inf\limits_{l\in\Z}|l-a|$. For $z=a+\sqrt{-1}b\in \C$ with $a,b\in\R$, define $\|z\|=\sqrt{|b|^2+\|a\|^2}.$ Denote by ${\rm dist} (\cdot,\cdot)$ the distance induced by the supremum norm on $\R^d$. Then we have
\begin{thm}\label{ggg}
Let $\omega\in\mathcal{R}_{\tau,\gamma}$. Then there is some $\varepsilon_0=\varepsilon_0(d,\tau,\gamma)>0$ so that, for $0<\varepsilon\leq \varepsilon_0$ and $ E\in[-2,2]$,  there exists a sequence $\{\theta_s=\theta_s(E)\}_{s=0}^{s'}\subset \C$ ($s'\in \N\cup \{+\infty\}$)  with the following properties. Fix any $\theta\in\T$. If a finite set $\Lambda\subset \mathbb{Z}^d$ is $s$-{\rm\bf good}  ({\rm cf}.  $({\bf e})_s$ of the  {\bf Statement}  \ref{state} for the definition of $s$-{\rm\bf good} sets, and  Section \ref{GFES} for the definitions of $\{\theta_s\}_{s=0}^{s'}$, the sets $P_s, Q_s,  \widetilde{\Omega}_k^s$), then
\begin{align*}
	\|T^{-1}_{\Lambda}(E;\theta)\|&<\delta^{-3}_{s-1}\sup_{\left\{k\in P_s:\   \widetilde{\Omega}_{k}^s\subset \Lambda \right\}}\|\theta+k\cdot\omega-\theta_{s}\|^{-1}\cdot \|\theta+k\cdot\omega+\theta_{s}\|^{-1}\\
	&<\delta^{-3}_s,\\
	\left|T^{-1}_{\Lambda}(E;\theta)(x, y)\right|&<e^{-\frac14|\log\varepsilon|\cdot\|x-y\|_1} \ { \rm for }\ \|x-y\|>N_{s}^{c^3}.
\end{align*}
In particular, for any finite set $\Lambda\subset\Z^d$, there exists some $\widetilde \Lambda$ satisfying 
	$$	\Lambda\subset\widetilde\Lambda\subset \left\{k\in\Z^d:\ {\rm dist}(k, \Lambda)\leq 50N_s^{c^2}\right\}$$
	so that, if
	$$
	\min_{k\in\widetilde \Lambda^*}\min_{\sigma=\pm 1}(\|\theta+ k\cdot\omega+\sigma \theta_s\|)\geq \delta_s,
	$$
	then
	\begin{align*}
	\|T^{-1}_{\widetilde\Lambda}(E;\theta)\|&\leq \delta_{s-1}^{-3}\delta_s^{-2}\leq \delta_s^{-3},\\
	|T^{-1}_{\widetilde\Lambda}(E;\theta)(x,y)|&\leq e^{-\frac14|\log\varepsilon|\cdot\|x-y\|}\ {\rm for}\ \|x-y\|> N_s^{c^3},
	\end{align*}
	where
	$$
	\widetilde\Lambda^*=\left\{k\in\frac12\Z^d:\ {\rm dist}(k, \widetilde\Lambda)\leq \frac12\right\}.
	$$

\end{thm}

Let us refer to Section \ref{GFES} for a complete description of our Green's function estimates.


\subsubsection{Arithmetic Anderson localization and H\"older continuity of the IDS}

As the application of quantitative Green's function estimates, we first  prove the following arithmetic version of Anderson localization for $H(\theta)$. Let $\tau_1>0$ and define
$$
\Theta_{\tau_1}=\{(\theta, \omega)\in\T\times\mathcal{R}_{\tau,\gamma}:\   \|2\theta+n \cdot\omega\|\leq e^{-\| n\|^{\tau_1}}\ {\rm holds\  for\  finitely\  many }\   n\in\Z^d\}.
$$

We have
\begin{thm}\label{thm1}
	Let $H(\theta)$ be given by \eqref{model} and let $0<\tau_1<\tau$. Then there exists some $\varepsilon_0=\varepsilon_0(d,\tau, \gamma)>0$ such that, if $0<\varepsilon\leq\varepsilon_0$, then for  $(\theta,\omega)\in \Theta_{\tau_1}$,  $H(\theta)$ satisfies the Anderson localization.
\end{thm}
\begin{rem}\label{remfp}
It is easy to check both  ${\rm mes}(\T\setminus\Theta_{\tau_1,\omega})=0$ and  ${\rm mes}(\mathcal{R}_{\tau,\gamma}\setminus\Theta_{\tau_1,\theta})=0$, where $\Theta_{\tau_1,\omega}=\{\theta\in\T:\ (\theta,\omega)\in\Theta_{\tau_1}\}$, $\Theta_{\tau_1,\theta}=\{\omega\in\mathcal{R}_{\tau,\gamma}:\ (\theta,\omega)\in\Theta_{\tau_1}\}$ and ${\rm mes}(\cdot)$ denotes the Lebesgue measure.  Thus Anderson localization can be established either by fixing  $\omega\in\mathcal{R}_{\tau,\gamma}$ and removing $\theta$ in the spirit of \cite{Jit99},  or by fixing $\theta\in \T$ and removing $\omega$ in the spirit of \cite{BG00,BGS02}.
\end{rem}

The second application is a proof  of  the  finite volume version of $(\frac12-)$-H\"older continuity of the IDS for $H(\theta).$  For a finite set $\Lambda$, denote by $\#\Lambda$ the cardinality of $\Lambda.$ Let
$$\mathcal{N}_{\Lambda}(E;\theta)=\frac{1}{\#\Lambda}\#\{\lambda\in\sigma({H_\Lambda(\theta)}):\  \lambda\leq E\}$$
and denote by
\begin{align}\label{ids}
	\mathcal{N}(E)=\lim_{N\to\infty}\mathcal{N}_{\Lambda_N}(E;\theta)
\end{align}
the IDS, where $\Lambda_N=\{k\in\Z^d:\ \|k\|\leq N\}$ for $N>0$. It is well-known that the limit in  \eqref{ids} exists and is independent of $\theta$  for a.e. $\theta$. 

\begin{thm}\label{thm2}
	Let $H(\theta)$ be given by \eqref{model} and let $\omega\in \mathcal{R}_{\tau,\gamma}$. Then there exists some $\varepsilon_0=\varepsilon_0(d,\tau,\gamma)>0$ such that if $0<\varepsilon\leq\varepsilon_0$, then for any small $\mu>0 $ and    $0<\eta
	<\eta_0(d,\tau,\gamma,\mu )$, we have for sufficiently large $N$ depending on $\eta$,
	\begin{align}\label{hold}
		\sup_{\theta\in\T, E\in\R}\left(\mathcal{N}_{\Lambda_N}(E+\eta;\theta)-\mathcal{N}_{\Lambda_N}(E-\eta;\theta)\right)\leq\eta^{\frac{1}{2}-\mu}.
	\end{align}	
	In particular, the IDS is  H\"older continuous of exponent $\iota$ for any $\iota\in(0, \frac{1}{2})$.  
\end{thm}

 
Let us give some remarks on our results.

\begin{itemize}

\item[(1)]
The Green's function estimates can be extended to the exponential long-range hopping case, and may not be  restricted to the cosine potential.   
 Except for the proof of arithmetic Anderson localization and the finite volume version of  $(\frac12-)$-H\"older regularity of the IDS, the  quantitative Green's function estimates  should have potential applications in other  problems, such as the estimates of Lebesgue measure of the spectrum, dynamical localization, the estimates of level spacings of eigenvalues  and finite volume version of  localization.  We can even expect fine results in dealing with the Melnikov's persistency problem (cf. \cite{Bou97}) by employing our Green's function estimates method.

\item[ (2)]  
As mentioned previously, Ge-You \cite{GY20} proved the first arithmetic Anderson localization result for  higher dimensional QP operators with  the  exponential long-range hopping and the cosine potential via their reducibility method. 
 Our result is valid for frequencies satisfying the sub-exponential non-resonant condition (cf. \eqref{bruno}) of   R\"ussmann type  \cite{Rus80}, which slightly generalizes the Diophantine type localization result of \cite{GY20}.  While the  R\"ussmann type condition is  sufficient for the use of  classical KAM method, it is not clear if such condition still suffices for the validity of MSA method.  Definitely,  the localization result of both \cite{GY20}  and the present work   is   perturbative   \footnote{In fact,   Bourgain \cite{Bou02} has proven  that the  non-perturbative localization can not  be expected  in dimensions $d\geq2$.   More precisely,   consider 	${H}^{(2)}=\lambda\Delta+2\cos2\pi(\theta+n\cdot \omega)\delta_{n,n'}$ on $\Z^2$. Using Aubry duality together with  result of Bourgain \cite{Bou02}  yields  for any $\lambda\neq0$, there exists a set  $\Omega\subset\mathbb{T}^2$ of  positive measure with the following property, namely, for  $\omega\in \Omega$,   there exists a  set $\Theta\subset\mathbb{T}$ of positive measure, s.t.,  for $\theta\in \Theta$, ${H}^{(2)}$ does not satisfy Anderson localization.}.  Finally, since our proof of arithmetic  Anderson localization is based on Green's function estimates,  it  could  be improved to obtain  the finite volume version of Anderson localization as that obtained in \cite{GS11}.

\item[(3)]  Apparently,  using the Aubry duality  together with Amor's result \cite{Amo09} has already led to the $\frac12$-H\"older continuity of the IDS for higher dimensional QP operators with small exponential long-range hopping and the cosine potential assuming Diophantine frequencies.  So our result of $(\frac12-)$-H\"older continuity is weaker than that of \cite{Amo09} in the Diophantine frequencies case. However, we want to emphasize  that the method of Amor seems only valid for estimating the limit   $\mathcal{N}(E)$ and provides no precise information on the  finite volume quantity   $\mathcal{N}_{\Lambda}(E;\theta)$. In this context, our result  (cf.  \eqref{hold})  is also new as it gives uniform  upper bound  on the number of eigenvalues inside a small interval.   In addition,  our result also improves the upper bound  on the number of eigenvalues  of Schlag (cf. Proposition 2.2 of \cite{Sch01})  in the special case that the potential is given by the cosine function.

 \end{itemize}

\subsection{Notations and structure of the paper}
   \begin{itemize}
  \item Given $A\in\C$ and $B\in\C$, we write $A\lesssim B$ (resp. $A\gtrsim$ B) if there is some $C=C(d,\tau,\gamma)>0$ depending only on $d, \tau, \gamma$ so that $|A|\leq C |B|$ (resp. $|A|\geq C|B|$).  We also denote  $A\sim B\Leftrightarrow \frac{1}{C}<\left| \frac{A}{B}\right|<C $,  and  for some $D>0$, $A\stackrel{D}{\sim}B\Leftrightarrow \frac{1}{CD}<\left| \frac{A}{B}\right| <CD$. 
  \item 
  The determinant of a matrix $M$ is denoted by $\det M.$ 
  \item For $n\in\R^d$, let $\|n\|_{1}:=\sum\limits_{i=1}^{d}\left|n_{i}\right| $
   	and $\|n\|:=\sup\limits _{1\leq i \leq d}\left|n_{i}\right| .$ Denote by ${\rm dist}(\cdot,\cdot)$ the distance induced by $\|\cdot\|$ on $\R^d,$
	and define $$\operatorname{diam}\Lambda=\sup_{k,k'\in \Lambda}\|k-k'\|.
	$$
	Given $n\in\Z^d$, $\Lambda_1\subset \frac12\Z^d$ and $L>0$,  denote $\Lambda_{L}(n)=\{k\in\Z^d:\ \|k-n\|\leq L\}$
	and $\Lambda_{L}(\Lambda_1)=\{k\in\Z^d:\ {\rm dist}(k, \Lambda_1)\leq L\}$. In particular,   write $\Lambda_L=\Lambda_L(0).$
	
		 \item  Assume $\Lambda'\subset\Lambda\subset \Z^d$. Define the relatively boundaries as  $\partial^+_{\Lambda}\Lambda'=\{k\in\Lambda:\  {\rm dist}(k,\Lambda')=1\}$, $\partial^-_{\Lambda}\Lambda'=\{k\in \Lambda :\  {\rm dist}(k,\Lambda\setminus\Lambda')=1\}$ and $\partial_{\Lambda}\Lambda'=\{(k,k'):\ \|k-k'\|=1, k\in  \partial^-_{\Lambda}\Lambda',k'\in  \partial^+_{\Lambda}\Lambda'\}$.
\item Let $\Lambda\subset\Z^d$ and let $T:\ \ell^2(\Z^d)\to\ell^2(\Z^d)$ be a linear operator. Define $T_{\Lambda}=R_{\Lambda} T R_{\Lambda}$, where $R_\Lambda$  is the restriction operator.  Denote by $\left\langle \cdot,\cdot \right\rangle$ the standard inner product on $\ell^2(\Z^d).$  Set $T_\Lambda(x,y)=\left\langle \delta_x,T_\Lambda \delta_y \right\rangle \ \text{for $x,y\in \Lambda$}$. By $\|T_\Lambda\|$ we mean the standard operator norm of $T_\Lambda.$ The spectrum of the operator $T$ is denoted by $\sigma(T).$ Finally, $I$ typically  denotes   the  identity operator. 
 \end{itemize}

The paper is organized as follows.  The key ideas of the proof are  introduced in \S 2. The proofs of Theorems \ref{ggg}, \ref{thm1} and \ref{thm2} are presented in \S3, \S4 and \S5, respectively.  Some useful estimates can be found in the appendix.

\section{Key ideas of the proof}

The main scheme of our proof is definitely adapted from Bourgain \cite{Bou00}.  The key ingredient of the proof in \cite{Bou00} is that the resonances in dealing with Green's function estimates can be completely determined by the roots of some quadratic polynomials. The polynomials were produced in a  Fr\"ohlich-Spencer  type  MSA  induction procedure.  However, in the estimates of Green's functions restricted on the resonant blocks,  Bourgain applied directly the Cramer's rule and provided estimates on  certain determinants. It turns out  these determinants  can be well controlled  via estimates of previous induction steps,  the Schur complement argument and  Weierstrass preparation theorem.  It is the preparation type technique  that  yields  the desired quadratic polynomials.  We emphasize that this new method of Bourgain  is  fully free from  eigenvalues variations or eigenfunctions parametrization.
\smallskip

However, in order to extend the  method to achieve arithmetic version of Anderson localization in  higher dimensions, some new ideas are required:\smallskip

\begin{itemize}
	\item The off-diagonal decay of the Green's function  obtained by Bourgain \cite{Bou00} is sub-exponential rather than exponential, which is not sufficient for a proof of Anderson localization. We resolve this issue by modifying  the definitions  of the resonant blocks $\Omega_k^s\subset\widetilde \Omega_k^s\subset\Z^d$, and allowing 
	$${\rm diam}\  \Omega_k^s\sim ({\rm diam}\ \widetilde\Omega_k^s) ^\rho,\  0<\rho<1.$$
	This sublinear bound is crucial for a proof of exponential off-diagonal decay. In the argument of Bourgain,  it requires actually that  $\rho=1.$  Another issue we want to highlight is that Bourgain just  provided  outputs  of iterating resolvent identity in many places of the paper \cite{Bou00}, but did not present the details.  This motivates  us to write down the whole iteration arguments that is also important to the exponential decay estimate. 
	
	\item  To prove Anderson localization, one has to eliminate the  energy  $E\in\R$ appeared in the  Green's function estimates by removing  $\theta$ or $\omega$  further.  Moreover, if one wants to prove  an arithmetic version of Anderson localization,  a geometric description of resonances  (i.e.,  the symmetry of zeros of certain functions appearing as the perturbations of quadratic polynomials in the present context) is essential. Precisely, at the $s$-th iteration step,  using the Weierstrass preparation theorem Bourgain \cite{Bou00} had shown the existence  of zeros $\theta_{s,1}(E)$ and $\theta_{s,2}(E)$,   but  provided  no symmetry information.  Indeed, the symmetry property of  $\theta_{s,1}(E)$ and $\theta_{s,2}(E)$  relies highly  on  that of resonant blocks $\widetilde\Omega_k^s$.  However, in the construction of $\widetilde\Omega_k^s$ in \cite{Bou00},  the  symmetry property is missing. In this paper,  we prove in fact 
	$$\theta_{s,1}(E)+\theta_{s,2}(E)=0.$$
	The main idea is that  we   reconstruct $\widetilde\Omega_k^s$ so that it is  symmetrical about $k$ and allow  the center $k\in\frac12 \Z^d. $  
	
	\item In the construction of resonant blocks  \cite{Bou00},  the property that 
	\begin{align}\label{inclu}
		\widetilde\Omega_{k'}^{s'}\cap\widetilde\Omega^s_{k}\neq \emptyset\Rightarrow\widetilde\Omega_{k'}^{s'}\subset\widetilde\Omega_k^s\  {\rm for}\ s'<s
	\end{align}
	plays a center role. In the $1D$ case,  $\widetilde\Omega_k^s$ can be defined as an interval so that \eqref{inclu}  holds true. This interval structure of  $\widetilde\Omega_k^s$  is important to get  desired estimates  using resolvent identity. However, to generalize  this argument to higher dimensions, one  needs to give up the ``interval''  structure of $\widetilde\Omega_k^s$  in order to fulfill  the property \eqref{inclu}.  As a result, the geometric description of $\widetilde\Omega_k^s$ becomes significantly complicated,  and  the estimates relying on resolvent identity remain unclear. We address this issue by proving that  $\widetilde\Omega_k^s$ can be  constructed satisfying \eqref{inclu} and  staying in some enlarged cubes, such as 
	$$\Lambda_{N_s^{c^2}}\subset\widetilde\Omega_k^s-k\subset \Lambda_{N_s^{c^2}+50N_{s-1}^{c^2}}.$$
	
	\item We want to mention that in the estimates of zeros for some perturbations of quadratic polynomials,  we use the standard R\'ouche theorem rather than the Weierstrass preparation theorem  as in \cite{Bou00}.   This  technical modification avoids controlling the first order derivatives of determinants and simplifies significantly the proof. 
\end{itemize}

The proofs of both Theorem \ref{thm1} and  Theorem \ref{thm2} follow  from the estimates in Theorem \ref{ggg}. 
   
\section{Quantitative Green's function estimates}\label{GFES}
The spectrum $\sigma(H(\theta))\subset[-2,2]$ since $\|H(\theta)\|\leq1+2d\varepsilon<2$ if $0<\varepsilon<\frac{1}{2d}.$
In this section, we fix
\begin{align*}
\theta \in \mathbb{T},\ E\in[-2,2].
\end{align*}
Write   $$ E=\cos 2 \pi \theta_{0}$$ with $\theta_{0} \in \mathbb{C}.$ 
 Consider \begin{equation}\label{T}
 	T(E; \theta)=H(\theta)-E=D_n\delta_{n,n'}+\varepsilon \Delta,
 \end{equation}where
 \begin{align}\label{DN}
 	D_{n}=\cos2\pi(\theta+n\cdot{\omega})-E.
 \end{align}
 For simplicity, we may omit the dependence of $T(E; \theta)$ on $E, \theta$ below. 

We will use a multi-scale analysis induction to provide estimates on Green's functions. Of particular importance is the analysis of resonances, which will be described by zeros of certain functions appearing as perturbations of some quadratic polynomials. Roughly speaking, at the $s$-th iteration step,  the  set  $Q_s\subset\frac12\Z^d$ of singular sites    will be  completely described by a pair of symmetric zeros of certain functions, i.e., 
$$
Q_{s}=\bigcup_{\sigma=\pm 1}\left\{k \in P_{s} :\  \left\|\theta+k\cdot \omega +\sigma \theta_{s}\right\|<\delta_{s}\right\}.
$$
While  the Green's functions  restricted on $Q_s$ can not be generally well controlled,  the algebraic structure of $Q_s$ combined with the  non-resonant  condition  of $\omega$ may lead to fine separation property of singular sites. As a result, one can cover $Q_s$ with a new generation of resonant blocks $\widetilde\Omega_{k}^{s+1}  (k\in P_{s+1})$.
It turns out that one can control $\|T_{\widetilde\Omega_{k}^{s+1}}^{-1}\|$ via  zeros $\pm \theta_{s+1}$ of some new functions which are also perturbations of quadratic polynomials  in the sense that
$$\det T_{\widetilde\Omega_{k}^{s+1}}\sim\delta_s^{-2} \|\theta+k\cdot \omega - \theta_{s+1}\|\cdot\|\theta+k\cdot \omega + \theta_{s+1}\|.$$
The key point is that some  $T^{-1}_{\widetilde\Omega_{k}^{s+1}}$  while $\widetilde\Omega_{k}^{s+1}$ intersecting $Q_s$ become controllable \footnote{Even more general sets , e.g., the $(s+1)$-{\bf good} sets  remain true.} in the $(s+1)$-th step. Moreover,  the completely uncontrollable singular sites form the $(s+1)$-th singular sites, i.e., 
$$Q_{s+1}=\bigcup_{\sigma=\pm 1}\left\{k \in P_{s+1} :\  \left\|\theta+k\cdot \omega +\sigma \theta_{s+1}\right\|<\delta_{s+1}\right\}.
$$


Now we turn to the statement of our main result on the multi-scale type Green's function estimates. Define the induction parameters as follows.  $$N_{s+1}=\left[|\log\frac{\gamma}{\delta_s}|^{\frac{1}{c^5\tau}}\right], \ |\log\frac{\gamma}{\delta_{s+1}}|=|\log\frac{\gamma}{\delta_s}|^{c^5}.$$
Thus $$N_s^{c^5}-1\leq N_{s+1}\leq( N_s+1)^{c^5}.$$


We first introduce the following statement.
	\begin{state}[${\bf\mathcal{P}}_s{(s\geq1)}$]\label{state}
	Let 
	\begin{align}
    \label{Q_{s-1}}
Q_{s-1}^{\pm}&=\left\{k \in P_{s-1} :\  \left\|\theta+k\cdot \omega \pm \theta_{s-1}\right\|<\delta_{s-1}\right\},
		\  Q_{s-1}=Q_{s-1}^{+}\cup Q_{s-1}^{-},\\
\label{wQ_{s-1}}\widetilde{Q}_{s-1}^{\pm}&=\left\{k \in P_{{s-1}} :\  \left\|\theta+k\cdot \omega \pm \theta_{s-1}\right\|<\delta_{s-1}^{\frac{1}{100}}\right\},\ \widetilde{Q}_{s-1}=\widetilde{Q}_{s-1}^{+}\cup\widetilde{Q}_{s-1}^{-}.	
	\end{align}
We distinguish the following two cases: 
	\begin{itemize}
		\item[{$({\bf C}1)_{s-1} $}.]
		\begin{equation}\label{scase1}
			{\rm dist}(\widetilde{Q}_{s-1}^{-}, Q_{s-1}^{+})>100N_{s}^{c},
		\end{equation}
		\item[{$({\bf C}2)_{s-1} $}.] 
		\begin{equation}\label{scase2}
			{\rm dist}(\widetilde{Q}_{s-1}^{-}, Q_{s-1}^{+})\leq100N_{s}^{c}.
		\end{equation}
	\end{itemize}
Let  $$\mathbb{Z}^d\ni l_{s-1}=\left\{\begin{aligned}
		&0 \ & {\rm if }\ \eqref{scase1}\ {\rm holds\  true},\\
		&i_{s-1}-j_{s-1} \ &{\rm if }\ \eqref{scase2}\ {\rm holds\  true},\end{aligned}
	\right.$$
	where $i_{s-1}\in Q_{s-1}^+$, $j_{s-1}\in \widetilde{Q}_{s-1}^-$ such that $\|i_{s-1}-j_{s-1}\|\leq 100N_s^{c}$ in {\rm$({\bf C}2)_{s-1} $}.
	Set $\Omega_k^0=\{k\}$ ($k\in \Z^d$). Let $\Lambda\subset\Z^d$ be a finite set. We say $\Lambda$ is $(s-1)$-{\bf good} iff
\begin{equation}\label{sgood}
		\left\{\begin{aligned}
			&k'\in Q_{s'},\widetilde{\Omega}_{k'}^{s'}\subset\Lambda,\widetilde{\Omega}_{k'}^{s'}\subset \Omega_{k}^{s'+1} \Rightarrow \widetilde{\Omega}_k^{s'+1} \subset\Lambda\ { \rm for }\ s'<s-1,\\
			&\{k\in P_{s-1}:\   \widetilde{\Omega}_{k}^{s-1}\subset \Lambda \}\cap Q_{s-1}= \emptyset.
		\end{aligned}\right.
\end{equation}
	Then  
	\begin{itemize}
		\item[$(\bf a)_s$.] There are $P_s\subset Q_{s-1}$ so that the following holds true. 
		We have in the case {\rm $({\bf C}1)_{s-1}$} that
		\begin{equation}\label{2233}
			P_s=Q_{s-1}\subset\left\{k\in \mathbb{Z}^d+\frac{1}{2}\sum_{i=0}^{s-1}l_i:\  \min_{\sigma=\pm1}\|\theta+k\cdot\omega+\sigma\theta_{s-1}\|<\delta_{s-1}\right\}.
		\end{equation}
		For the  case {\rm $({\bf C}2)_{s-1}$}, we have
		\begin{align}\label{3322}
			\begin{split}
				P_s&\subset\left\{k\in \mathbb{Z}^d+\frac{1}{2}\sum_{i=0}^{s-1}l_i :\ \|\theta+k\cdot\omega\|<3\delta_{s-1}^{\frac{1}{100}}\right\},\\
				{\rm or }\ 	P_s&\subset\left\{k\in \mathbb{Z}^d+\frac{1}{2}\sum_{i=0}^{s-1}l_i :\ \|\theta+k\cdot\omega+\frac{1}{2}\|<3\delta_{s-1}^{\frac{1}{100}}\right\}.
			\end{split}
		\end{align}
		For every $k \in P_s$, we can find resonant blocks  $\Omega_{k}^{s}$, $\widetilde{\Omega}_k^{s}\subset\Z^d$ with the following properties. If \eqref{scase1} holds true, then
		\begin{align*}
			&\Lambda_{N_s}(k)\subset \Omega_{k}^{s}\subset \Lambda_{N_s+50N_{s-1}^{c^2}}(k),\\
			&\Lambda_{N_s^c}(k)\subset \widetilde{\Omega}_k^{s}\subset \Lambda_{N_s^c+50N_{s-1}^{c^2}}(k),
		\end{align*} 
		and if \eqref{scase2} holds true, then
		\begin{align*}
	\Lambda_{100N_s^c}(k)&\subset \Omega_{k}^{s}\subset \Lambda_{100N_s^c+50N_{s-1}^{c^2}}(k),\\
	\Lambda_{N_s^{c^2}}(k)&\subset \widetilde{\Omega}_k^{s}\subset \Lambda_{N_s^{c^2}+50N_{s-1}^{c^2}}(k).
		\end{align*} 
		These resonant blocks are constructed  satisfying  the following two properties.
		\begin{itemize}
			\item [$({\bf a}1)_s$.] \begin{equation}\label{haha}
				\left\{\begin{aligned}
					&\Omega_{k}^{s} \cap \widetilde{\Omega}_{k'}^{s'} \neq \emptyset\ \left(s^{\prime}<s\right) \Rightarrow \widetilde{\Omega}_{k'}^{s'} \subset \Omega_{k}^{s}, \\
					&\widetilde{\Omega}_{k}^{s} \cap \widetilde{\Omega}_{k'}^{s'} \neq \emptyset\ \left(s^{\prime}<s\right) \Rightarrow \widetilde{\Omega}_{k'}^{s'} \subset \widetilde{\Omega}_{k}^{s}, \\
					&{\rm dist}(\widetilde{\Omega}_{k}^{s}, \widetilde{\Omega}_{k^{\prime}}^{s})>10\operatorname{diam}\widetilde{\Omega}_{k}^{s} \ { \rm for }\  k \neq k^{\prime} \in P_{s} .
				\end{aligned}\right.	
			\end{equation}
			\item [$({\bf a}2)_s$.] The translation of $\widetilde{\Omega}_k^{s},$  $$\widetilde{\Omega}_k^{s}-k\subset \mathbb{Z}^d+\frac{1}{2}\sum_{i=0}^{s-1}l_i$$ is independent  of $k\in P_s$ and  symmetrical about the origin.
		\end{itemize}
			\item[$({\bf b})_s$.]
		$Q_{s-1}$ is covered by $\Omega_k^s\  (k\in  P_s)$ in the sense that, for every $k' \in Q_{s-1}$, there exists $k \in P_{s}$ such that
		\begin{equation}\label{gai}
			\widetilde{\Omega}_{k'}^{s-1} \subset \Omega_k^{s}.
		\end{equation}
		\item[$({\bf c})_s$.]
		For each $k\in P_s$, $\widetilde{\Omega}_k^s$ contains a subset
		$A_k^s \subset \Omega_k^s$ with $\# A_k^s\leq2^s$
		such that $\widetilde{\Omega}_k^s\setminus A_k^s$ is $(s-1)$-{\rm\bf good}.
		Moreover,  $A_k^{s}-k$ is independent  of $k$ and is symmetrical about the origin.		
		\item[$({\bf d})_s$.]   There is $\theta_s=\theta_s(E)\in\C$ with the following properties. Replacing
		$\theta +n\cdot \omega$ by $z+(n-k)\cdot \omega,$ and restricting $z$	in  \begin{equation}\label{dede}
			\{z\in \mathbb{C}:\ \min_{\sigma=\pm1} \left\|z+\sigma\theta_s\right\|<\delta_s^{\frac{1}{10^4}}\},
		\end{equation}
		then $T_{\widetilde{\Omega}_{k}^{s}}$ becomes
		$$M_s(z)=T(z)_{\widetilde{\Omega}_k^{s}-k}=\left(\cos 2 \pi(z+n \cdot\omega)\delta_{n,n'}-E+\varepsilon \Delta\right)_{\widetilde{\Omega}_k^s-k}.$$
		Then $M_s(z)_{(\widetilde{\Omega}_k^s-k)\setminus(A_k^s-k)}$ is invertible  and we can define the Schur complement
		\begin{align*}	S_s(z)&=M_s(z)_{A_k^s-k}-R_{A_k^s-k} M_s(z) R_{(\widetilde{\Omega}_k^s-k)\setminus(A_k^s-k)}\left( M_s(z)_{(\widetilde{\Omega}_k^s-k)\setminus(A_k^s-k)}\right)^{-1}\\
			&\ \ \ \times R_{(\widetilde{\Omega}_k^s-k)\setminus(A_k^s-k)} M_s(z)R_{A_k^s-k}. 
		\end{align*}
Moreover, if $z$ belongs to the set defined by \eqref{dede},  then we have
		\begin{equation}\label{10}
			\max_{x}\sum_y|S_s(z)(x,y)|<4+\sum_{l=0}^{s-1}\delta_l<10,
		\end{equation}and
		\begin{equation}\label{sSsim}
			\operatorname{det}S_{s}(z)\stackrel{\delta_{s-1}}{\sim}\|z-\theta_{s}\|\cdot \|z+\theta_{s}\|.
		\end{equation}	
		
	\item[$({\bf e})_s$.] We say a finite set $\Lambda\subset \mathbb{Z}^d$ is $s$-{\rm\bf good} iff
\begin{equation}\label{sgood}
		\left\{\begin{aligned}
			&k'\in Q_{s'},\widetilde{\Omega}_{k'}^{s'}\subset\Lambda,\widetilde{\Omega}_{k'}^{s'}\subset \Omega_{k}^{s'+1} \Rightarrow \widetilde{\Omega}_k^{s'+1} \subset\Lambda\ { \rm for }\ s'<s,\\
			&\{k\in P_s:\   \widetilde{\Omega}_{k}^s\subset \Lambda \}\cap Q_{s}= \emptyset.
		\end{aligned}\right.
\end{equation}
		Assuming  $\Lambda$ is  $s$-{\rm\bf good}, then 
		\begin{align}
			\label{L2}\|T_{\Lambda}^{-1}\|&<\delta^{-3}_{s-1}\sup_{\left\{k\in P_s:\   \widetilde{\Omega}_{k}^s\subset \Lambda \right\}}\|\theta+k\cdot\omega-\theta_{s}\|^{-1}\cdot \|\theta+k\cdot\omega+\theta_{s}\|^{-1}\\
			\nonumber&<\delta^{-3}_s,\\
			\label{exp}\left|T_{\Lambda}^{-1}(x, y)\right|&<e^{-\gamma_s\|x-y\|_1} \ { \rm for }\ \|x-y\|>N_{s}^{c^3},
		\end{align}
		where $$\gamma_0=\frac{1}{2}|\log\varepsilon|,\ \gamma_s=\gamma_{s-1}(1-N_s^{\frac{1}{c}-1})^3.$$ Thus $\gamma_s\searrow\gamma_{\infty}\geq\frac{1}{2}\gamma_0=\frac{1}{4}|\log\varepsilon|$.		
		\item[$({\bf f})_s$.] We have
		\begin{equation}\label{P>}
			\left\{k \in \mathbb{Z}^{d}+\frac{1}{2}\sum_{i=0}^{s-1}l_i:\ \min_{\sigma=\pm1}\left\|\theta+k\cdot \omega+\sigma \theta_{s}\right\|<10\delta_{s}^{\frac{1}{100}}\right\}\subset P_s.
		\end{equation}		
		\end{itemize}
	
\end{state}
	
The main theorem of this section is 	
\begin{thm}\label{Inthm}
Let $\omega\in\mathcal{R}_{\tau,\gamma}$. Then there is some $\varepsilon_0(d,\tau,\gamma)>0$ so that for $0<\varepsilon\leq \varepsilon_0$, the statement ${\bf\mathcal{P}}_s$ holds for all $s\geq 1.$ 	
	\end{thm}

The following three subsections are devoted to prove Theorem \ref{Inthm}.

\subsection{The initial step}
Recalling \eqref{T}--\eqref{DN} and $\cos2\pi\theta_0=E$,  we have
\begin{align*}
	|D_{n}|&=2 |\sin \pi(\theta+n\cdot{\omega}+\theta_{0}) \sin \pi(\theta+n\cdot{\omega}-\theta_{0})|\\
	&\geq2\|\theta+n\cdot{\omega}+\theta_{0}\|\cdot \|\theta+n\cdot{\omega}-\theta_{0}\| .
\end{align*}
Denote  $\delta_{0}=\varepsilon^{1 / 10}$ and
$$
P_{0}=\mathbb{Z}^d, \  Q_{0}=\{k \in P_0:\  \min (\|\theta +k\cdot  \omega+\theta_{0}\|,\|\theta+ k \cdot  \omega-\theta_{0}\|)<\delta_0\} .
$$
We say a finite set $\Lambda\subset \mathbb{Z}^d$ is $0$-\textbf{good} iff	  $$\Lambda \cap Q_{0}=\emptyset.$$
\begin{lem}\label{le1}
If the finite set $\Lambda\subset\Z^d$ is $0$-{\bf good}, then 
\begin{align}
\label{0good2}\|T_\Lambda^{-1}\|&< 2\|D_\Lambda^{-1}\|<\delta_0^{-2},\\
\label{0good1}|T_\Lambda^{-1}(x,y)|&<e^{-\gamma_0\|x-y\|_1}\   {\rm for}\  \|x-y\|>0.
\end{align}
where $\gamma_0=5|\log \delta_0|=\frac{1}{2}|\log \varepsilon|$.
\end{lem}
\begin{proof}[Proof of Lemma \ref{le1}]
Assuming $\Lambda$ is $0$-\textbf{good}, we have
$$
\left\|D_{\Lambda}^{-1}\right\|<\frac{1}{2}\delta_0^{-2},\ 
\left\|\varepsilon D_{\Lambda}^{-1}\Delta_{\Lambda}\right\|<d\varepsilon \delta_0^{-2}<\frac{1}{2}\delta_{0}^{7}<\frac12.
$$
Thus
$$
T_{\Lambda}^{-1}=\left(I+\varepsilon D_{\Lambda}^{-1}\Delta_{\Lambda}\right)^{-1}D_{\Lambda}^{-1}
$$
and $\left(I+\varepsilon D_{\Lambda}^{-1}\Delta_{\Lambda}\right)^{-1}$ may be expanded in the Neumann series
$$
(I+\varepsilon D_{\Lambda}^{-1}\Delta_{\Lambda})^{-1}=\sum_{i=0}^{+\infty}(-\varepsilon D_{\Lambda}^{-1}\Delta_{\Lambda})^{i}.
$$
Hence
\begin{equation*}
	\|T_\Lambda^{-1}\|< 2\|D_\Lambda^{-1}\|<\delta_0^{-2},
\end{equation*}
which implies \eqref{0good2}. 

In addition,  if $\|x-y\|_1>i,$ then
$$\left( (\varepsilon D_{\Lambda}^{-1}\Delta_{\Lambda})^{i}D_{\Lambda}^{-1}\right) (x,y)=0.$$
Hence
\begin{align*}
	|T_\Lambda^{-1}(x,y)|&=|\sum_{i\geq\|x-y\|_1}\left( (\varepsilon D_{\Lambda}^{-1}\Delta_{\Lambda})^{i}D_{\Lambda}^{-1}\right)(x,y)|<\delta_0^{7\|x-y\|_1-2}.
\end{align*}
In particular,
\begin{equation*}
	|T_\Lambda^{-1}(x,y)|<e^{-\gamma_0\|x-y\|_1}\  {\rm for}\  \|x-y\|>0
\end{equation*}
with $\gamma_0=5|\log \delta_0|=\frac{1}{2}|\log \varepsilon|$, which yields \eqref{0good1}.
\end{proof}

\subsection{Verification of ${\bf \mathcal{P}}_1$}\label{p1vf}

If $\Lambda\cap Q_0\neq \emptyset$,  then the Neumann series argument of previous subsection does not work. Thus we use the resolvent identity argument to  estimate $T^{-1}_\Lambda$, where $\Lambda$ is $1$-\textbf{good} ($1$-\textbf{good} will be specified later) but might intersect with $Q_0$ (not $0$-\textbf{good}).  

First, we construct blocks  $\Omega^1_k\ (k\in P_1)$ to cover the singular point $Q_0$. Second, we get the bound  estimate   
 $$	\|T_{\widetilde{\Omega}_k^1}^{-1}\|<\delta_{0}^{-2}\|\theta+k\cdot\omega-\theta_1\|^{-1}\cdot\|\theta+k\cdot\omega+\theta_1\|^{-1},$$ 
  where $\widetilde{\Omega}^1_k$ is an   extension of $\Omega^1_k$, and $\theta_{1}$ is obtained by analyzing the root of the equation $\operatorname{det}T(z-k\cdot \omega)_{\widetilde{\Omega}^1_k}=0$ about  $z$.
  Finally, we combine the estimates of $T^{-1}_{\widetilde{\Omega}^1_k}$  to get that of  $T^{-1}_{\Lambda}$ by resolvent identity assuming $\Lambda$ is $1$-\textbf{good}.

Recall  that
$$1<c^{20}<\frac{1}{\tau}.$$
Let
$$N_1=\left[|\log\frac{\gamma}{\delta_0}|^{\frac{1}{c^5\tau}}\right].$$
Define
\begin{align*}
&Q_{0}^{\pm}=\left\{k \in \mathbb{Z}^d :\ \left\|\theta+k\cdot \omega \pm \theta_{0}\right\|<\delta_{0}\right\},\  Q_{0}=Q_{0}^{+}\cup Q_{0}^{-},\\
&\widetilde{Q}_{0}^{\pm}=\left\{k \in \mathbb{Z}^d :\ \left\|\theta+k\cdot \omega \pm \theta_{0}\right\|<\delta_{0}^{\frac{1}{100}}\right\}, \  \widetilde{Q}_{0}=\widetilde{Q}_{0}^{+}\cup\widetilde{Q}_{0}^{-}.
\end{align*}

We distinguish  three steps. 

{\bf STEP1: The case $({\bf C}1)_0$ Occurs }: i.e., 

\begin{equation}\label{case1}
	{\rm dist}\left(\widetilde{Q}_{0}^{-}, Q_{0}^{+}\right)>100 N_{1}^{c}.
\end{equation}
\begin{rem}\label{r0}
	We have in fact   $${\rm dist}\left(\widetilde{Q}_{0}^{-}, Q_{0}^{+}\right)={\rm dist}\left(\widetilde{Q}_{0}^{+}, Q_{0}^{-}\right).$$Thus \eqref{case1}  also implies
	$${\rm dist}\left(\widetilde{Q}_{0}^{+}, Q_{0}^{-}\right)>100 N_{1}^{c}.$$
We refer to the Appendix \ref{appa} for a detailed proof. 
\end{rem}
Assuming \eqref{case1},  we define
\begin{equation}\label{P11}
	P_1=Q_0= \{k \in \mathbb{Z}^d :\  \min (\|\theta +k\cdot  \omega+\theta_{0}\|,\|\theta+ k \cdot  \omega-\theta_{0}\|)<\delta_0\} .
\end{equation}
Associate every $k\in P_1$ an $N_{1}$-block $\Omega_k^1:=\Lambda_{N_1}(k)$ and an $N_{1}^c$-block $\widetilde{\Omega}_k^1:=\Lambda_{N_1^c}(k)$.
Then $\widetilde{\Omega}_k^1-k\subset\mathbb{Z}^d$ is independent of $k\in P_1 $ and symmetrical about the origin.
If $k\neq k' \in P_1$,
$$\|k-k'\|\geq\min \left(100N_1^c,|\log\frac{\gamma}{2\delta_0}|^{\frac{1}{\tau}}\right)\geq100N_1^c.$$
Thus
$$	{\rm dist}\left( \widetilde{\Omega}_{k}^{1},  \widetilde{\Omega}_{k'}^{1}\right)>10\operatorname{diam}\widetilde{\Omega}_{k}^{1} \  \text {for } k \neq k^{\prime} \in P_{1} .$$
For $k\in Q_0^-$, we  consider 
\begin{align*}
	M_1(z):=T(z)_{\widetilde{\Omega}_k^1-k}=\left(\cos 2 \pi(z+n\cdot  \omega)\delta_{n,n'}-E+\varepsilon \Delta\right)_{n\in \widetilde{\Omega}_k^1-k}
\end{align*}
defined in
\begin{equation}\label{d010-}
	\left\{z \in \mathbb{C}:\  \left| z-\theta_{0} \right|< \delta_{0}^{\frac{1}{10}}\right\}.
\end{equation}
For $n\in (\widetilde{\Omega}_k^1-k)\setminus\{0\}$, we have for $0<\delta_0\ll1,$\begin{align*}
	\|z+n\cdot \omega-\theta_{0}\|&\geq \|n\cdot \omega\|-|z-\theta_0|\\
	&\geq \gamma e^{-N_1^{c\tau}}-\delta_0^{\frac{1}{10}}\\
	&\geq \gamma e^{-|\log\frac{\gamma}{\delta_{0}}|^{\frac{1}{c^4}}}-\delta_0^{\frac{1}{10}}\\
	&>\delta_0^{\frac{1}{10^4}}.
\end{align*}
For $n\in \widetilde{\Omega}_k^1-k $, we have 
\begin{align*}
	\|z+n\cdot \omega+\theta_{0}\|&\geq \|\theta+(n+k)\cdot \omega+\theta_0\|-|z-\theta_0|-\|\theta+k\cdot \omega-\theta_{0}\|\\
	&\geq \delta_{0}^{\frac{1}{100}}-\delta_0^{\frac{1}{10}}-\delta_0\\
	&>\frac{1}{2}\delta_0^{\frac{1}{100}}.
\end{align*}
Since $\delta_{0}\gg \varepsilon$,  we have by Neumann series argument 
$$ \left\|\left(M_1(z)_{(\widetilde{\Omega}_k^1-k)\setminus\{0\}}\right) ^{-1}\right\|<3\delta_0^{-\frac{1}{50}}.$$
 Now we can apply the Schur complement lemma (cf. Lemma \ref{Su} in the appendix) to provide desired estimates.  By  Lemma \ref{Su}, $M_1(z)^{-1}$ is controlled by the inverse of the Schur complement  (of   $(\widetilde{\Omega}_k^1-k)\setminus\{0\}$)
\begin{align*}
	S_1(z)&=M_1(z)_{\{0\}}-R_{\left\{0\right\}} M_1(z) R_{(\widetilde{\Omega}_k^1-k)\setminus\{0\}}\left( M_1(z)_{(\widetilde{\Omega}_k^1-k)\setminus\{0\}}\right)^{-1} R_{(\widetilde{\Omega}_k^1-k)\setminus\{0\}} M_1(z)R_{\{0\}}\\
	&=-2 \sin \pi( z-\theta_{0})\sin \pi(z+\theta_{0})+r(z)\\
	&=g(z)((z-\theta_0)+r_1(z)),
\end{align*}
where $g(z)$ and $r_1(z)$ are analytic functions in  the set defined by  \eqref{d010-}  satisfying $|g(z)|\geq2\|z+\theta_0\|>\delta_{0}^{\frac{1}{100}}$ and   $|r_1(z)|<\varepsilon^{2}\delta_0^{-1}<\varepsilon$.
Since
$$|r_1(z)|<|z-\theta_0|\ {\rm for}\ |z-\theta_0|= \delta_0^\frac{1}{10},$$
using  R\'ouche theorem implies the equation
$$(z-\theta_0)+r_1(z)=0$$
has  a unique root $\theta_1$ in  the set of  \eqref{d010-} satisfying  \[|\theta_0-\theta_1|=|r_1(\theta_1)|<\varepsilon,\  |(z-\theta_0)+r_1(z)|\sim|z-\theta_1|.\]
Moreover, $\theta_1$ is the unique root of  $\operatorname{det}M_1(z)=0$ in  the set \eqref{d010-}.
Since $\|z+\theta_0\|>\frac{1}{2}\delta_{0}^{\frac{1}{100}}$ and $|\theta_0-\theta_1|<\varepsilon,$ we get   $$\|z+\theta_1\|\sim \|z+\theta_0\|,$$	
which shows for $z$ being in the set  of \eqref{d010-},
\begin{align}
\label{S1-}|S_1(z)|&\sim\|z+\theta_1\|\cdot \|z-\theta_1\|,\\
\nonumber\|M_1(z)^{-1}\|&<4\left( 1+\|( M_1(z)_{(\widetilde{\Omega}_k^1-k)\setminus\{0\}})^{-1}\|\right)^2 \left( 1+|S_1(z)|^{-1}\right)\\
\label{M1-}&<\delta_0^{-2}\|z+\theta_1\|^{-1}\cdot \|z-\theta_1\|^{-1}.
\end{align}
where in the first  inequality we use Lemma \ref{Su}.
Now, for $k\in Q_0^+$, we consider $M_1(z)$ in
\begin{equation}\label{d010+}
	\left\{z \in \mathbb{C}:\ \left| z+\theta_{0} \right|< \delta_{0}^{\frac{1}{10}}\right\}.
\end{equation}
The similar argument shows that  $\operatorname{det}M_1(z)=0$ has a unique root $\theta_1'$ in the set of \eqref{d010+}.
We will show $\theta_1+\theta_1'=0.$ In fact, by Lemma \ref{even}, $\operatorname{det}M_1(z)$ is an even function of $z$.  Then the uniqueness of  the root  implies  $\theta_1'=-\theta_1$.
Thus for $z$ being in the set  of $\eqref{d010+}$,   both \eqref{S1-} and \eqref{M1-} hold true as well.
Finally, since $M_1(z)$ is $1$-periodic, \eqref{S1-} and \eqref{M1-} remain valid for 
\begin{equation}\label{D11}
	z\in\{z \in \mathbb{C}:\ \min_{\sigma=\pm1} \left\| z+\sigma\theta_{0} \right\|< \delta_{0}^{\frac{1}{10}}\}.
\end{equation}
From $\eqref{P11}$,  we have $\theta+k\cdot \omega$ belongs to the set of  \eqref{D11}. Thus  for $k\in P_1$, we get
\begin{align}
	\nonumber \|T_{\widetilde{\Omega}_k^1}^{-1}\|&=\|M_1(\theta+k\cdot\omega)^{-1}\|\\
	\label{111e}&<\delta_{0}^{-2}\|\theta+k\cdot\omega-\theta_1\|^{-1}\cdot\|\theta+k\cdot\omega+\theta_1\|^{-1}.
\end{align}


{\bf STEP2: The case $({\bf C}2)_0$ Occurs: }  i.e., 
\[{\rm dist}\left( \widetilde{Q}_{0}^{-},Q_{0}^{+}\right) \leq 100 N_{1}^c.\]
Then there exist  $ i_0 \in Q_{0}^{+}$ and $ j_0 \in \widetilde{Q}_{0}^{-}$ with $\|i_0-j_0\| \leq 100 N_{1}^c$, such that
\[\left\|\theta+i_0\cdot  \omega+\theta_{0}\right\|<\delta_{0}, \ \left\|\theta+j_0 \cdot \omega-\theta_{0}\right\|<\delta_{0}^{\frac{1}{100}}.\]
Denote \[l_0=i_0-j_0.\]
Then $\|l_0\|={\rm dist}(Q_{0}^{+}, \widetilde{Q}_{0}^{-})={\rm dist}(\widetilde{Q}_{0}^{+}, Q_{0}^{-}).$
Define\[O_{1}=Q_0^-\cup( Q_0^+-l_0).\]
For $k\in Q_0^+$, we have 
\begin{align*}
\|\theta+(k-l_0)\cdot \omega-\theta_0\|&<\|\theta+k\cdot \omega+\theta_0\|+\|l_0\cdot \omega+2\theta_0\|\\
&<\delta_0+\delta_{0}+\delta_{0}^\frac{1}{100}<2\delta_{0}^\frac{1}{100}. 
\end{align*}
Thus  $$O_1\subset\left\{o\in \mathbb{Z}^d:\  \|\theta+ o\cdot \omega-\theta_0\|<2\delta_{0}^{\frac{1}{100}}\right\}.$$
For every $o\in O_1$, define its mirror point
\[o^{*}=o+l_0.\]
Next define
\begin{equation}\label{P12}
	P_{1}=\left\{\frac{1}{2}\left(o+o^{*}\right) :\  o \in O_{1}\right\}=\left\{o+\frac{l_0}{2} :\  o \in O_{1}\right\}.
\end{equation}
Associate every $k\in P_1$ with a  $100N_{1}^c$-block $\Omega_k^1:=\Lambda_{100N_{1}^c}(k)$ and a  $N_{1}^{c^2}$-block $\widetilde{\Omega}_k^1:=\Lambda_{N_{1}^{c^2}}(k)$.
Thus
\[Q_{0} \subset \bigcup_{k \in P_{1}} \Omega_{k}^{1}\]
and $\widetilde{\Omega}_k^1-k\subset\mathbb{Z}^d+\frac{l_0}{2}$ is independent of $k\in P_1 $ and symmetrical about the origin.
Notice that
\begin{align*}
	&\min \left(\|\frac{l_0}{2}\cdot  \omega+\theta_{0}\|,\|\frac{l_0}{2}\cdot  \omega+\theta_{0}-\frac{1}{2}\|\right)  =
	\frac{1}{2}\left\|l_0 \cdot \omega+2 \theta_{0}\right\|\\
	&\ \ \ \ \leq \frac{1}{2}\left(\left\|\theta+i_0 \cdot \omega+\theta_{0}\right\|+\left\|\theta+j_0 \cdot \omega-\theta_{0}\right\| \right)< \delta_{0}^{\frac{1}{100}}.\end{align*}
Since $\delta_0\ll1$,   only one of  $$\|\frac{l_0}{2}\cdot  \omega+\theta_{0}\|< \delta_{0}^{\frac{1}{100}},\quad \|\frac{l_0}{2}\cdot  \omega+\theta_{0}-\frac{1}{2}\|< \delta_{0}^{\frac{1}{100}}$$ holds true.	
First,  we consider the case
\begin{equation}\label{sem1}
	\|\frac{l_0}{2} \cdot \omega+\theta_{0}\|< \delta_{0}^{\frac{1}{100}}.
\end{equation}
Let $k \in P_1$. Since $k=\frac{1}{2}\left(o+o^{*}\right)=(o+\frac{l_0}{2})$ (for some $o\in O_1$), we have
\begin{equation}\label{P12g}
	\|\theta+k \cdot \omega\| \leq\|\theta+o\cdot  \omega-\theta_{0}\|+	\|\frac{l_0}{2} \cdot \omega+\theta_{0}\|<3\delta_{0}^{\frac{1}{100}}.
\end{equation}
Thus if $k\neq k'\in P_1$,	 we obtain 
$$\|k-k'\|\geq \left|\log\frac{\gamma}{6\delta_0^\frac{1}{100}}\right|^{\frac{1}{\tau}}\sim N_1^{c^5}\gg 10N_1^{c^2},$$
which implies
$${\rm dist}( \widetilde{\Omega}_{k}^{1},  \widetilde{\Omega}_{k'}^{1})>10\operatorname{diam}\widetilde{\Omega}_{k}^{1} \  \text {for } k \neq k^{\prime} \in P_{1} .$$
Consider 
\begin{align*}
	M_1(z):=T(z)_{\widetilde{\Omega}_k^1-k}=\left((\cos 2 \pi(z+n\cdot  \omega)\delta_{n,n'}-E+\varepsilon \Delta\right)_{n\in \widetilde{\Omega}_k^1-k}
\end{align*}
in \begin{equation}\label{d01000}
	\left\{z \in \mathbb{C}:\ \left| z\right|< \delta_{0}^{\frac{1}{10^3}}\right\}.
\end{equation}
For $n\neq \pm \frac{l_0}{2}$ and  $n\in \widetilde{\Omega}_k^1-k$, we have
\begin{align*}
\left\|n \cdot \omega \pm \theta_{0}\right\| &\geq\|\left(n \mp \frac{l_0}{2}\right)\cdot  \omega\|-\|\frac{l_0}{2} \omega+\theta_{0}\|\\
&>\gamma e^{-(2N_1^{c^2})^{\tau}}-\delta_{0}^{\frac{1}{100}}>2\delta_{0}^{\frac{1}{10^{4}}}.
\end{align*}
Thus for $z$ being in the set of  \eqref{d01000} and  $n\neq \pm \frac{l_0}{2}$, we have
\[\left\|z+n \cdot \omega \pm \theta_{0}\right\| \geq\left\|n \cdot \omega \pm \theta_{0}\right\| -|z|>\delta_{0}^{\frac{1}{10^{4}}}.\]
Hence
\[|\cos 2 \pi(z+n \cdot \omega)-E| \geq \delta_{0}^{2 \times \frac{1}{10^{4}}}\gg\varepsilon.\]
Using  Neumann series argument concludes 
\begin{equation}\label{ca}
	\left\|\left(M_1(z)_{(\widetilde{\Omega}_k^1-k)\backslash\left\{\pm \frac{l_0}{2}\right\}}\right)^{-1} \right\|< \delta_{0}^{-3 \times \frac{1}{10^{4}}}.
\end{equation}
Thus by Lemma \ref{Su},   $M_1(z)^{-1}$ is controlled by the inverse of the Schur complement of    $(\widetilde{\Omega}_k^1-k)\setminus\{\pm\frac{l_0}{2}\}$, i.e., 
\begin{align*}
	S_1(z)&=M_1(z)_{\left\{\pm \frac{l_0}{2}\right\}}-R_{\left\{\pm \frac{l_0}{2}\right\}} M_1(z) R_{(\widetilde{\Omega}_k^1-k)\backslash\left\{\pm\frac{l_0}{2}\right\}}\\
	&\ \ \ \ \times \left(M_1(z)_{(\widetilde{\Omega}_k^1-k) \backslash\left\{\pm \frac{l_0}{2}\right\}}\right)^{-1}R_{(\widetilde{\Omega}_k^1-k)\backslash\left\{\pm\frac{l_0}{2}\right\}} M_1(z)R_{\left\{\pm \frac{l_0}{2}\right\}}.
\end{align*}
Clearly,
\begin{align*}
	\operatorname{det} S_1(z)&=\operatorname{det}\left( M_1(z)_{\left\{\pm \frac{l_0}{2}\right\}}\right)+O(\varepsilon^2\delta_{0}^{-\frac{3}{10^4}}) \\
	&=4 \sin \pi( z+\frac{l_0}{2}\cdot  \omega-\theta_{0}) \sin \pi(z+\frac{l_0}{2} \cdot \omega+\theta_{0}) \\
	&\ \ \ \ \times \sin \pi(z-\frac{l_0}{2}\cdot  \omega-\theta_{0})  \sin \pi(z-\frac{l_0}{2}\cdot \omega+\theta_{0})+O(\varepsilon^{1.5}).
\end{align*}
If $l_0=0$, then  $$	\operatorname{det} S_1(z)=-2\sin \pi(z-\theta_{0})  \sin \pi(z+\theta_{0})+O(\varepsilon^{1.5}).$$ In this case,  the  argument is easier,  and  we omit the discussion. In the following,  we deal with $l_0\neq 0.$
By \eqref{sem1} and \eqref{d01000}, we have
\begin{align*}
	\|z+\frac{l_0}{2} \cdot \omega-\theta_{0}\|&\geq\|l_0\cdot \omega\|-\|\frac{l_0}{2}\cdot  \omega+\theta_{0}\|-|z|\\
	&>\gamma e^{-(100N_1^c)^\tau}-\delta_0^\frac{1}{100}-\delta_0^\frac{1}{10^3}\\
	&>\delta_{0}^{ \frac{1}{10^{4}}},
\end{align*}
\begin{align*}
	\|z-\frac{l_0}{2} \cdot \omega+\theta_{0}\|&\geq\|l_0\cdot \omega\|-\|\frac{l_0}{2}\cdot  \omega+\theta_{0}\|-|z|\\
	&>\gamma e^{-(100N_1^c)^\tau}-\delta_0^\frac{1}{100}-\delta_0^\frac{1}{10^3}\\
	&>\delta_{0}^{ \frac{1}{10^{4}}}.
\end{align*}
Let $z_1$  satisfy 
\[z_{1} \equiv \frac{l_0}{2} \cdot \omega+\theta_{0}\ (\bmod \ \mathbb{Z}), \ \left|z_{1}\right|=\| \frac{l_0}{2}\cdot  \omega+\theta_{0}\| < \delta_{0}^{\frac{1}{100}}.\]
Then \begin{align*}
	\operatorname{det} S_1(z&)\sim \|z+\frac{l_0}{2} \cdot \omega-\theta_{0}\|\cdot \|z-\frac{l_0}{2}\cdot \omega+\theta_{0}\|\cdot |(z-z_1)(z+z_1)+r_1(z)|\\
	&\stackrel{\delta_{0}^{\frac{2}{10^4}}}{\sim}|\left(z-z_{1}\right)\left(z+z_{1}\right)+r_1(z)|,
\end{align*}
where $r_1(z)$ is an analytic function  in  the set of \eqref{d01000} with $|r_1(z)|<\varepsilon\ll\delta_{0}^{\frac{1}{10^3}}$. Applying R\'ouche theorem shows the equation
$$\left(z-z_{1}\right)\left(z+z_{1}\right)+r_1(z)=0$$
has exact two roots $\theta_{1}$, $\theta_{1}'$ in the set of \eqref{d01000}, which are perturbations of $\pm z_1$.
Notice  that \[\left\{|z|< \delta_{0}^{\frac{1}{10^3}}:\ \operatorname{det} M_1(z)=0\right\}=\left\{|z|< \delta_{0}^{\frac{1}{10^3}} :\  \operatorname{det}S_1(z)=0\right\} \]
and $\operatorname{det}M_1(z)$ is an even function (cf. Lemma \ref{even}) of $z$. Thus $$\theta_{1}'=-\theta_{1}.$$
Moreover, we have \[|\theta_{1}-z_1|\leq|r_1(\theta_1)|^{\frac{1}{2}}<\varepsilon^{\frac{1}{2}},\  |\left(z-z_{1}\right)\left(z+z_{1}\right)+r_1(z)|\sim |\left(z-\theta_{1}\right)\left(z+\theta_{1}\right)|. \]
Thus for $z$ being in the set of   \eqref{d01000}, we have
\begin{equation}\label{ye}
	\operatorname{det}	S_1(z)\stackrel{\delta_{0}}{\sim}\|z-\theta_{1}\|\cdot \|z+\theta_{1}\|,
\end{equation}
which concludes  $$\|S_1(z)^{-1}\|\leq C\delta_{0}^{-1}\|z-\theta_1\|^{-1}\cdot \|z+\theta_1\|^{-1}.$$
Recalling \eqref{ca}, we get since Lemma \ref{Su}
\begin{align}
	\nonumber \|M_1(z)^{-1}\|&<4\left( 1+\|( M_1(z)_{(\widetilde{\Omega}_k^1-k)\setminus\{0\}})^{-1}\|\right)^2 \left( 1+\|S_1(z)^{-1}\|\right)\\
	\label{ya}&<\delta_0^{-2}\|z+\theta_1\|^{-1}\cdot \|z-\theta_1\|^{-1}
\end{align}
Thus for the case  \eqref{sem1},  both \eqref{ye} and \eqref{ya} are established  for $z$ belonging to
$$	\left\{z \in \mathbb{C}:\  \left\| z \right\|< \delta_{0}^{\frac{1}{10^3}}\right\}$$
since  $M_1(z)$ is $1$-periodic (in $z$).
By \eqref{P12g}, for $k\in P_1$, we also have 
\begin{align}
	\nonumber\|T_{\widetilde{\Omega}_k^1}^{-1}\|&=\|M_1(\theta+k\cdot\omega)^{-1}\|\\
	\label{1211e}&<\delta_{0}^{-2}\|\theta+k\cdot\omega-\theta_1\|^{-1}\cdot\|\theta+k\cdot\omega+\theta_1\|^{-1}.
\end{align}

For the case
\begin{equation}\label{sem2}
	\|\frac{l_0}{2} \cdot \omega+\theta_{0}-\frac{1}{2}\|<\delta_{0}^{\frac{1}{100}},
\end{equation}
we have for  $k\in P_1$,
\begin{equation}\label{P12g2}
	\|\theta+k \cdot \omega-\frac{1}{2}\| <3\delta_{0}^{\frac{1}{100}}.
\end{equation}
Consider 
\begin{align*}
	M_1(z):=T(z)_{\widetilde{\Omega}_k^1-k}=\left((\cos 2 \pi(z+n\cdot  \omega)\delta_{n,n'}-E+\varepsilon \Delta\right)_{n\in \widetilde{\Omega}_k^1-k}
\end{align*}
in \begin{equation}\label{d010002}
	\left\{z \in \mathbb{C}:\  | z-\frac{1}{2}|< \delta_{0}^{\frac{1}{10^3}}\right\}.
\end{equation}
By the similar argument as above, we get
\[\left\|\left(M_1(z)_{(\widetilde{\Omega}_k^1-k)\backslash\left\{\pm \frac{l_0}{2}\right\}}\right)^{-1} \right\|< \delta_{0}^{-3 \times \frac{1}{10^{4}}}.\]
Thus  $M_1(z)^{-1}$ is controlled by the inverse of the Schur complement of    $(\widetilde{\Omega}_k^1-k)\setminus\{\pm\frac{l_0}{2}\}$: 
\begin{align*}
	S_1(z)&=M_1(z)_{\left\{\pm \frac{l_0}{2}\right\}}-R_{\left\{\pm \frac{l_0}{2}\right\}} M_1(z) R_{(\widetilde{\Omega}_k^1-k)\backslash\left\{\pm\frac{l_0}{2}\right\}}\\
	&\ \ \ \ \times\left(M_1(z)_{(\widetilde{\Omega}_k^1-k) \backslash\left\{\pm \frac{l_0}{2}\right\}}\right)^{-1}R_{(\widetilde{\Omega}_k^1-k)\backslash\left\{\pm\frac{l_0}{2}\right\}} M_1(z)R_{\left\{\pm \frac{l_0}{2}\right\}}.
\end{align*}
Direct computation shows
\begin{align*}
	\operatorname{det} S_1(z)&=\operatorname{det}\left( M_1(z)_{\left\{\pm \frac{l_0}{2}\right\}}\right)+O(\varepsilon^2\delta_{0}^{-\frac{3}{10^4}}) \\
	&=4 \sin \pi( z+\frac{l_0}{2}\cdot  \omega-\theta_{0}) \sin \pi(z+\frac{l_0}{2} \cdot \omega+\theta_{0}) \\
	&\ \ \ \ \times\sin \pi(z-\frac{l_0}{2}\cdot  \omega-\theta_{0})  \sin \pi(z-\frac{l_0}{2}\cdot \omega+\theta_{0})+O(\varepsilon^{1.5}).
\end{align*}
By \eqref{sem2}  and \eqref{d010002}, we have
\begin{align*}
	\|z+\frac{l_0}{2} \cdot \omega-\theta_{0}\|&\geq\|l_0\cdot \omega\|-\|\frac{l_0}{2}\cdot  \omega+\theta_{0}-\frac{1}{2}\|-|z-\frac{1}{2}|\\
	&>\gamma e^{-(100N_1^c)^\tau}-\delta_0^\frac{1}{100}-\delta_0^\frac{1}{10^3}\\
	&>\delta_{0}^{ \frac{1}{10^{4}}},
\end{align*}
\begin{align*}
	\|z-\frac{l_0}{2} \cdot \omega+\theta_{0}\|&\geq\|l_0\cdot \omega\|-\|\frac{l_0}{2}\cdot  \omega+\theta_{0}-\frac{1}{2}\|-|z-\frac{1}{2}|\\
	&>\gamma e^{-(100N_1^c)^\tau}-\delta_0^\frac{1}{100}-\delta_0^\frac{1}{10^3}\\
	&>\delta_{0}^{ \frac{1}{10^{4}}}.
\end{align*}
Let $z_1$  satisfy 
\[z_1 \equiv \frac{l_0}{2} \cdot \omega+\theta_{0}\ (\bmod\  \mathbb{Z}), \ |z_1-\frac{1}{2}|=\| \frac{l_0}{2}\cdot  \omega+\theta_{0}-\frac{1}{2}\| < \delta_{0}^{\frac{1}{100}}.\]
Then \begin{align*}
	\operatorname{det} S_1(z&)\sim \|z+\frac{l_0}{2} \cdot \omega-\theta_{0}\|\cdot \|z-\frac{l_0}{2}\cdot \omega+\theta_{0}\|\cdot |(z-z_1)(z-(1-z_1))+r_1(z)|\\
	&\stackrel{\delta_{0}^{\frac{2}{10^4}}}{\sim}|\left(z-z_1\right)\left(z-(1-z_1)\right)+r_1(z)|,
\end{align*}
where $r_1(z)$ is an analytic function in the set of \eqref{d010002} with $|r_1(z)|<\varepsilon\ll\delta_{0}^{\frac{1}{10^3}}$.
Using again R\'ouche theorem shows the equation
$$\left(z-z_1\right)\left(z-(1-z_{1
})\right)+r_1(z)=0$$
has exact two roots $\theta_{1}$, $\theta_{1}'$ in \eqref{d010002}, which are perturbations of $z_1$ and $1-z_1$.
Notice  that \[\left\{|z-\frac{1}{2}|< \delta_{0}^{\frac{1}{10^3}}:\ \operatorname{det} M_1(z)=0\right\}=\left\{|z-\frac{1}{2}|< \delta_{0}^{\frac{1}{10^3}} :\  \operatorname{det}S_1(z)=0\right\} \]
and $\operatorname{det}M_1(z)$ is a $1$-periodic even function of $z$ (cf. Lemma \ref{even}). Thus $$\theta_{1}'=1-\theta_{1}.$$
Moreover,   \[|\theta_{1}-z_1|\leq|r_1(\theta_1)|^{\frac{1}{2}}<\varepsilon^{\frac{1}{2}},\  |\left(z-z_{1}\right)\left(z-1+z_{1}\right)+r_1(z)|\sim |\left(z-\theta_{1}\right)\left(z-(1-\theta_1)\right)|.\]
Thus for $z$ belonging to  the set of \eqref{d010002}, we have
$$\operatorname{det}S_1(z)\stackrel{\delta_{0}}{\sim}\|z-\theta_{1}\|\cdot \|z-(1-\theta_{1})\|=\|z-\theta_{1}\|\cdot \|z+\theta_{1}\|$$
and
$$	\|M_1(z)^{-1}\|<\delta_0^{-2}\|z-\theta_1\|^{-1}\cdot \|z+\theta_1\|^{-1}.$$
Thus for the case  \eqref{sem2}, both \eqref{ye} and \eqref{ya} hold  for $z$ being in 
$$\{z \in \mathbb{C}:\  \| z-\frac{1}{2} \|< \delta_{0}^{\frac{1}{10^3}}\}.$$
By \eqref{P12g2}, for $k\in P_1$, we obtain
\begin{align}
	\nonumber\|T_{\widetilde{\Omega}_k^1}^{-1}\|&=\|M_1(\theta+k\cdot\omega)^{-1}\|\\
	\label{1221e}&<\delta_{0}^{-2}\|\theta+k\cdot\omega-\theta_1\|^{-1}\cdot\|\theta+k\cdot\omega+\theta_1\|^{-1}.
\end{align}

For $k\in P_1$, we define $A_k^1 \subset \Omega_k^1$ to be
\begin{equation}\label{Ak1}
	A_k^1:=\left\{\begin{aligned}
		&\{k\}\ &\text{  case $({\bf C}1)_0$  }\\
		&\{o\}\cup \{o^*\} \  &\text{  case $({\bf C}2)_0$}
	\end{aligned},\right.
\end{equation}
where $k=\frac{1}{2}(o+o^*)$ for some $o\in O_1$   (cf. \eqref{P12})  in  the case $({\bf C}2)_0$.
We have verified $({\bf a})_1$--$({\bf d})_1$ and $({\bf f})_1.$

{\bf STEP3: Application of resolvent identity}

Now we verify $({\bf e})_1$ which is based on iterating resolvent identity. 

Note that  
$$\left|\log \frac{\gamma}{\delta_{1}}\right|=\left|\log \frac{\gamma}{\delta_{0}}\right|^{c^5}.$$ 
Recall that  $$Q_{1}^{\pm}=\left\{k \in P_1 :\ \left\|\theta+k\cdot \omega \pm \theta_{1}\right\|<\delta_{1}\right\},\ 
Q_{1}=Q_{1}^{+}\cup Q_{1}^{-}.$$
We say a finite set $\Lambda\subset \mathbb{Z}^d$ is $1$-{\bf good} iff
\begin{equation}\label{1good}
	\left\{
	\begin{aligned}
		\Lambda\cap Q_0\cap\Omega_{k}^1\neq\emptyset \Rightarrow \widetilde{\Omega}_{k}^1\subset \Lambda, \\
		\{k\in P_1 :\   \widetilde{\Omega}_{k}^1\subset \Lambda \}\cap Q_{1}= \emptyset.\end{aligned}\right.
\end{equation}
\begin{thm}\label{th11}
	If $\Lambda$ is $1$-\textbf{good}, then
	\begin{align}
		\label{th1L2}\|T_{\Lambda}^{-1}\|&<\delta^{-3}_0\sup_{\left\{k\in P_1 :\   \widetilde{\Omega}_{k}^1\subset \Lambda \right\}}\|\theta+k\cdot\omega-\theta_{1}\|^{-1}\cdot \|\theta+k\cdot\omega+\theta_{1}\|^{-1},\\
	\label{th1exp}
		\left|T_{\Lambda}^{-1}(x, y)\right|&<e^{-\gamma_1\|x-y\|_1} \  {\rm for }\ \|x-y\|>N_{1}^{c^3}.
	\end{align}
	where $\gamma_1=\gamma_0(1-N_1^{\frac{1}{c}-1})^3.$
\end{thm}	
\begin{proof}[Proof of Theorem \ref{th11}]

Denote   $$2\Omega_{k}^1:=\Lambda_{\operatorname{diam}\Omega_{k}^1}(k).$$ 
We have
\begin{lem}\label{lem11}
	For  $k\in P_1\setminus Q_{1}$, we have
	\begin{align}\label{exp1}
		|T^{-1}_{\widetilde{\Omega}_{k}^{1}}(x,y)|<e^{-\widetilde{\gamma}_0\|x-y\|_1} \  \text{\rm for }   x\in \partial^-\widetilde{\Omega}_{k}^{1} , \ y\in 2\Omega_{k}^1,
	\end{align}
	where $\widetilde{\gamma}_0=\gamma_0(1-N_1^{\frac{1}{c}-1})$.
\end{lem}
\begin{proof}[Proof of Lemma \ref{lem11}]
	From our construction, we have $$Q_{0} \subset \bigcup_{k \in P_{1}} A_{k}^{1}\subset \bigcup_{k \in P_{1}} \Omega_{k}^{1}.$$
	Thus
	$$(\widetilde{\Omega}_{k}^{1}\setminus A_k^1)\cap Q_{0}=\emptyset, $$
which shows   $\widetilde{\Omega}_{k}^{1}\setminus A_k^1$ is $0$-\textbf{good}. As a result, one has by \eqref{0good1},  $$|T^{-1}_{\widetilde{\Omega}_{k}^{1}\setminus A_k^1}(x,w)|<e^{-\gamma_0\|x-w\|_1} \  \text{\rm for }  x\in \partial^-\widetilde{\Omega}_{k}^{1} ,\ w\in  (\widetilde{\Omega}_{k}^{1}\setminus A_k^1)\cap2\Omega_{k}^1. $$
	Since  \eqref{1221e} and  $k\notin Q_1$, we have $$	\left\|T_{\widetilde{\Omega}_{k}^1}^{-1}\right\|<\delta_{0}^{-2}\delta_1^{-2}< \delta_1^{-3}.$$
	Using  resolvent identity implies
	\begin{align*}
		|T^{-1}_{\widetilde{\Omega}_{k}^{1}}(x,y)|&=|T^{-1}_{\widetilde{\Omega}_{k}^{1}\setminus A_k^1}(x,y)\chi_{\widetilde{\Omega}_{k}^{1}\setminus A_k^1}(y)-\sum_{(w', w)\in\partial A_k^1  } T^{-1}_{\widetilde{\Omega}_{k}^{1}\setminus A_k^1}(x,w)\Gamma(w,w')T^{-1}_{\widetilde{\Omega}_{k}^{1}}(w',y)|\\
		&< e^{-\gamma_0\|x-y\|_1}+ 4d\sup_{w\in\partial^+ A_k^1}e^{-\gamma_0\|x-w\|_1}\|T^{-1}_{\widetilde{\Omega}_{k}^{1}}\|\\
		&<e^{-\gamma_0\|x-y\|_1}+\sup_{w\in\partial^+ A_k^1}e^{-\gamma_0\left( \|x-y\|_1-\|y-w\|_1\right)+C|\log\delta_{1}|} \\\
		&<e^{-\gamma_0\|x-y\|_1}+e^{- \gamma_0\left(1-C\left(  \|x-y\|_1^{\frac{1}{c}-1}+\frac{|\log\delta_{1}|}{\|x-y\|_1}\right) \right)\|x-y\|_1}\\
		&<e^{-\gamma_0\|x-y\|_1}+e^{- \gamma_0\left(1-N_1^{\frac{1}{c}-1}\right)\|x-y\|_1}\\
		&=e^{-\widetilde{\gamma}_0\|x-y\|_1}
	\end{align*}
	since
	$$N_1^{c}\lesssim\operatorname{diam}\widetilde{\Omega}_k^1\sim\|x-y\|_1,  \    \|y-w\|_1\lesssim\operatorname{diam}\Omega_{k}^1\lesssim\left( \operatorname{diam}\widetilde{\Omega}_k^1\right) ^\frac{1}{c}$$
	and
	\begin{equation}\label{guj}
		|\log\delta_{1}|\sim|\log\delta_{0}|^{c^5}\sim N_1^{c^{10}\tau}<N_1^{\frac{1}{c}}.\end{equation}
\end{proof}
	
	We are able to prove Theorem \ref{th11}. First,  we prove the  estimate \eqref{th1L2} by  Schur's  test.
	Define $$\widetilde{P}_1=\{k \in P_1 :\  \Lambda\cap \Omega_{k}^{1}\cap Q_{0}\neq \emptyset\}, \  \Lambda'=\Lambda\backslash \bigcup_{k\in \widetilde{P}_1}\Omega^1_{k}.$$
	Then   $\Lambda' \cap Q_0=\emptyset$, which shows $\Lambda'$ is $0$-\textbf{good},  and  \eqref{0good2}--\eqref{0good1} hold for $\Lambda'$.  We have the following cases. 
	\begin{itemize}
		\item[(1).]
		Let $x\notin \bigcup_{k\in \widetilde{P}_1}2\Omega^1_{k}$.   Thus  $N_1\leq {\rm dist}(x,\partial_\Lambda^-\Lambda')$. For $y \in \Lambda$, the resolvent identity reads as 
		\begin{align*}
			T_{\Lambda}^{-1}(x,y)=T_{\Lambda'}^{-1}(x,y)\chi_{\Lambda'}(y)-\sum_{(w, w')\in\partial_\Lambda\Lambda'} T_{\Lambda'}^{-1}(x,w)\Gamma(w,w')T_{\Lambda}^{-1}(w',y).
		\end{align*}
		Since
		\begin{align*}
			\sum_{y\in \Lambda'}|T_{\Lambda'}^{-1}(x,y)\chi_{\Lambda'}(y)|&\leq|T_{\Lambda'}^{-1}(x,x)|+\sum_{\|x-y\|_1>0}|T_{\Lambda'}^{-1}(x,y)\chi_{\Lambda'}(y)|\\	
			&\leq\|T_{\Lambda'}^{-1}\|+\sum_{\|x-y\|_1>0}e^{-\gamma_0\|x-y\|_1}\\
			&\leq2\delta_0^{-2}
		\end{align*}
		and
		\begin{align*}
			\sum_{w\in\partial^-_\Lambda\Lambda'}|T_{\Lambda'}^{-1}(x,w)|&\leq\sum_{\|x-w\|_1\geq N_1}e^{-\gamma_0\|x-w\|_1}<e^{-\frac{1}{2}\gamma_0N_1},
		\end{align*}
		we get
		\begin{align*}
			\sum_{y\in \Lambda}|T_{\Lambda}^{-1}(x,y)|&\leq \sum_{y\in \Lambda'}|T_{\Lambda'}^{-1}(x,y)\chi_{\Lambda'}(y)|+\sum_{y\in \Lambda,(w, w')\in\partial_\Lambda\Lambda'}|T_{\Lambda'}^{-1}(x,w)\Gamma(w,w')T_{\Lambda}^{-1}(w',y)|  \nonumber \\
			&\leq2\delta_0^{-2}+2d\sum_{w\in\partial^-_\Lambda\Lambda'}|T_{\Lambda'}^{-1}(x,w)|\cdot \sup_{w'\in \Lambda}\sum_{y\in \Lambda}|T_{\Lambda}^{-1}(w',y)|\\
			&\leq2\delta_0^{-2}+\frac{1}{10}\sup_{w'\in \Lambda}\sum_{y\in \Lambda}|T_{\Lambda}^{-1}(w',y)|.
		\end{align*}
		\item[(2).]
		Let $x\in 2\Omega^1_{k}$ for some $k\in \widetilde{P}_1$. Thus by \eqref{1good}, we have  $\widetilde{\Omega}_{k}^1 \subset \Lambda$ and  $  k\notin Q_1$.
		For $y\in \Lambda,$ using resolvent identity shows
		\begin{align*}
			T_{\Lambda}^{-1}(x,y)=T_{\widetilde{\Omega}_{k}^1}^{-1}(x,y)\chi_{\widetilde{\Omega}_{k}^1}(y)-\sum_{(w,w')\in\partial_\Lambda\widetilde{\Omega}_{k}^1}T_{\widetilde{\Omega}_{k}^1}^{-1}(x,w)\Gamma(w,w')T_{\Lambda}^{-1}(w',y).
		\end{align*}
		By \eqref{1221e} ,  \eqref{exp1} and  since
		$$ N_1<\operatorname{diam}\widetilde{\Omega}_{k}^1\lesssim {\rm dist}(x,\partial^-_\Lambda\widetilde{\Omega}_{k}^1),$$
		we get
		\begin{align*}
			\sum_{y\in \Lambda}|T_{\Lambda}^{-1}(x,y)|&\leq \sum_{y\in \Lambda}|T_{\widetilde{\Omega}_{k}^1}^{-1}(x,y)\chi_{\widetilde{\Omega}_{k}^1}(y)|+\sum_{y\in\Lambda,(w,w')\in\partial_\Lambda\widetilde{\Omega}_{k}^1}|T_{\widetilde{\Omega}_{k}^1}^{-1}(x,w)\Gamma(w,w')T_{\Lambda}^{-1}(w',y)|   \\
			&<\# \widetilde{\Omega}_{k}^1\cdot \|T_{\widetilde{\Omega}_{k}^1}^{-1}\|+CN_1^{c^2d}e^{-\widetilde{\gamma}_0N_1}\sup_{w'\in \Lambda}\sum_{y\in \Lambda}|T_{\Lambda}^{-1}(w',y)|\\
			&<CN_1^{c^2d}\delta_{0}^{-2}\left\|\theta+k \cdot \omega-\theta_{1}\right\|^{-1} \cdot\left\|\theta+k \cdot \omega+\theta_{1}\right\|^{-1}+\frac{1}{10}\sup_{w'\in \Lambda}\sum_{y\in \Lambda}|T_{\Lambda}^{-1}(w',y)|\\
			&<\frac{1}{2}\delta_{0}^{-3}\left\|\theta+k \cdot \omega-\theta_{1}\right\|^{-1} \cdot\left\|\theta+k \cdot \omega+\theta_{1}\right\|^{-1}+\frac{1}{10}\sup_{w'\in \Lambda}\sum_{y\in \Lambda}|T_{\Lambda}^{-1}(w',y)|.
		\end{align*}	
	\end{itemize}
Combining estimates of  the above two cases yields
	\begin{align}\label{kuaile}
		\|T_{\Lambda}^{-1}\|&\leq\sup_{x\in \Lambda}\sum_{y\in \Lambda}|T_{\Lambda}^{-1}(x,y)| \nonumber \\
		&<\delta^{-3}_0\sup_{\left\{k\in P_1 :\   \widetilde{\Omega}_{k}^1\subset \Lambda \right\}}\left\|\theta+k\cdot  \omega-\theta_{1}\right\|^{-1} \cdot\left\|\theta+k \cdot \omega+\theta_{1}\right\|^{-1}.
	\end{align}
	
	Now we prove the off-diagonal decay estimate \eqref{th1exp}.	For every $w\in \Lambda,$ define its block in $\Lambda$
	\[J_w=\left\{\begin{aligned}
		&\Lambda_{\frac{1}{2}N_1}(w)\cap\Lambda \quad \text{if }   w\notin \bigcup_{k\in \widetilde{P}_1}2\Omega^1_{k}, \ &\textcircled{1} \\
		&\widetilde{\Omega}^{1}_k \quad \text{if }  w \in2\Omega_k^{1} \text{ for some }k\in  \widetilde{P}_1. \ &\textcircled{2}
	\end{aligned}\right. \]
	Then $\operatorname{diam}J_w\leq \max\left( \operatorname{diam}\Lambda_{\frac{1}{2}N_1}(w),\operatorname{diam}\widetilde{\Omega}^{1}_k\right) < 3N_1^{c^2}$.
	For \textcircled{1}, since $${\rm dist}(w,\Lambda\cap Q_0)\geq {\rm dist}(w,\bigcup_{k\in \widetilde{P}_1}\Omega_{k}^1)\geq N_1,$$   we have $J_w\cap Q_0=\emptyset.$  Thus  $J_w$ is $0$-\textbf{good}. Noticing that  ${\rm dist}(w,\partial^-_\Lambda J_w)\geq\frac{1}{2}N_1,$  from \eqref{0good1}, we have
	$$| T^{-1}_{J_w}(w,w')|<e^{-\gamma_0\|w-w'\|_1} \  \text{for $w'\in\partial^-_\Lambda J_w$.  }$$
	For \textcircled{2}, by \eqref{exp1}, we have
	$$  |T_{J_w}(w,w')|<e^{-\widetilde{\gamma}_0\|w-w'\|_1} \  \text{for $w'\in\partial^-_\Lambda J_w$.  } $$
	Let $\|x-y\|>N_{1}^{c^3}$.  Using  resolvent identity shows
	\begin{align*}
		T_{\Lambda}^{-1}(x,y)=T_{J_x}^{-1}(x,y)\chi_{J_x}(y)-\sum _{(w,w')\in\partial_\Lambda J_x}T_{J_x}^{-1}(x,w)\Gamma(w,w')T_{\Lambda}^{-1}(w',y).
	\end{align*}
	The first term of the above identity is zero because $y\notin J_x$  (since $\|x-y\|>N_{1}^{c^3}> 3N_1^{c^2}$). It follows that
	\begin{align*}
		|T_{\Lambda}^{-1}(x,y)|&\leq C N_1^{c^2d}e^{-\min\left( \gamma_0(1-2N_1^{-1}),\widetilde{\gamma}_0(1-N_1^{-1})\right)\|x-x_1\|_1 }|T_{\Lambda}^{-1}(x_1,y)| \\
		&\leq C N_1^{c^2d} e^{-\widetilde{\gamma}_0(1-N_1^{-1})\|x-x_1\|_1}|T_{\Lambda}^{-1}(x_1,y)|\\
		&<e^{-\widetilde{\gamma}_0(1-N_1^{-1}-\frac{C\log N_1}{N_1})\|x-x_1\|_1}|T_{\Lambda}^{-1}(x_1,y)|\\
		&<e^{-\gamma_0(1-N_1^{\frac{1}{c}-1})^2\|x-x_1\|_1}|T_{\Lambda}^{-1}(x_1,y)|\\
		&=e^{-\gamma_0'\|x-x_1\|_1}|T_{\Lambda}^{-1}(x_1,y)|
	\end{align*}	for some $x_1\in \partial_\Lambda^+ J_x $,
	where $\gamma_0'=\gamma_0(1-N_1^{\frac{1}{c}-1})^2.$  Then iterate and stop for some step $L$ such that $\|x_{L}-y\|<3N_1^{c^2}$. Recalling \eqref{guj}  and \eqref{kuaile}, we get
	\begin{align*}
		|T_{\Lambda}^{-1}(x,y)|&\leq e^{-\gamma_0'\|x-x_1\|_1}\cdots e^{-\gamma_0'\|x_{L-1}-x_L\|_1}|T_{\Lambda}^{-1}(x_L,y)| \\
		&\leq e^{-\gamma_0'(\|x-y\|_1-3N_1^{c^2})}\|T_\Lambda^{-1}\|\\
		&<e^{-\gamma_0'(1-3N_1^{c^2-c^3})\|x-y\|_1}\delta_1^{-3}\\
		&<e^{-\gamma_0'(1-3N_1^{c^2-c^3}-3\frac{|\log\delta_{1}|}{N_1^{c^3}})\|x-y\|_1}\\
		&<e^{-\gamma_0'(1-N_1^{\frac{1}{c}-1})\|x-y\|_1}\\
		&=e^{-\gamma_1\|x-y\|_1}.
	\end{align*}

This completes the proof of Theorem \ref{th11}.
\end{proof}

\subsection{Proof of Theorem \ref{Inthm}:  from ${\bf\mathcal{P}}_s$ to ${\bf\mathcal{P}}_{s+1}$}
	\begin{proof}[Proof of Theorem \ref{Inthm}]
	  We have finished the proof of ${\bf\mathcal{P}}_1$ in Subsection \ref{p1vf}. 
	Assume that ${\bf\mathcal{P}}_s$ holds true. In order to complete the proof of Theorem \ref{Inthm} it suffices to establish ${\bf\mathcal{P}}_{s+1}$.
		
	In the following,  we try to prove ${\bf\mathcal{P}}_{s+1}$ holds true. For this purpose,  we will  establish $({\bf a})_{s+1}$--$({\bf f})_{s+1}$ assuming $({\bf a})_{s}$--$({\bf f})_{s}$.	
	We divide the proof into 3 steps. 
	Let
	\begin{equation}\label{Q_s}
		Q_{s}^{\pm}=\left\{k \in P_s :\  \left\|\theta+k\cdot \omega \pm \theta_{s}\right\|<\delta_{s}\right\},
		\  Q_{s}=Q_{s}^{+}\cup Q_{s}^{-}.
	\end{equation}
and
	\begin{equation}\label{wQ_s}
		\widetilde{Q}_{s}^{\pm}=\left\{k \in P_s:\  \left\|\theta+k\cdot \omega \pm \theta_{s}\right\|<\delta_{s}^{\frac{1}{100}}\right\},\ \widetilde{Q}_{s}=\widetilde{Q}_{s}^{+}\cup\widetilde{Q}_{s}^{-}.
	\end{equation}		
	{\bf STEP1: The case $({\bf C}1)_s$ Occurs }: i.e., 
\begin{equation}\label{case1s+1}
	{\rm dist}\left(\widetilde{Q}_{s}^{-}, Q_{s}^{+}\right)>100N_{s+1}^c.
\end{equation}
\begin{rem}\label{r1}
	We can prove  that  $${\rm dist}\left(\widetilde{Q}_{s}^{-}, Q_{s}^{+}\right)={\rm dist}\left(\widetilde{Q}_{s}^{+}, Q_{s}^{-}\right).$$Thus \eqref{case1s+1}  also implies that
	\begin{equation}\label{cases+11}
		{\rm dist}\left(\widetilde{Q}_{s}^{+}, Q_{s}^{-}\right)>100 N_{s+1}^{c}.
	\end{equation}
	By \eqref{P>}  and the definitions  of $Q_s^{\pm}$ (cf. \eqref{Q_s})  and $\widetilde{Q}_s^{\pm}$ (cf. \eqref{wQ_s}), we obtain
	\begin{align}		
	\label{keyi}Q_s^{\pm}&=\{k \in \mathbb{Z}^d+\frac{1}{2}\sum_{i=0}^{s-1}l_i:\ \left\|\theta+k\cdot \omega \pm \theta_{s}\right\|<\delta_{s}\},\\
	\nonumber \widetilde{Q}_s^{\pm}&=\{k \in \mathbb{Z}^d+\frac{1}{2}\sum_{i=0}^{s-1}l_i:\ \left\|\theta+k\cdot \omega \pm \theta_{s}\right\|<\delta_{s}^\frac{1}{100}\}.
	\end{align}
Then the proof  is similar to that of  Remark \ref{r0} and we omit the details.
\end{rem}
Assuming \eqref{case1s+1},	then we define \begin{equation}\label{dy}
	P_{s+1}=Q_{s},\  l_s=0.
\end{equation}
By \eqref{2233} and \eqref{3322},  we have \begin{equation}\label{Ps+11}
	P_{s+1}\subset\left\{k \in \mathbb{Z}^d+\frac{1}{2}\sum_{i=0}^{s}l_i :\  \min_{\sigma=\pm1}(\|\theta+k\cdot\omega+\sigma\theta_{s}\| )<\delta_s\right\}.
\end{equation}
Thus from  \eqref{cases+11}, we get for $k,k'\in P_{s+1}$ with $k\neq k',$
\begin{equation}\label{aha}
	\|k-k'\|>\min\left( |\log\frac{\gamma}{2\delta_s}|^\frac{1}{\tau},100N_{s+1}^c\right)\geq100N_{s+1}^c.
\end{equation}
In the following,  we will associate every $k\in P_{s+1}$ with blocks $\Omega_k^{s+1}$ and  $\widetilde{\Omega}_k^{s+1}$ 
so that \begin{align*}
\Lambda_{N_{s+1}}(k)&\subset \Omega_{k}^{s+1}\subset \Lambda_{N_{s+1}+50N_{s}^{c^2}}(k), \\
\Lambda_{N_{s+1}^c}(k)&\subset \widetilde{\Omega}_k^{s+1}\subset \Lambda_{N_{s+1}^c+50N_{s}^{c^2}}(k),
\end{align*}
and
\begin{equation}\label{hehe}
	\left\{\begin{aligned}
		&\Omega_{k}^{s+1} \cap \widetilde{\Omega}_{k'}^{s'} \neq \emptyset\ \left(s^{\prime}<s+1\right) \Rightarrow \widetilde{\Omega}_{k'}^{s'} \subset \Omega_{k}^{s+1}, \\
		&\widetilde{\Omega}_{k}^{s+1} \cap \widetilde{\Omega}_{k'}^{s'} \neq \emptyset\ \left(s^{\prime}<s+1\right) \Rightarrow \widetilde{\Omega}_{k'}^{s'} \subset \widetilde{\Omega}_{k}^{s+1}, \\
		&{\rm dist}(\widetilde{\Omega}_{k}^{s+1}, \widetilde{\Omega}_{k^{\prime}}^{s+1})>10\operatorname{diam}\widetilde{\Omega}_{k}^{s+1} \ { \rm for } \ k\neq k^{\prime} \in P_{s+1} .
	\end{aligned}\right.
\end{equation}
In addition, the set
$$\widetilde{\Omega}_k^{s+1}-k\subset \mathbb{Z}^d+\frac{1}{2}\sum_{i=0}^{s}l_i$$ is independent  of $k\in P_{s+1}$ and  is symmetrical about the origin.

Such $\Omega_{k}^{s+1}$ and  $\widetilde{\Omega}_{k}^{s+1}$ can be constructed  by the following argument (We only consider  $\widetilde{\Omega}_{k}^{s+1}$ since $\Omega^{s+1}_k$ is discussed by  the similar  argument).
Fixing  $k_0 \in Q_{s}^+$, we  start from
\[J_{0,0}=\Lambda _{N_{s+1}^c}(k_0).\]
Define
$$H_r=( k_0-P_{s+1}+P_{s-r})\cup(k_0+P_{s+1}-P_{s-r})  \quad (0\leq r\leq s-1).$$
Notice that by  \eqref{Ps+11},  we have $k_0-P_{s+1}\in \mathbb{Z}^d$ and,  $P_{s-r}\subset\mathbb{Z}^d+\frac{1}{2}\sum_{i=0}^{s-r-1}l_i$ since \eqref{2233}--\eqref{3322}. Thus $$H_{s-r}\subset\mathbb{Z}^d+\frac{1}{2}\sum_{i=0}^{s-r-1}l_i.$$
Define inductively $$J_{r,0} \subsetneqq J_{r,1}\subsetneqq \cdots \subsetneqq J_{r,t_r}:=J_{r+1,0}, $$
where
$$J_{r,t+1}=J_{r,t}\bigcup \left(  \bigcup_{\{ h\in H_{r}:\ \Lambda_{2N_{s-r}^{c^2}}(h)\cap J_{r,t}\neq\emptyset\} }\Lambda_{2N_{s-r}^{c^2}}(h)\right)  $$
and $t_r$ is the  largest integer satisfying the $\subsetneqq$ relationship (the following argument shows that  $t_r<10$).
Thus
\begin{equation}\label{??}
	h\in H_{r}, \ \Lambda_{2N_{s-r}^{c^2}}(h)\cap J_{r+1,0}\neq\emptyset\Rightarrow \Lambda_{2N_{s-r}^{c^2}}(h)\subset J_{r+1,0}.
\end{equation}
For $\widetilde{k} \in k_0- P_{s+1}$, we have  since
\eqref{Ps+11} $$\min\left(\|\widetilde{k}\cdot  \omega\|,\|\widetilde{k} \cdot \omega+ 2\theta_{s}\|\right)<2 \delta_s.$$
For  $k'\in P_{s-r}$,  we get since   \eqref{2233} and \eqref{3322} that 
\begin{align}
	\label{,,}&\min_{\sigma=\pm1}(\|\theta+k'\cdot \omega+\sigma\theta_{s-r-1}\|)<\delta_{s-r-1}\  {\rm if\   ({\bf C}1)_{s-r}\ holds \ true},\\
\label{..}&\|\theta+k'\cdot\omega\|<3\delta_{s-r-1}^{\frac{1}{100}}\text{ or } \|\theta+k'\cdot\omega+\frac{1}{2}\|<3\delta_{s-r-1}^{\frac{1}{100}}\  {\rm if\   ({\bf C}2)_{s-r}\ holds \ true}.
\end{align}
Thus for $h \in k_0-P_{s+1}+P_{s-r}$,  we  obtain for  \eqref{,,},
\begin{align*}
\min_{\sigma=\pm1}\left(\|\theta+h\cdot  \omega+\sigma\theta_{s-r-1}\|,\|\theta+h \cdot \omega+ 2\theta_{s}+\sigma\theta_{s-r-1}\|\right)<2\delta_{s-r-1}
\end{align*}
and for \eqref{..},
\begin{align*}
\min(\|\theta+h\cdot\omega\|,\|\theta+h\cdot\omega+\frac{1}{2}\|,\|\theta+h\cdot\omega+2\theta_{s}\|,\|\theta+h\cdot\omega+\frac{1}{2}+2\theta_{s}\|)<4\delta_{s-r-1}^{\frac{1}{100}}.
\end{align*}
Notice that $ k_0+P_{s+1}-P_{s-r}=2k_0-(k_0-P_{s+1}+P_{s-r})$ is symmetrical to  $k_0-P_{s+1}+P_{s-r}$  about $k_0$.	
Thus, if a set $\Lambda  \ (\subset\mathbb{Z}^d+\frac{1}{2}\sum_{i=0}^{s-r-1}l_i)$ contains $10$  distinct elements of $H_r$, then
\begin{equation}\label{contradiction}
	\operatorname{diam}\Lambda>\left|\log\frac{\gamma}{8\delta_{s-r-1}^{\frac{1}{100}}}\right|^\frac{1}{\tau}\gg100N_{s-r}^{c^2}.
\end{equation}
We claim that $t_r<10$. Otherwise, there exist distinct $h_t\in H_r\ (1\leq t\leq10)$, such that
$$\Lambda_{2N_{s-r}^{c^2}}(h_{1})\cap J_{r,0}\neq\emptyset,\  \Lambda_{2N_{s-r}^{c^2}}(h_t)\cap\Lambda_{2N_{s-r}^{c^2}}(h_{t+1})\neq\emptyset. $$
In particular,
$${\rm dist}(h_t,h_{t+1})\leq4N_{s-r}^{c^2}.$$ Thus
$$h_t\in \Lambda_{40N_{s-r}^{c^2}}(0)+h_1 \  (1\leq t\leq10). $$
This  contradicts \eqref{contradiction}. Thus we have shown
\begin{equation}\label{?}
	J_{r+1,0}=J_{r,t_r}\subset \Lambda_{40N_{s-r}^{c^2}}(J_{r,0}).
\end{equation}	
Since $$\sum_{r=0}^{s-1}40N_{s-r}^{c^2}<50N_s^{c^2},$$
we find $J_{s,0}$  to satisfy
$$\Lambda _{N_{s+1}^c}(k_0)=J_{0,0}\subset J_{s,0}\subset \Lambda_{50N_{s}^{c^2}}(J_{0,0})\subset\Lambda _{N_{s+1}^c+50N_{s}^{c^2}}(k_0).$$
Now, for any $k\in P_{s+1},$  define
\begin{equation}\label{wuhu}
	\widetilde{\Omega}_k^{s+1}=J_{s,0}+(k-k_0).
\end{equation}
Using  $k-k_0\in \mathbb{Z}^d$ and $\widetilde{\Omega}_k^{s+1}\subset\mathbb{Z}^d$  yields
$$\Lambda_{N_{s+1}^c}(k)\subset \widetilde{\Omega}_k^{s+1}\subset \Lambda_{N_{s+1}^c+50N_{s}^{c^2}}(k).$$
We are able  to verify \eqref{hehe}. 
In fact, since  \eqref{aha} and  $50N_{s}^{c^2}\ll N_{s+1}^{c}$, we get
$${\rm dist}(\widetilde{\Omega}_{k}^{s+1}, \widetilde{\Omega}_{k^{\prime}}^{s+1})>10\operatorname{diam}\widetilde{\Omega}_{k}^{s+1} \ \text {for } k \neq k^{\prime} \in P_{s+1} .$$
Assume that for some $k\in P_{s+1}$ and $k'\in P_{s'}\ (1\leq s'\leq s)$, 
$$\widetilde{\Omega}_{k}^{s+1} \cap \widetilde{\Omega}_{k'}^{s'} \neq \emptyset. $$
Then \begin{equation}\label{cb}
	\left( \widetilde{\Omega}_{k}^{s+1}+(k_0-k)\right)  \cap \left( \widetilde{\Omega}_{k'}^{s'}+(k_0-k) \right) \neq \emptyset .
\end{equation}
From  $$\Lambda_{N_{s'}^c}(k')\subset \widetilde{\Omega}_{k'}^{s'}\subset \Lambda_{N_{s'}^c+50N_{s'-1}^{c^2}}(k')\subset\Lambda_{1.5N_{s'}^{c^2}}(k') , $$
$\widetilde{\Omega}_{k}^{s+1}+(k_0-k)=J_{s,0}$  and \eqref{cb},  we obtain
$$J_{s,0}\cap\Lambda_{1.5N_{s'}^{c^2}}(k'+k_0-k)\neq\emptyset.$$
Recalling \eqref{?}, we have
$$J_{s,0}\subset \Lambda_{50N_{s'-1}^{c^2}}(J_{s-s'+1,0}).$$
Thus $$\Lambda_{50N_{s'-1}^{c^2}}(J_{s-s'+1,0})\cap \Lambda_{1.5N_{s'}^{c^2}}(k'+k_0-k)\neq\emptyset.   $$
From $50N_{s'-1}^{c^2}\ll 0.5N_{s'}^{c^2}$, it follows that
$$J_{s-s'+1,0}\cap \Lambda_{2N_{s'}^{c^2}}(k'+k_0-k)\neq\emptyset.$$
Since  $k'\in P_{s'}$, we have $k'+k_0-k\in H_{s-s'}$, and by \eqref{??},
$$\Lambda_{2N_{s'}^{c^2}}(k'+k_0-k)\subset J_{s-s'+1,0}\subset J_{s,0}.$$
Hence
$$\widetilde{\Omega}_{k'}^{s'}\subset\Lambda_{2N_{s'}^{c^2}}(k')\subset J_{s,0}+(k-k_0)=\widetilde{\Omega}_k^{s+1}.$$
Next, we will show $\widetilde{\Omega}_k^{s+1}-k$ is independent of $k$. For this, recalling \eqref{wuhu} and from   $l_i\in\mathbb{Z}^d,$ $\widetilde{\Omega}_k^{s+1}\subset\mathbb{Z}^d$, $k\in P_{s+1}\subset\mathbb{Z}^d+\frac{1}{2}\sum_{i=0}^{s}l_i,$ we obtain that
$$\widetilde{\Omega}_k^{s+1}-k\subset \mathbb{Z}^d-\frac{1}{2}\sum_{i=0}^{s}l_i=\mathbb{Z}^d+\frac{1}{2}\sum_{i=0}^{s}l_i$$
and
$$\widetilde{\Omega}_k^{s+1}-k=J_{s,0}+(k-k_0)-k=\widetilde{\Omega}_{k_0}^{s+1}-k_0$$ is independent of $k$.
Finally, we prove the symmetry property of $\widetilde{\Omega}_k^{s+1}$. The definition of $H_r$ implies that it is symmetrical about $k_0$,  which implies all   $J_{r,t}$ is symmetrical about $k_0$ as well.  In particular, $	\widetilde{\Omega}_{k_0}^{s+1}=J_{s,0}$ is  symmetrical about $k_0$, i.e.,  $\widetilde{\Omega}_{k_0}^{s+1}-k_0$ is symmetrical about origin. In summary, we have established $({\bf a})_{s+1}$ and $({\bf b})_{s+1}$ in the case $({\bf C}1)_s.$

Now we turn to the proof of $({\bf c})_{s+1}$.  First,  in this construction we have for every $k' \in Q_{s}\ (= P_{s+1})$,
$$\widetilde{\Omega}_{k'}^{s} \subset \Omega_{k'}^{s+1}.$$
For every  $k\in P_{s+1}$,  define
$$A_k^{s+1}=A_k^{s}.$$	Then $A_k^{s+1}\subset \Omega_k^s \subset\Omega_k^{s+1}$
and $\# A_k^{s+1}=\# A_k^{s}\leq2^s.$ It remains to show
$\widetilde{\Omega}_{k}^{s+1} \setminus A_k^{s+1}$ is $s$-\textbf{good}, i.e., 
$$\left\{\begin{aligned}
	&l'\in Q_{s'},\widetilde{\Omega}_{l'}^{s'}\subset(\widetilde{\Omega}_{k}^{s+1} \setminus A_k^{s+1}),\widetilde{\Omega}_{l'}^{s'}\subset \Omega_{l}^{s'+1} \Rightarrow \widetilde{\Omega}_{l}^{s'+1} \subset(\widetilde{\Omega}_{k}^{s+1} \setminus A_k^{s+1})\ \text{for }s'<s,\\
	&\left\{l\in P_s :\   \widetilde{\Omega}_{l}^s\subset (\widetilde{\Omega}_{k}^{s+1} \setminus A_k^{s+1}) \right\}\cap Q_{s}= \emptyset.
\end{aligned}\right.
$$
Assume  that  $$l'\in Q_{s'},\ \widetilde{\Omega}_{l'}^{s'}\subset(\widetilde{\Omega}_{k}^{s+1}\setminus A_k^{s+1}),\ \widetilde{\Omega}_{l'}^{s'}\subset \Omega_{l}^{s'+1} .$$
We have the following two cases.  The first case is  $s'\leq s-2$. In this case,  since $\emptyset\neq\widetilde{\Omega}_{l'}^{s'}\subset\widetilde{\Omega}_{l}^{s'+1}\cap\widetilde{\Omega}_{k}^{s+1}
,$
we get  by using \eqref{hehe} that $ \widetilde{\Omega}_{l}^{s'+1}\subset \widetilde{\Omega}_{k}^{s+1}.$
Assuming \begin{equation}\label{cuo}
	\widetilde{\Omega}_{l}^{s'+1}\cap A_k^{s+1}\neq \emptyset,
\end{equation}
then  $\widetilde{\Omega}_{l}^{s'+1}\cap \widetilde{\Omega}_{k}^{s}\neq\emptyset.$
Thus  from \eqref{haha}  (since $s'+1< s$), one has $\widetilde{\Omega}_{l}^{s'+1}\subset \widetilde{\Omega}_{k}^{s},$ which implies
$\widetilde{\Omega}_{l'}^{s'}\subset(\widetilde{\Omega}_{k}^{s}\setminus A_k^{s}).$ Since $(\widetilde{\Omega}_{k}^{s}\setminus A_k^{s})$ is $(s-1)$-\textbf{good},   we get
$$\widetilde{\Omega}_{l}^{s'+1}\subset(\widetilde{\Omega}_{k}^{s}\setminus A_k^{s})\subset(\widetilde{\Omega}_{k}^{s+1} \setminus A_k^{s+1}).$$
This contradicts  \eqref{cuo}. We then consider  the case   $s'=s-1$.  From   $\widetilde{\Omega}_{l'}^{s-1}\subset \Omega_{l}^{s}$ and $\widetilde{\Omega}_{l}^{s}\cap A_k^s \neq\emptyset$,  then $l=k$ and $\widetilde{\Omega}_{l'}^{s-1}\subset (\widetilde{\Omega}_{k}^{s}\setminus A_k^s).$
This  contradicts 
$$\left\{l\in P_{s-1} :\   \widetilde{\Omega}_{l}^{s-1}\subset(\widetilde{\Omega}_{k}^{s}\setminus A_k^s) \right \}\cap Q_{s-1}= \emptyset$$
because $(\widetilde{\Omega}_{k}^{s}\setminus A_k^s)$  is $(s-1)$-\textbf{good}.
Finally, if $l\in Q_s$  and $\widetilde{\Omega}_{l}^s\subset \widetilde{\Omega}_{k}^{s+1},$
then  $l=k$ since $k$ is the only element of $Q_s$ such that  $\widetilde{\Omega}_{k}^s\subset \widetilde{\Omega}_{k}^{s+1}$  by the separation  property of $Q_s$. As a result, 
$\widetilde{\Omega}_{l}^s\nsubseteq (\widetilde{\Omega}_{k}^{s+1} \setminus A_k^{s+1}),$
which implies  $$\left\{l\in P_s :\   \widetilde{\Omega}_{l}^s\subset (\widetilde{\Omega}_{k}^{s+1} \setminus A_k^{s+1}) \right\}\cap Q_{s}= \emptyset.$$	
Moreover, the set $$A_k^{s+1}-k=A_k^s-k$$ is independent of $k\in P_{s+1} $ and symmetrical about the origin since the induction assumptions on $A_k^s$ of the $s$-th step. This finishes the proof of $({\bf c})_{s+1}$ in the case $({\bf C}1)_{s}.$

In the following,  we try to prove $({\bf d})_{s+1}$ and $({\bf f})_{s+1}$ in the case $({\bf C}1)_{s}$. For  the case $k\in Q_s^-$, we  consider the analytic matrix-valued function  
\begin{align*}
	M_{s+1}(z):=T(z)_{\widetilde{\Omega}_k^{s+1}-k}=\left(\cos 2 \pi(z+n\cdot  \omega)\delta_{n,n'}-E+\varepsilon \Delta\right)_{n\in \widetilde{\Omega}_k^{s+1}-k}
\end{align*}
defined  in
\begin{equation}\label{d1-}
	\{z \in \mathbb{C}:\  \left| z-\theta_{s} \right|< \delta_{s}^{\frac{1}{10}}\}.
\end{equation}
If $k'\in P_s$ and $\widetilde{\Omega}_{k'}^s\subset(\widetilde{\Omega}_{k}^{s+1} \setminus A_k^{s+1})$, then $0\neq \|k'-k\|\leq 2N_{s+1}^{c}$. Thus
\begin{align*}
	\|\theta+k'\cdot \omega-\theta_{s}\|&\geq \|(k'-k)\cdot \omega\|-\|\theta+k\cdot \omega-\theta_{s}\|\\
	&\geq \gamma e^{-(2N_{s+1}^c)^{\tau}}-\delta_s\\
	&\geq \gamma e^{-2^{\tau}|\log\frac{\gamma}{\delta_{s}}|^{\frac{1}{c^4}}}-\delta_s\\
	&>\delta_s^{\frac{1}{10^4}}.
\end{align*}
By \eqref{cases+11}, we have $k'\notin \widetilde{Q}_{s}^+$ and thus
$$	\|\theta+k'\cdot \omega+\theta_{s}\|>\delta_s^\frac{1}{100}.$$
From \eqref{L2}, we obtain
\begin{align}
	\nonumber\|T_{\widetilde{\Omega}_{k}^{s+1} \setminus A_k^{s+1}}^{-1}\|&<\delta^{-3}_{s-1}\sup_{\left\{k'\in P_s :\   \widetilde{\Omega}_{k'}^s\subset (\widetilde{\Omega}_{k}^{s+1} \setminus A_k^{s+1})\right\}}\|\theta+k'\cdot\omega-\theta_{s}\|^{-1}\cdot \|\theta+k'\cdot\omega+\theta_{s}\|^{-1}\\
	\label{sgood1}&<\frac{1}{2}\delta_s^{-2\times \frac{1}{100}}.
\end{align}
One  may restate \eqref{sgood1} as
$$	\left\|\left( M_{s+1}(\theta+k\cdot\omega)_{(\widetilde{\Omega}_{k}^{s+1} \setminus A_k^{s+1})-k}\right) ^{-1}\right\|<\frac{1}{2}\delta_s^{-2\times \frac{1}{100}}.$$
Notice that  \begin{align}
	\nonumber\|z-(\theta+k\cdot\omega)\|&\leq|z-\theta_s|+\|\theta+k\cdot\omega-\theta_s\|\\
	\label{dd}&<\delta_{s}^{\frac{1}{10}}+\delta_{s}<2\delta_s^{\frac{1}{10}}\ll\delta_s^{2\times\frac{1}{100}}.
\end{align}
Thus by Neumann series argument, we can show
\begin{equation}\label{404}
	\left\|\left( M_{s+1}(z)_{(\widetilde{\Omega}_{k}^{s+1} \setminus A_k^{s+1})-k}\right) ^{-1}\right\|<\delta_s^{-2\times \frac{1}{100}}.
\end{equation}
We may then control $M_{s+1}(z)^{-1}$ by the inverse of
\begin{align*}
	S_{s+1}(z)&=M_{s+1}(z)_{ A_k^{s+1}-k}-R_{A_k^{s+1}-k} M_{s+1}(z) R_{(\widetilde{\Omega}_{k}^{s+1} \setminus A_k^{s+1})-k}\\ 
	&\ \ \ \ \times \left( M_{s+1}(z)_{(\widetilde{\Omega}_{k}^{s+1} \setminus A_k^{s+1})-k}\right) ^{-1} R_{(\widetilde{\Omega}_{k}^{s+1} \setminus A_k^{s+1})-k}  M_{s+1}(z)R_{A_k^{s+1}-k}.
\end{align*}
Our next aim is to analyze the function $\operatorname{det}S_{s+1}(z).$
Since $A_k^{s+1}-k=A_k^{s}-k\subset \Omega_k^{s}-k$ and ${\rm dist}(\Omega_k^{s},\partial^+\widetilde{\Omega}_{k}^{s})>1 $, we obtain
$$R_{A_k^{s+1}-k} M_{s+1}(z) R_{(\widetilde{\Omega}_{k}^{s+1} \setminus A_k^{s+1})-k}=R_{A_k^{s}-k} M_{s+1}(z) R_{(\widetilde{\Omega}_{k}^{s} \setminus A_k^{s})-k}.$$
Thus \begin{align*}
	S_{s+1}(z)&=M_{s+1}(z)_{ A_k^{s}-k}-R_{A_k^{s}-k} M_{s+1}(z) R_{(\widetilde{\Omega}_{k}^{s} \setminus A_k^{s})-k}\\ 
&\ \ \ \ \times \left( M_{s+1}(z)_{(\widetilde{\Omega}_{k}^{s+1} \setminus A_k^{s+1})-k}\right) ^{-1} R_{(\widetilde{\Omega}_{k}^{s} \setminus A_k^{s})-k}  M_{s+1}(z)R_{A_k^{s}-k}.
\end{align*}
Since $\widetilde{\Omega}_k^s\setminus A_k^s$ is $(s-1)$-\textbf{good} and by \eqref{L2}--\eqref{exp},  we get
\begin{align*}\|{T}_{\widetilde{\Omega}_k^s\setminus A_k^s}^{-1}\|&<\delta_{s-1}^{-3},\\
	\left|{T}_{\widetilde{\Omega}_k^s\setminus A_k^s}^{-1}(x, y)\right|&<e^{-\gamma_{s-1}\|x-y\|_1} \ \text {for }\|x-y\|>N_{s-1}^{c^3}.
\end{align*}
Equivalently, 
\begin{align}
\label{de}\left\|\left( M_{s+1}(\theta+k\cdot \omega)_{(\widetilde{\Omega}_k^s\setminus A_k^s)-k}\right) ^{-1}\right\|&<\delta_{s-1}^{-3},\\
\label{da}\left|\left( M_{s+1}(\theta+k\cdot \omega)_{(\widetilde{\Omega}_{k}^{s} \setminus A_k^{s})-k}\right) ^{-1}(x,y) \right| &<e^{-\gamma_{s-1}\|x-y\|_1}\  \text { for }\|x-y\|>N_{s-1}^{c^3}.
\end{align}
In the set defined by  \eqref{dd},  we  claim that
\begin{equation}\label{jidan}
	\left|\left( M_{s+1}(z)_{(\widetilde{\Omega}_{k}^{s} \setminus A_k^{s})-k}\right) ^{-1}(x,y) \right| <\delta_s^{10}\ \text { for }\|x-y\|>N_{s-1}^{c^4}.
\end{equation}
\begin{proof}[Proof of the Claim {\rm (i.e., \eqref{jidan})}]
	Denote $$T_1= M_{s+1}(\theta+k\cdot \omega)_{(\widetilde{\Omega}_{k}^{s} \setminus A_k^{s})-k},\    T_2=M_{s+1}(z)_{(\widetilde{\Omega}_{k}^{s} \setminus A_k^{s})-k}.$$
	Then $D=T_1-T_2$ is diagonal  so that  $\|D\|<4\pi\delta_s^\frac{1}{10}$ since \eqref{dd}.
	Using  Neumann series expansion yields
	\begin{equation}\label{z}
		T_2^{-1}=(I-T_1^{-1}D)^{-1}T_1^{-1}=\sum_{i=0}^{+\infty}\left( T_1^{-1}D\right)^iT_1^{-1}.
	\end{equation}
	Since \eqref{de}  and \eqref{da},  we have
	$$	 \left|T_1^{-1}(x,y) \right| <\delta_{s-1}^{-3}e^{-\gamma_{s-1}(\|x-y\|_1-N_{s-1}^{c^3})}. $$
	Thus for $\|x-y\|>N_{s-1}^{c^4}$ and $0\leq i\leq200$,
	\begin{align*}
		|( \left( T_1^{-1}D\right)^iT_1^{-1}) (x,y)|&\leq\left(4\pi\delta_s^\frac{1}{10} \right)^i\sum_{w_1,\cdots,w_i}|T_1(x,w_1)\cdots T_1(w_{i-1},w_i)T_1(w_i,y)|\\
		&<\left(4\pi\delta_s^\frac{1}{10} \right)^iCN_s^{c^2d}\delta_{s-1}^{-3(i+1)}e^{-\gamma_{s-1}(\|x-y\|_1-(i+1)N_{s-1}^{c^3})}	\\
		&<\delta_s^{\frac{1}{20}(i-1)}e^{-\gamma_{s-1}(N_{s-1}^{c^4}-(i+1)N_{s-1}^{c^3})}.
	\end{align*}
	From  $2<\gamma_{s-1}$, $201N_{s-1}^{c^3}<\frac{1}{2}N_{s-1}^{c^4}$ and  $|\log\delta_s|\sim|\log\delta_{s-1}|^{c^5}\sim N_{s}^{c^{10}\tau}\sim N_{s-1}^{c^{15}\tau}<N_{s-1}^{c^3},$ we get
	$$e^{-\gamma_{s-1}(N_{s-1}^{c^4}-(i+1)N_{s-1}^{c^3})}<e^{-N_{s-1}^{c^4}}<\delta_{s}^{20}.$$
	Hence \begin{equation}\label{y}
		\sum_{ i=0}^{200}|( \left( T_1^{-1}D\right)^iT_1^{-1}) (x,y)|<\frac{1}{2}\delta_{s}^{10}.
	\end{equation}
	For $i>200$,
	$$	|( \left( T_1^{-1}D\right)^iT_1^{-1}) (x,y)|<\left(4\pi\delta_s^\frac{1}{10} \right)^i\delta_{s-1}^{-3(i+1)}<\delta_s^{\frac{1}{20}i}<\delta_s^{10}\delta_s^{\frac{1}{20}(i-200)}.$$
	Thus\begin{equation}\label{x}
		\sum_{ i>200}|( \left( T_1^{-1}D\right)^iT_1^{-1}) (x,y)|<\frac{1}{2}\delta_{s}^{10}.
	\end{equation}
	Combining \eqref{z}, \eqref{y}  and  \eqref{x}, we get
	$$\left|T_2^{-1}(x,y) \right| <\delta_s^{10}\ \text { for }\|x-y\|>N_{s-1}^{c^4}.$$
This completes the proof of \eqref{jidan}. 
\end{proof}
Denote $X=(\widetilde{\Omega}_{k}^{s} \setminus A_k^{s})-k$  and $Y=(\widetilde{\Omega}_{k}^{s+1} \setminus A_k^{s+1})-k$. Let $x\in X$ satisfy  ${\rm dist}(x,A^{s}_k-k)\leq1$. By resolvent identity, we have for any  $y \in Y$,
\begin{align}\label{22}
	&\left(M_{s+1}(z)_{Y}\right) ^{-1}(x,y)-\chi_{X}(y)\left( M_{s+1}(z)_{X}\right) ^{-1}(x,y)\nonumber \\
	&\ \ \ \ =-\sum _{(w,w') \in\partial_{Y} X} \left(M_{s+1}(z)_{X}\right)^{-1}(x,w)\Gamma(w,w')\left(M_{s+1}(z)_{Y}\right) ^{-1}(w',y).
\end{align}
Since
$${\rm dist}(x,w)\geq {\rm dist}(A_k^s-k,\partial^- \widetilde{\Omega}_k^s-k )-2>N_s> N_{s-1}^{c^4},$$
 \eqref{404} and \eqref{jidan}, the right hand side (RHS) of \eqref{22} is bounded by
$CN_s^{c^2d}\delta_s^{-\frac{1}{50}}\delta_s^{10}<\delta_s^9.$
It then follows that   \begin{align*}
	&R_{A_k^{s}-k} M_{s+1}(z) R_{X}  \left( M_{s+1}(z)_{Y}\right) ^{-1}\\  
	&\ \ \ \  = R_{A_k^{s}-k} M_{s+1}(z) R_{X}  \left( M_{s+1}(z)_{X}\right) ^{-1} R_{X}+O(\delta_s^9).
\end{align*}
As a result,  \begin{align*}
	&R_{A_k^{s}-k} M_{s+1}(z) R_{X} \left( M_{s+1}(z)_{Y}\right) ^{-1} R_{X}  M_{s+1}(z)R_{A_k^{s}-k}\\
	&\ \ \ \ = R_{A_k^{s}-k} M_{s+1}(z) R_{X} \left( M_{s+1}(z)_{X}\right) ^{-1} R_{X}  M_{s+1}(z)R_{A_k^{s}-k}+O(\delta_s^9)\\
	&\ \ \ \ = R_{A_k^{s}-k} M_{s}(z) R_{X} \left( M_{s}(z)_{X}\right) ^{-1} R_{X}  M_{s}(z)R_{A_k^{s}-k}+O(\delta_s^9)
\end{align*}
and
\begin{align*}
	S_{s+1}(z)&=M_{s}(z)_{ A_k^{s}-k}-R_{A_k^{s}-k} M_{s}(z) R_{X} \left( M_{s}(z)_{X}\right) ^{-1} R_{X}  M_{s}(z)R_{A_k^{s}-k}+O(\delta_s^9)\\
	&=S_s(z)+O(\delta_s^9),
\end{align*}
which implies \eqref{10} for  the  $(s+1)$-th step. 
Recalling  \eqref{d1-} and \eqref{dede}, we have since \eqref{sSsim}
$$	\operatorname{det}S_{s}(z)\stackrel{\delta_{s-1}}{\sim}\|z-\theta_{s}\|\cdot \|z+\theta_{s}\|.$$
By Hadamard's inequality,  we obtain \begin{align*}
	\operatorname{det}S_{s+1}(z)&=\operatorname{det}S_s(z)+O((2^s)^210^{2^s}\delta_s^9)\\
	&=\operatorname{det}S_s(z)+O(\delta_s^8),
\end{align*}
where we use the fact that  $\#( A^{s}_k-k)\leq2^{s}$, \eqref{10} and $\log\log|\log \delta_s|\sim s.$
Notice that \begin{align*}
	\|z+\theta_{s}\|&\geq\|\theta+k\cdot\omega+\theta_{s}\|-\|z-\theta_{s}\|-\|\theta+k\cdot\omega-\theta_{s}\|\\
	&>\delta_{s}^{\frac{1}{100}}-\delta_{s}^{\frac{1}{10}}-\delta_{1}\\&>\frac{1}{2}\delta_{s}^{\frac{1}{100}}.
\end{align*}
Then  we have
$$\operatorname{det}S_{s+1}(z)\stackrel{\delta_{s}}{\sim}(z-\theta_{s})+r_{s+1}(z),$$
where $r_{s+1}(z)$ is an analytic function defined  in \eqref{d1-} with $|r_{s+1}(z)|<\delta_s^7.$
Finally, by  the R\'ouche theorem, the equation $$(z-\theta_{s})+r_{s+1}(z)=0 $$ has a unique root $\theta_{s+1}$ in  the set defined by \eqref{d1-} satisfying
$$|\theta_{s+1}-\theta_s|=|r_{s+1}(\theta_{s+1})|<\delta_s^7,\ |(z-\theta_{s})+r_{s+1}(z)|\sim|z-\theta_{s+1}|.$$
Moreover $\theta_{s+1}$ is also the unique root of $\operatorname{det}M_{s+1}(z)=0$ in the set defined by \eqref{d1-}.
From $\|z+\theta_{s}\|>\frac{1}{2}\delta_{s}^{\frac{1}{100}}$  and $|\theta_{s+1}-\theta_s|<\delta_s^7$, we have
$$\|z+\theta_{s}\|\sim\|z+\theta_{s+1}\|.$$
Thus if $z$ belongs to the set defined by $\eqref{d1-}$, we have 
\begin{equation}\label{dada}
	\operatorname{det}S_{s+1}(z)\stackrel{\delta_{s}}{\sim}\|z-\theta_{s+1}\|\cdot \|z+\theta_{s+1}\|.
\end{equation}
Since $|\log\delta_{s+1}|\sim|\log\delta_s|^{c^5},$ we get $\delta_{s+1}^{\frac{1}{10^4}}<\frac{1}{2}\delta_s^{\frac{1}{10}}.$
Recalling \eqref{d1-},  
then \eqref{dada} remains  valid for $z$ satisfying
$$ \left\|z- \theta_{s+1}\right\|<\delta_{s+1}^{\frac{1}{10^4}}.$$
For $k\in Q_s^+$,  one  considers 
\begin{align*}
	M_{s+1}(z):=T(z)_{\widetilde{\Omega}_k^{s+1}-k}=\left(\cos 2 \pi(z+n\cdot  \omega)\delta_{n,n'}-E+\varepsilon \Delta\right)_{n\in \widetilde{\Omega}_k^{s+1}-k}
\end{align*}
for $z$ being in 
\begin{equation}\label{d1+}
	\{z \in \mathbb{C}:\  \left| z+\theta_{s} \right|< \delta_{s}^{\frac{1}{10}}\}.
\end{equation}
The same argument shows that  $\operatorname{det}M_{s+1}(z)=0$ has a unique root $\theta_{s+1}'$ in the set defined by \eqref{d1+}. Since $\operatorname{det}M_{s+1}(z)$ is an even function of $z,$  we get $\theta_{s+1}'=-\theta_{s+1}.$
Thus if $z$ belongs to the set defined by $\eqref{d1+}$, we also have \eqref{dada}. 
In conclusion, \eqref{dada}  is established for $z$ belonging to
$$\left\{z\in \mathbb{C}:\ \min_{\sigma=\pm1} \left\|z+\sigma \theta_{s+1}\right\|<\delta_{s+1}^{\frac{1}{10^4}}\right\},$$
which proves \eqref{sSsim} for the $(s+1)$-th step. Combining  $l_s=0$,  \eqref{keyi}--\eqref{dy} 
and  the following
$$\|\theta+k\cdot \omega \pm \theta_{s+1}\|<10\delta_{s+1}^{\frac{1}{100}},\   |\theta_{s+1}-\theta_{s}|<\delta_s^{7} \Rightarrow\|\theta+k\cdot\omega\pm\theta_{s}\|<\delta_s,$$
we get
\begin{equation*}
	\left\{k \in \mathbb{Z}^d+\frac{1}{2}\sum_{i=0}^{s}l_i:\ \min_{\sigma=\pm1}\left\|\theta+k\cdot  \omega +\sigma \theta_{s+1}\right\|<10\delta_{s+1}^{\frac{1}{100}}\right\}\subset P_{s+1},
\end{equation*}
which proves \eqref{P>} at the  $(s+1)$-th step. Finally, we want to estimate $T^{-1}_{\widetilde\Omega_k^{s+1}}.$
For  $k\in P_{s+1}$, by \eqref{Ps+11}, we obtain
$$\theta+k\cdot\omega \in \{z\in \mathbb{C}:\ \min_{\sigma=\pm1} \left\|z+\sigma \theta_s\right\|<\delta_s^{\frac{1}{10}}\},$$
{which together with  \eqref{dada} implies
\begin{align*}
	&\left|\operatorname{det}(T_{A_{k}^{s+1}}-R_{A_{k}^{s+1}} T R_{\widetilde{\Omega}_{k}^{s+1} \backslash A_{k}^{s+1}} T_{\widetilde{\Omega}_{k}^{s+1} \backslash A_{k}^{s+1}}^{-1} R_{\widetilde{\Omega}_{k}^{s+1} \backslash A_{k}^{s+1}}T R_{A_{k}^{s+1}})\right|\\
	&\ \ \ \ =|\operatorname{det}S_{s+1}(\theta+k\cdot\omega)|\\
	&\ \ \ \ \geq\frac{1}{C}\delta_{s}\left\|\theta +k\cdot  \omega-\theta_{s+1}\right\| \cdot\left\|\theta +k\cdot  \omega+\theta_{s+1}\right\|.
\end{align*}
By Cramer's rule and Hadamard's inequality, one has
\begin{align*}
	\|(T_{A_{k}^{s+1}}-R_{A_{k}^{s+1}} T R_{\widetilde{\Omega}_{k}^{s+1} \backslash A_{k}^{s+1}} T_{\widetilde{\Omega}_{k}^{s+1} \backslash A_{k}^{s+1}}^{-1} R_{\widetilde{\Omega}_{k}^{s+1} \backslash A_{k}^{s+1}}T R_{A_{k}^{s+1}})^{-1}\| \\
	<C2^s10^{2^{s}} \delta_{s}^{-1}\left\|\theta +k\cdot  \omega-\theta_{s+1}\right\|^{-1}\cdot\left\|\theta +k \cdot \omega+\theta_{s+1}\right\|^{-1}.
\end{align*}
From Schur complement argument (cf. Lemma \ref{Su}) and \eqref{sgood1}, we get
\begin{align}
		\nonumber\|T_{\widetilde{\Omega}_{k}^{s+1}}^{-1}\|&<4\left( 1+\|T_{\widetilde{\Omega}_{k}^{s+1} \backslash A_{k}^{s+1}}^{-1}\|\right) ^2\\
		\nonumber&\ \ \ \ \times\left( 1+ \|(T_{A_{k}^{s+1}}-R_{A_{k}^{s+1}} T R_{\widetilde{\Omega}_{k}^{s+1} \backslash A_{k}^{s+1}} T_{\widetilde{\Omega}_{k}^{s+1} \backslash A_{k}^{s+1}}^{-1} R_{\widetilde{\Omega}_{k}^{s+1} \backslash A_{k}^{s+1}}T R_{A_{k}^{s+1}})^{-1}\|\right) \\
		\label{111}&< \delta_{s}^{-2}\left\|\theta +k\cdot  \omega-\theta_{s+1}\right\|^{-1}\cdot\left\|\theta +k \cdot \omega+\theta_{s+1}\right\|^{-1}.
\end{align}}

{\bf STEP2: The case $({\bf C}2)_s$ Occurs: }  i.e., 
$${\rm dist}\left(\widetilde{Q}_{s}^{-}, Q_{s}^{+}\right)\leq100N_{s+1}^c.$$
Then there exist $ i_s \in Q_{s}^{+}$ and $j_s \in \widetilde{Q}_{s}^{-}$ with  $\|i_s-j_s\| \leq 100 N_{s+1}^{c}$, such that
\[\left\|\theta+i_s\cdot  \omega+\theta_{s}\right\|<\delta_{s}, \ \left\|\theta+j_s\cdot  \omega-\theta_{s}\right\|<\delta_{s}^{\frac{1}{100}}.\]
Denote \[l_s=i_s-j_s.\]
Using \eqref{2233} and \eqref{3322} yields
$$ Q_{s}^{+},\ \widetilde{Q}_{s}^{-}\subset P_s\subset\mathbb{Z}^d+\frac{1}{2}\sum_{i=0}^{s-1}l_i.$$
Thus $i_s\equiv j_s  \ ({\rm mod}\  \mathbb{Z}^d) $ and  $l_s\in \mathbb{Z}^d$.
Define\begin{equation}\label{Os+1}
	O_{s+1}=Q_s^-\cup( Q_s^+-l_s).
\end{equation}
For every $o\in O_{s+1}$, define its mirror point
\[o^{*}=o+l_s.\]
Then we have
$$O_{s+1}\subset\left\{ o\in \mathbb{Z}^d+\frac{1}{2}\sum_{i=0}^{s-1}l_i:\ \|\theta+o\cdot \omega-\theta_s\|<2\delta_s^\frac{1}{100} \right\}$$ and
$$O_{s+1}+l_s\subset \left\{ o^*\in \mathbb{Z}^d+\frac{1}{2}\sum_{i=0}^{s-1}l_i:\ \|\theta+o^*\cdot \omega+\theta_s\|<2\delta_s^\frac{1}{100} \right\}.$$
Then by \eqref{P>}, we obtain \begin{equation}\label{keke}
	O_{s+1}\cup (O_{s+1}+l_s)\subset P_s.
\end{equation}
Define \begin{equation}\label{P_{s+1}}
	P_{s+1}=\left\{ \frac{1}{2}( o+o^*):\  o\in O_{s+1}\right\}=\left\{o+\frac{l_s}{2}:\  o\in O_{s+1}\right\}.
\end{equation}
Notice that
\begin{align*}
	&\min \left(\|\frac{l_s}{2}\cdot \omega+\theta_{s}\|,\|\frac{l_s}{2}\cdot \omega+\theta_{s}-\frac{1}{2}\|\right) =
	\frac{1}{2}\left\|l_s\cdot  \omega+2 \theta_{s}\right\|\\ 
	&\ \ \ \ \leq \frac{1}{2}(\left\|\theta+i_s\cdot  \omega+\theta_{s}\right\|+\left\|\theta+j_s\cdot \omega-\theta_{s}\right\| )< \delta_{s}^{\frac{1}{100}}.
\end{align*}
Since $\delta_s\ll 1$,  only one of  the following 		
$$\|\frac{l_s}{2}\cdot \omega+\theta_{s}\|< \delta_{s}^{\frac{1}{100}},\  \|\frac{l_s}{2}\cdot \omega+\theta_{s}-\frac{1}{2}\|< \delta_{s}^{\frac{1}{100}}$$
occurs. First, we consider the case
\begin{equation}\label{0}
	\|\frac{l_s}{2}\cdot \omega+\theta_{s}\|< \delta_{s}^{\frac{1}{100}}.
\end{equation}
Let $k\in P_{s+1}$. Since $k=o+\frac{l_s}{2}$ for some $o\in O_{s+1}$ and \eqref{0}, we get 
\[\|\theta+k\cdot \omega\| \leq\|\theta+o\cdot \omega-\theta_{s}\|+	\|\frac{l_s}{2}\cdot \omega+\theta_{s}\|<3\delta_{s}^{\frac{1}{100}},\]
which implies  \begin{equation}\label{wewe}
	P_{s+1}\subset \left\{ k\in \mathbb{Z}^d+\frac{1}{2}\sum_{i=0}^{s}l_i:\ \|\theta+k\cdot \omega\|<3\delta_s^\frac{1}{100} \right\}.
\end{equation}
Moreover, if $k\neq k' \in P_{s+1}$, then
$$\|k-k'\|>\left|\log \frac{\gamma}{6\delta_s^\frac{1}{100}}\right|\sim N_{s+1}^{c^5}\gg10N_{s+1}^{c^2}.$$
Similar to the proof  appeared in   {\bf STEP1} (i.e.,  the $({\bf C}1)_s$ case), we can associate $k\in P_{s+1}$ the blocks $\Omega_k^{s+1}$ and  $\widetilde{\Omega}_k^{s+1}$ with
\begin{align*}
\Lambda_{100N_{s+1}^c}(k)&\subset \Omega_{k}^{s+1}\subset \Lambda_{100N_{s+1}^{c}+50N_{s}^{c^2}}(k),\\
\Lambda_{N_{s+1}^{c^2}}(k)&\subset \widetilde{\Omega}_k^{s+1}\subset \Lambda_{N_{s+1}^{c^2}+50N_{s}^{c^2}}(k)\end{align*}
satisfying
\begin{equation}\label{xsxs}
	\left\{\begin{aligned}
		&\Omega_{k}^{s+1} \cap \widetilde{\Omega}_{k'}^{s'} \neq \emptyset\ \left(s^{\prime}<s+1\right) \Rightarrow \widetilde{\Omega}_{k'}^{s'} \subset \Omega_{k}^{s+1}, \\
		&\widetilde{\Omega}_{k}^{s+1} \cap \widetilde{\Omega}_{k'}^{s'} \neq \emptyset\ \left(s^{\prime}<s+1\right) \Rightarrow \widetilde{\Omega}_{k'}^{s'} \subset \widetilde{\Omega}_{k}^{s+1}, \\
		&{\rm dist}(\widetilde{\Omega}_{k}^{s+1}, \widetilde{\Omega}_{k^{\prime}}^{s+1})>10\operatorname{diam}\widetilde{\Omega}_{k}^{s+1}  \ { \rm for } \ k \neq k^{\prime} \in P_{s+1} .
	\end{aligned}\right.
\end{equation}
In addition, the set
$$\widetilde{\Omega}_k^{s+1}-k\subset \mathbb{Z}^d+\frac{1}{2}\sum_{i=0}^{s}l_i$$ is independent  of $k\in P_{s+1}$ and  symmetrical about the origin.
Clearly, in this construction,  for every $k' \in Q_{s}$, there exists $k=k'-\frac{l_s}{2}\text{ or } k'+\frac{l_s}{2}\in P_{s+1}$, such that
$$\widetilde{\Omega}_{k'}^{s} \subset \Omega_k^{s+1}.$$

For every $k\in P_{s+1}$, we have $o, o^*\in P_s$ since \eqref{keke}.
Define
$$A_k^{s+1}=A_o^{s}\cup A_{o^*}^{s} ,$$
where $o\in O_{s+1}$ and $k=o+o^*$ (cf. \eqref{P_{s+1}}).
Then
 \begin{align*}
 A_k^{s+1}&\subset \Omega_o^s\cup \Omega_{o^*}^s \subset\Omega_k^{s+1},\\
\# A_k^{s+1}&=\# A_o^{s}+\# A_{o^*}^{s}\leq2^{s+1}.
\end{align*}
Now we will verify that
$(\widetilde{\Omega}_{k}^{s+1} \setminus A_k^{s+1})$ is $s$-\textbf{good}, i.e., 
$$\left\{\begin{aligned}
	&l'\in Q_{s'},\widetilde{\Omega}_{l'}^{s'}\subset(\widetilde{\Omega}_{k}^{s+1} \setminus A_k^{s+1}),\widetilde{\Omega}_{l'}^{s'}\subset \Omega_{l}^{s'+1} \Rightarrow \widetilde{\Omega}_{l}^{s'+1} \subset(\widetilde{\Omega}_{k}^{s+1} \setminus A_k^{s+1})\ \text{for }s'<s,\\
	&\left\{l\in P_s :\   \widetilde{\Omega}_{l}^s\subset (\widetilde{\Omega}_{k}^{s+1} \setminus A_k^{s+1}) \right\}\cap Q_{s}= \emptyset.
\end{aligned}\right.
$$
For this purpose, assume that  $$l'\in Q_{s'},\ \widetilde{\Omega}_{l'}^{s'}\subset(\widetilde{\Omega}_{k}^{s+1}\setminus A_k^{s+1}),\ \widetilde{\Omega}_{l'}^{s'}\subset \Omega_{l}^{s'+1}. $$
If $s'\leq s-2$, since $\emptyset\neq\widetilde{\Omega}_{l'}^{s'}\subset\widetilde{\Omega}_{l}^{s'+1}\cap\widetilde{\Omega}_{k}^{s+1}
,$ we have by \eqref{xsxs}
$ \widetilde{\Omega}_{l}^{s'+1}\subset \widetilde{\Omega}_{k}^{s+1}.$
If $\widetilde{\Omega}_{l}^{s'+1}\cap A_k^{s+1}\neq \emptyset,$
then we have
$ \widetilde{\Omega}_{l}^{s'+1}\cap A_o^{s}\neq \emptyset\text{ or }\widetilde{\Omega}_{l}^{s'+1}\cap A_{o^*}^{s}\neq\emptyset.$
Thus by \eqref{haha} $ (s'+1< s)$, we get
$\widetilde{\Omega}_{l'}^{s'+1}\subset \widetilde{\Omega}_{o}^{s}\text{ or }\widetilde{\Omega}_{l'}^{s'+1}\subset \widetilde{\Omega}_{o^*}^{s},$
which implies
$\widetilde{\Omega}_{l'}^{s'}\subset(\widetilde{\Omega}_{o}^{s}\setminus A_o^{s})\text{ or } \widetilde{\Omega}_{l'}^{s'}\subset(\widetilde{\Omega}_{o^*}^{s}\setminus A_{o^*}^{s}).$
Thus we have either
$\widetilde{\Omega}_{l'}^{s'+1}\subset(\widetilde{\Omega}_{o}^{s}\setminus A_o^{s})\subset(\widetilde{\Omega}_{k}^{s+1} \setminus A_k^{s+1})$
or
$\widetilde{\Omega}_{l'}^{s'+1}\subset(\widetilde{\Omega}_{o^*}^{s}\setminus A_{o^*}^{s})\subset(\widetilde{\Omega}_{k}^{s+1} \setminus A_k^{s+1})$
since  both $(\widetilde{\Omega}_{o}^{s}\setminus A_o^{s})$  and $(\widetilde{\Omega}_{o^*}^{s}\setminus A_{o^*}^{s})$ are  $(s-1)$-\textbf{good}. This is a contradiction.
If  $s'=s-1$, $\widetilde{\Omega}_{l'}^{s-1}\subset \Omega_{l}^{s}$ and $\widetilde{\Omega}_{l}^{s}\cap A_k^{s+1} \neq\emptyset$, then either  $l=o$ or $l=o^*$, thus  $\widetilde{\Omega}_{l'}^{s-1}\subset (\widetilde{\Omega}_{o}^{s}\setminus A_o^s)\text{ or }\widetilde{\Omega}_{l'}^{s-1}\subset (\widetilde{\Omega}_{o^*}^{s}\setminus A_{o^*}^s).$
This  contradicts 
$$\left\{l\in P_{s-1} :\   \widetilde{\Omega}_{l}^{s-1}\subset(\widetilde{\Omega}_{o}^{s}\setminus A_o^s)  \right\}\cap Q_{s-1}= \left\{l\in P_{s-1} :\   \widetilde{\Omega}_{l}^{s-1}\subset(\widetilde{\Omega}_{o^*}^{s}\setminus A_{o^*}^s) \right\}\cap Q_{s-1}= \emptyset$$
since both $(\widetilde{\Omega}_{o}^{s}\setminus A_o^s)$ and $(\widetilde{\Omega}_{o^*}^{s}\setminus A_{o^*}^s)$  are  $(s-1)$-\textbf{good}.
Finally, if  $l\in Q_s$  and  $ \widetilde{\Omega}_{l}^s\subset \widetilde{\Omega}_{k}^{s+1},$
then  $l=o$ or $l=o^*$. Thus
$\widetilde{\Omega}_{l}^s\nsubseteq (\widetilde{\Omega}_{k}^{s+1} \setminus A_k^{s+1}),$
which implies  $$\left\{l\in P_s :\   \widetilde{\Omega}_{l}^s\subset (\widetilde{\Omega}_{k}^{s+1} \setminus A_k^{s+1}) \right\}\cap Q_{s}= \emptyset.$$
Moreover,  we have 
\begin{align*}
A_k^{s+1}-k&=(A_o^s-k)\cup (A_{o^*}^s-k)\\
&=\left( (A_o^s-o)-\frac{l_s}{2}\right) \cup\left( (A_{o^*}^s-o^*)+\frac{l_s}{2}\right) \end{align*}
is independent of $k\in P_{s+1}$ and symmetrical about  the origin.

Now consider  the  analytic matrix-valued function
$$M_{s+1}(z):=T(z)_{\widetilde{\Omega}_{k}^{s+1}-k}=\left(\cos 2 \pi(z+n \cdot\omega)\delta_{n,n'}-E+\varepsilon \Delta\right)_{n\in \widetilde{\Omega}_{k}^{s+1}-k}$$
defined in \begin{equation}\label{1000}
	\{z\in \mathbb{C}:\  \left|z\right|<\delta_s^{\frac{1}{10^3}}\}.
\end{equation}
If  $k'\in P_{s}$ and  $\widetilde{\Omega}_{k'}^{s}\subset (\widetilde{\Omega}_{k}^{s+1} \setminus A_k^{s+1})$, then $k'\neq o, o^*$ and $\|k'-o\|,\|k'-o^*\|\leq4N_{s+1}^{c^2}$.
Thus
\begin{align*}
	\|\theta+k'\cdot \omega-\theta_{s}\|&\geq \|(k'-o)\cdot \omega\|-\|\theta+o\cdot \omega-\theta_{s}\|\\
	&\geq \gamma e^{-(4N_{s+1}^{c^2})^{\tau}}-2\delta_s^\frac{1}{100}\\
	&\geq \gamma e^{-4^{\tau}|\log\frac{\gamma}{\delta_{s}}|^{\frac{1}{c}}}-2\delta_s^\frac{1}{100}\\
	&>\delta_s^{\frac{1}{10^4}},
\end{align*}
and
\begin{align*}
	\|\theta+k'\cdot \omega+\theta_{s}\|&\geq \|(k'-o^*)\cdot \omega\|-\|\theta+o^*\cdot \omega+\theta_{s}\|\\
	&\geq \gamma e^{-(4N_{s+1}^{c^2})^{\tau}}-2\delta_s^\frac{1}{100}\\
	&\geq \gamma e^{-4^{\tau}|\log\frac{\gamma}{\delta_{s}}|^{\frac{1}{c}}}-2\delta_s^\frac{1}{100}\\
	&>\delta_s^{\frac{1}{10^4}}.
\end{align*}
By \eqref{L2}, we have
\begin{align}
	\nonumber\|T_{\widetilde{\Omega}_{k}^{s+1} \setminus A_k^{s+1}}^{-1}\|&<\delta^{-3}_{s-1}\sup_{\left\{k'\in P_s :\   \widetilde{\Omega}_{k'}^s\subset (\widetilde{\Omega}_{k}^{s+1} \setminus A_k^{s+1}) \right\}}\|\theta+k'\cdot\omega-\theta_{s}\|^{-1}\cdot \|\theta+k'\cdot\omega+\theta_{s}\|^{-1}\\
	\label{sgood2}&<\frac{1}{2}\delta_s^{-3\times \frac{1}{10^4}}.
\end{align}
One may restate \eqref{sgood2} as
$$	\left\|\left( M_{s+1}(\theta+k\cdot\omega)_{(\widetilde{\Omega}_{k}^{s+1} \setminus A_k^{s+1})-k}\right) ^{-1}\right\|<\frac{1}{2}\delta_s^{-3\times \frac{1}{10^4}}.$$
Since \begin{align}
	\nonumber\|z-(\theta+k\cdot\omega)\|&\leq|z|+\|\theta+k\cdot\omega\|\\
	\label{ff}&<\delta_{s}^{\frac{1}{10^3}}+3\delta_{s}^\frac{1}{100} <2\delta_s^{\frac{1}{10^3}}\ll\delta_s^{3\times\frac{1}{10^4}},
\end{align}
we obtain using Neumann series argument
\begin{equation}\label{808}
	\left\|\left( M_{s+1}(z)_{(\widetilde{\Omega}_{k}^{s+1} \setminus A_k^{s+1})-k}\right) ^{-1}\right\|<\delta_s^{-3\times \frac{1}{10^4}}.
\end{equation}
We may   control $M_{s+1}(z)^{-1}$ by the inverse of
\begin{align*}
	S_{s+1}(z)&=M_{s+1}(z)_{ A_k^{s+1}-k}-R_{A_k^{s+1}-k} M_{s+1}(z) R_{(\widetilde{\Omega}_{k}^{s+1} \setminus A_k^{s+1})-k}\\ 
	&\ \ \ \  \times \left( M_{s+1}(z)_{(\widetilde{\Omega}_{k}^{s+1} \setminus A_k^{s+1})-k}\right) ^{-1} R_{(\widetilde{\Omega}_{k}^{s+1} \setminus A_k^{s+1})-k}  M_{s+1}(z)R_{A_k^{s+1}-k}.
\end{align*}
Our next aim is to analyze   $\operatorname{det}S_{s+1}(z).$
Since $A_k^{s+1}-k=(A_o^{s}-k)\cup(A_{o^*}^{s}-k)$,  $A_o^{s}-k \subset \Omega_o^{s}-k,\  A_{o^*}^{s}-k\subset  \Omega_{o^*}^{s}-k$
and $${\rm dist}(\Omega_o^{s}-k, \Omega_{o^*}^{s}-k)>10\operatorname{diam}\widetilde{\Omega}_o^{s},$$
we have
$$M_{s+1}(z)_{ A_k^{s+1}-k}=M_{s+1}(z)_{A_o^{s}-k}\oplus M_{s+1}(z)_{A_{o^*}^{s}-k}.$$
From  ${\rm dist}(\Omega_o^{s},\partial^+\widetilde{\Omega}_{o}^{s})$  and ${\rm dist}(\Omega_{o^*}^{s},\partial^+\widetilde{\Omega}_{o^*}^{s})>1$,  we have
\begin{align*}
&R_{A_o^{s}-k} M_{s+1}(z) R_{(\widetilde{\Omega}_{k}^{s+1} \setminus A_k^{s+1})-k}=R_{A_o^{s}-k} M_{s+1}(z) R_{(\widetilde{\Omega}_{o}^{s} \setminus A_o^{s})-k},\\
&R_{A_{o^*}^{s}-k} M_{s+1}(z) R_{(\widetilde{\Omega}_{k}^{s+1} \setminus A_k^{s+1})-k}=R_{A_{o^*}^{s}-k} M_{s+1}(z) R_{(\widetilde{\Omega}_{o^*}^{s} \setminus A_{o^*}^{s})-k}.
\end{align*}
Denote $$X=(\widetilde{\Omega}_{o}^{s} \setminus A_{o}^{s})-k, \ X^*=(\widetilde{\Omega}_{o^*}^{s} \setminus A_{o^*}^{s})-k,\ Y=(\widetilde{\Omega}_{k}^{s+1} \setminus A_k^{s+1})-k.$$  Then	direct computations yield
{\small\begin{align}\label{chang}
		\begin{split}
			S_{s+1}(z)&=M_{s+1}(z)_{A_o^{s}-k}\oplus M_{s+1}(z)_{A_{o^*}^{s}-k}-(R_{A_o^{s}-k}\oplus R_{A_{o^*}^{s}-k}) M_{s+1}(z) R_{Y} M_{s+1}(z)_{Y}^{-1} R_{Y}  M_{s+1}(z)R_{A_k^{s+1}-k}  \\
			&=\left( M_{s+1}(z)_{A_o^{s}-k}-R_{A_o^{s}-k} M_{s+1}(z) R_{X} M_{s+1}(z)_{Y} ^{-1} R_{Y}  M_{s+1}(z)R_{A_k^{s+1}-k}\right) \\
			&\ \ \ \ \oplus\left( M_{s+1}(z)_{A_{o^*}^{s}-k}-R_{A_{o^*}^{s}-k} M_{s+1}(z) R_{X^*}M_{s+1}(z)_{Y}^{-1} R_{Y}  M_{s+1}(z)R_{A_k^{s+1}-k}\right).
		\end{split}
\end{align}}
Since $\widetilde{\Omega}_o^s\setminus A_o^s$ is $(s-1)$-\textbf{good},  we have by \eqref{L2}--\eqref{exp}
\begin{align*}
\|{T}_{\widetilde{\Omega}_o^s\setminus A_o^s}^{-1}\|&<\delta_{s-1}^{-3},\\
	\left|{T}_{\widetilde{\Omega}_o^s\setminus A_o^s}^{-1}(x, y)\right|&<e^{-\gamma_{s-1}\|x-y\|_1} \  \text {for }\|x-y\|>N_{s-1}^{c^3}.
\end{align*}
In other words, 
\begin{align}
	\label{ed}\left\|\left( M_{s+1}(\theta+k\cdot \omega)_{X}\right) ^{-1}\right\|&<\delta_{s-1}^{-3},\\
	\label{ad}\left|\left( M_{s+1}(\theta+k\cdot \omega)_{X}\right) ^{-1}(x,y) \right| &<e^{-\gamma_{s-1}\|x-y\|_1}\  \text { for }\|x-y\|>N_{s-1}^{c^3}.
\end{align}
From the approximation  \eqref{ff}, we deduce by the same argument as \eqref{jidan} that 
\begin{equation}\label{yadan}
	\left|\left( M_{s+1}(z)_{(\widetilde{\Omega}_{k}^{s} \setminus A_k^{s})-k}\right) ^{-1}(x,y) \right| <\delta_s^{10}\  \text { for }\|x-y\|>N_{s-1}^{c^4}.
\end{equation}
Let $x\in X$ and  ${\rm dist}(x,A^{s}_o-k)\leq1$. By resolvent identity, we have for any  $y \in Y$,
\begin{align}\label{33}
	\nonumber&\left(M_{s+1}(z)_{Y}\right) ^{-1}(x,y)-\chi_{X}(y)\left( M_{s+1}(z)_{X}\right) ^{-1}(x,y) \\
	&\ \ \ \ = -\sum _{(w,w') \in\partial_{Y} X} \left( M_{s+1}(z)_{X}\right)^{-1}(x,w)\Gamma(w,w')\left( M_{s+1}(z)_{Y}\right) ^{-1}(w',y).
\end{align}
From
$${\rm dist}(x,w)\geq {\rm dist}(A_o^s-k,\partial^- \widetilde{\Omega}_o^s-k )-2>N_s> N_{s-1}^{c^4}, $$
\eqref{808}  and \eqref{yadan}, the RHS of \eqref{33} is bounded by
$CN_s^{c^2d}\delta_s^{-\frac{3}{10^4}}\delta_s^{10}<\delta_s^9.$
It follows that   \begin{align*}
	&R_{A_o^{s}-k} M_{s+1}(z) R_{X}  \left( M_{s+1}(z)_{Y}\right) ^{-1}= R_{A_o^{s}-k} M_{s+1}(z) R_{X}  \left( M_{s+1}(z)_{X}\right) ^{-1} R_{X}+O(\delta_s^9).
\end{align*}
Similarly,  \begin{align*}
	&R_{A_{o^*}^{s}-k} M_{s+1}(z) R_{X^*}  \left( M_{s+1}(z)_{Y}\right) ^{-1} = R_{A_{o^*}^{s}-k} M_{s+1}(z) R_{X^*}  \left( M_{s+1}(z)_{X^*}\right) ^{-1} R_{X^*}+O(\delta_s^9).
\end{align*}
Recalling \eqref{chang}, we get  {\footnotesize \begin{align}\label{xixi}
		S_{s+1}(z)
		&=\left( M_{s+1}(z)_{A_o^{s}-k}-R_{A_o^{s}-k} M_{s+1}(z) R_{X}\left( M_{s+1}(z)_{X}\right) ^{-1} R_{(\widetilde{\Omega}_{o}^{s} \setminus A_o^{s})-k}  M_{s+1}(z)R_{A_o^{s}-k}\right) \nonumber\\
		&\ \ \ \  \oplus\left( M_{s+1}(z)_{A_{o^*}^{s}-k}-R_{A_{o^*}^{s}-k} M_{s+1}(z) R_{X^*}\left( M_{s+1}(z)_{X^*}\right) ^{-1} R_{X^*}  M_{s+1}(z)R_{A_{o^*}^{s}-k}\right)+O(\delta_s^9) \nonumber\\
		&=S_s(z-\frac{l_s}{2}\cdot \omega)\oplus S_s(z+\frac{l_s}{2}\cdot \omega)+O(\delta_s^9).
\end{align}}
From \eqref{0} and \eqref{1000}, we have $$\|z-\frac{l_s}{2}\cdot \omega-\theta_s\|\leq|z|+\|\frac{l_s}{2}\cdot \omega+\theta_s\|<\delta_s^{\frac{1}{10^3}}+\delta_s^\frac{1}{100}<\delta_s^{\frac{1}{10^4}},$$
$$\|z+\frac{l_s}{2}\cdot \omega+\theta_s\|<|z|+\|\frac{l_s}{2}\cdot \omega+\theta_s\|<\delta_s^{\frac{1}{10^3}}+\delta_s^\frac{1}{100}<\delta_s^{\frac{1}{10^4}}.$$
Thus  both $z-\frac{l_s}{2}\cdot \omega$ and $z+\frac{l_s}{2}\cdot \omega$ belong to the set defined by  \eqref{dede},
which together with \eqref{sSsim} implies
\begin{align}
	\label{kx}&\operatorname{det}S_s(z-\frac{l_s}{2}\cdot \omega)\stackrel{\delta_{s-1}}{\sim}\|(z-\frac{l_s}{2}\cdot \omega)-\theta_{s}\|\cdot \|(z-\frac{l_s}{2}\cdot \omega)+\theta_{s}\|,\\
\label{xk}&\operatorname{det}S_s(z+\frac{l_s}{2}\cdot \omega)\stackrel{\delta_{s-1}}{\sim}\|(z+\frac{l_s}{2}\cdot \omega)-\theta_{s}\|\cdot \|(z+\frac{l_s}{2}\cdot \omega)+\theta_{s}\|.
\end{align}
Moreover, \begin{align}\label{le}
	\begin{split}
		\operatorname{det}S_{s+1}(z)&=\operatorname{det}S_s(z-\frac{l_s}{2}\omega)\cdot \operatorname{det}S_s(z+\frac{l_s}{2}\omega)+O((2^{s+1})^210^{2^{s+1}}\delta_s^9)\\
		&=\operatorname{det}S_s(z-\frac{l_s}{2}\omega)\cdot \operatorname{det}S_s(z+\frac{l_s}{2}\omega)+O(\delta_s^8)
	\end{split}	
\end{align}
since $\#( A^{s+1}_k-k)\leq2^{s+1}$, \eqref{10} and $\log\log|\log \delta_s|\sim s.$
Notice that   	\begin{align}
	\begin{split}\label{la}
		\|z+\frac{l_s}{2} \cdot \omega-\theta_{s}\|&\geq\|l_s\cdot \omega\|-\|z-\frac{l_s}{2}\cdot \omega-\theta_s\|\\
		&>\gamma e^{-(100N_s^c)^\tau}-\delta_{s}^{ \frac{1}{10^{4}}}\\
		&>\delta_{s}^{ \frac{1}{10^{4}}},
	\end{split}	
\end{align}
\begin{align}\label{li}
	\begin{split}
		\|z-\frac{l_s}{2} \cdot \omega+\theta_{s}\|&\geq\|l_s\cdot \omega\|-\|z+\frac{l_s}{2}\cdot \omega+\theta_s\|\\
		&>\gamma e^{-(100N_s^c)^\tau}-\delta_{s}^{ \frac{1}{10^{4}}}\\
		&>\delta_{s}^{ \frac{1}{10^{4}}}.
	\end{split}	
\end{align}
Let $z_{s+1}$  satisfy
\begin{equation}\label{ap}
	z_{s+1} \equiv \frac{l_s}{2} \cdot \omega+\theta_{s}\ (\bmod \ \mathbb{Z}), \ \left|z_{s+1}\right|=\| \frac{l_s}{2}\cdot  \omega+\theta_{s}\| < \delta_{s}^{\frac{1}{100}}.
\end{equation}
From \eqref{kx}--\eqref{li}, we get
$$\operatorname{det}S_{s+1}(z)\stackrel{\delta_{s}}{\sim}(z-z_{s+1})\cdot (z+z_{s+1})+r_{s+1}(z),$$
where $r_{s+1}(z)$ is an analytic function in the set defined by  \eqref{1000} with $|r_{s+1}(z)|<\delta_s^7.$	
By R\'ouche theorem, the equation
$$\left(z-z_{s+1}\right)\left(z+z_{s+1}\right)+r_{s+1}(z)=0$$
has exactly two roots $\theta_{s+1}$, $\theta_{s+1}'$ in  the set defined by \eqref{1000}, which are perturbations of $\pm z_{s+1}$, respectively.
Notice  that \[\left\{|z|< \delta_{s}^{\frac{1}{10^3}}:\ \operatorname{det} M_{s+1}(z)=0\right\}=\left\{|z|< \delta_{s}^{\frac{1}{10^3}} :\  \operatorname{det}S_{s+1}(z)=0\right\} \]
and $\operatorname{det}M_{s+1}(z)$ is an even function of $z$. Thus $$\theta_{s+1}'=-\theta_{s+1}.$$
Moreover, we get \begin{equation}\label{app}
	|\theta_{s+1}-z_{s+1}|\leq|r_{s+1}(\theta_{s+1})|^{\frac{1}{2}}<\delta_s^3 \end{equation}   and  \[|\left(z-z_{s+1}\right)\left(z+z_{s+1}\right)+r_{s+1}(z)|\sim |\left(z-\theta_{s+1}\right)\left(z+\theta_{s+1}\right)|.\]
Thus for $z$ being in the set defined by \eqref{1000}, we have
\begin{equation}\label{nana}
	\operatorname{det}S_{s+1}(z)\stackrel{\delta_{s}}{\sim}\|z-\theta_{s+1}\|\cdot \|z+\theta_{s+1}\|.
\end{equation}
Since $\delta_{s+1}^{\frac{1}{10^4}}<\frac{1}{2}\delta_s^{\frac{1}{10^3}},$  by combining \eqref{ap} and \eqref{app}, we get   $$\{z\in \mathbb{C}:\  \min_{\sigma=\pm1}\left|z+\sigma \theta_{s+1}\right|<\delta_{s+1}^{\frac{1}{10^4}}\}\subset\{z\in \mathbb{C}:\  \left|z\right|<\delta_s^{\frac{1}{10^3}}\}.$$
Hence \eqref{nana} also holds true  for $z$ belonging to
$$\{z\in \mathbb{C}:\  \left\|z\pm \theta_{s+1}\right\|<\delta_{s+1}^{\frac{1}{10^4}}\},$$
which proves \eqref{sSsim} for the  $(s+1)$-th step. 

Notice that
$$\|\theta+k\cdot \omega + \theta_{s+1}\|<10\delta_{s+1}^{\frac{1}{100}},\   |\theta_{s+1}-z_{s+1}|<\delta_s^{3} \Rightarrow\|\theta+k\cdot\omega+\frac{l_s}{2}+\theta_{s}\|<\delta_s.$$
Thus  if
\begin{equation*}
	k \in \mathbb{Z}^d+\frac{1}{2}\sum_{i=0}^{s}l_i\ {\rm and }\ \left\|\theta+k\cdot  \omega + \theta_{s+1}\right\|<10\delta_{s+1}^{\frac{1}{100}},
\end{equation*}
then
\begin{equation*}
	k+\frac{l_s}{2} \in \mathbb{Z}^d+\frac{1}{2}\sum_{i=0}^{s-1}l_i\ {\rm and}\  \|\theta+(k+\frac{l_s}{2} )\cdot  \omega + \theta_{s}\|<\delta_{s}.
\end{equation*}
Thus by \eqref{keyi},  we have $k+\frac{l_s}{2}\in Q_s^+$.
Recalling also \eqref{Os+1} and \eqref{P_{s+1}}, we have $k\in P_{s+1}.$ Thus  
\begin{equation*}
	\left\{k \in \mathbb{Z}^d+\frac{1}{2}\sum_{i=0}^{s}l_i:\ \left\|\theta+k\cdot  \omega + \theta_{s+1}\right\|<10\delta_{s+1}^{\frac{1}{100}}\right\}\subset P_{s+1}.
\end{equation*}
Similarly,
\begin{equation*}
	\left\{k \in \mathbb{Z}^d+\frac{1}{2}\sum_{i=0}^{s}l_i:\ \left\|\theta+k\cdot  \omega - \theta_{s+1}\right\|<10\delta_{s+1}^{\frac{1}{100}}\right\}\subset P_{s+1}.
\end{equation*}
Hence we  prove \eqref{P>} for the $(s+1)$-th step.

Finally, we will estimate $T^{-1}_{\widetilde\Omega_k^{s+1}}$.  For  $k\in P_{s+1}$,  we have by  \eqref{wewe}
$$\theta+k\cdot\omega \in \left\{z\in \mathbb{C}:\  \left\|z\right\|<\delta_s^{\frac{1}{10^3}}\right\}.$$
Thus   from \eqref{nana}, we obtain
\begin{align*}
	&\left|\operatorname{det}(T_{A_{k}^{s+1}}-R_{A_{k}^{s+1}} T R_{\widetilde{\Omega}_{k}^{s+1} \backslash A_{k}^{s+1}} T_{\widetilde{\Omega}_{k}^{s+1} \backslash A_{k}^{s+1}}^{-1} R_{\widetilde{\Omega}_{k}^{s+1} \backslash A_{k}^{s+1}}T R_{A_{k}^{s+1}})\right|\\
	&\ \ \ \ =|\operatorname{det}S_{s+1}(\theta+k\cdot\omega)|\\
	&\ \ \ \  \geq\frac{1}{C}\delta_{s}\left\|\theta +k\cdot  \omega-\theta_{s+1}\right\| \cdot\left\|\theta +k\cdot  \omega+\theta_{s+1}\right\|.
\end{align*}
Using  Cramer's rule and Hadamard's inequality implies
\begin{align*}
	&\|(T_{A_{k}^{s+1}}-R_{A_{k}^{s+1}} T R_{\widetilde{\Omega}_{k}^{s+1} \backslash A_{k}^{s+1}} T_{\widetilde{\Omega}_{k}^{s+1} \backslash A_{k}^{s+1}}^{-1} R_{\widetilde{\Omega}_{k}^{s+1} \backslash A_{k}^{s+1}}T R_{A_{k}^{s+1}})^{-1}\| \\
	&\ \ \ \ <C2^{s+1}10^{2^{s+1}} \delta_{s}^{-1}\left\|\theta +k\cdot  \omega-\theta_{s+1}\right\|^{-1}\cdot\left\|\theta +k \cdot \omega+\theta_{s+1}\right\|^{-1}.
\end{align*}
Recalling  Schur complement argument (cf. Lemma \ref{Su}) and \eqref{sgood2}, we get
\begin{align}\label{222}
	\begin{split}
		\|T_{\widetilde{\Omega}_{k}^{s+1}}^{-1}\|&<4\left( 1+\|T_{\widetilde{\Omega}_{k}^{s+1} \backslash A_{k}^{s+1}}^{-1}\|\right)^2\\
		&\ \ \ \  \times \left( 1+\|(T_{A_{k}^{s+1}}-R_{A_{k}^{s+1}} T R_{\widetilde{\Omega}_{k}^{s+1} \backslash A_{k}^{s+1}} T_{\widetilde{\Omega}_{k}^{s+1} \backslash A_{k}^{s+1}}^{-1} R_{\widetilde{\Omega}_{k}^{s+1} \backslash A_{k}^{s+1}}T R_{A_{k}^{s+1}})^{-1}\|\right) \\
		&< \delta_{s}^{-2}\left\|\theta +k\cdot  \omega-\theta_{s+1}\right\|^{-1}\cdot\left\|\theta +k \cdot \omega+\theta_{s+1}\right\|^{-1}.
	\end{split}	
\end{align}

For the case \begin{equation}\label{1}
	\left\|\frac{l_s}{2} \cdot \omega+\theta_{s}-\frac{1}{2}\right\|<\delta_{s}^{\frac{1}{100}},
\end{equation}
we have
\begin{equation}\label{Ps2g2}
	P_{s+1}\subset \left\{ k\in \mathbb{Z}^d+\frac{1}{2}\sum_{i=0}^{s}l_i:\ \|\theta+k\cdot \omega-\frac{1}{2}\|<3\delta_s^\frac{1}{100} \right\}.
\end{equation}
Thus we can consider  
\begin{align*}
	M_{s+1}(z):=T(z)_{\widetilde{\Omega}_k^1-k}=\left(\cos 2 \pi(z+n\cdot  \omega)\delta_{n,n'}-E+\varepsilon \Delta\right)_{n\in \widetilde{\Omega}_k^1-k}
\end{align*}
in \begin{equation}\label{hei}
	\left\{z \in \mathbb{C}:\  | z-\frac{1}{2}|< \delta_{s}^{\frac{1}{10^3}}\right\}.
\end{equation}
By the similar arguments as above, we obtain
both $\theta_{s+1}$  and $1-\theta_{s+1}$ belong to the set defined by \eqref{hei}. Moreover,   all the corresponding conclusions in the case  of \eqref{0} hold for the case \eqref{1}.
Recalling \eqref{111}, estimate  \eqref{222} holds for  the case \eqref{1} as well.

{\bf STEP3: Application of resolvent identity}

Finally, we aim to establish $({\bf e})_{s+1}$ by iterating the resolvent identity. 

Recall that  
$$\left|\log \frac{\gamma}{\delta_{s+1}}\right|=\left|\log \frac{\gamma}{\delta_{s}}\right|^{c^5}.$$ Define  $$Q_{s+1}=\left\{k \in P_{s+1} :\ \min_{\sigma=\pm1}\left\|\theta+k\cdot \omega +\sigma \theta_{s+1}\right\|<\delta_{s+1}\right\}. $$
Assume the finite set $\Lambda\subset \mathbb{Z}^d$ is $(s+1)$-\textbf{good}, i.e., 
\begin{equation}\label{s+1g}
	\left\{\begin{aligned}
		&k'\in Q_{s'},\widetilde{\Omega}_{k'}^{s'}\subset\Lambda,\widetilde{\Omega}_{k'}^{s'}\subset \Omega_{k}^{s'+1} \Rightarrow \widetilde{\Omega}_k^{s'+1} \subset\Lambda\ \text{ for }s'<s+1,\\
		&\{k\in P_{s+1} :\   \widetilde{\Omega}_{k}^{s+1}\subset \Lambda \}\cap Q_{s+1}= \emptyset.
	\end{aligned}\right.
\end{equation}
It remains to verify the implications \eqref{L2} and \eqref{exp}  with $s$ being replaced with $s+1.$

For $k\in P_t\ (1\leq t\leq s+1)$,  denote  by  $$2\Omega_{k}^{t}:=\Lambda_{\operatorname{diam}\Omega_{k}^{t}}(k)$$   the ``double''-size  block  of  $\Omega^{t}_k$. Define moreover
\begin{equation}\label{wp}
	\widetilde{P}_{t}=\{k \in P_{t} :\  \exists \ k'\in Q_{t-1}\ {\rm s.t., }\ \widetilde{\Omega}_{k'}^{t-1}\subset\Lambda,\widetilde{\Omega}_{k'}^{t-1}\subset \Omega_{k}^{t} \}\quad (1\leq t\leq s+1).
\end{equation}
\begin{lem}\label{lems+1}
	For  $k\in P_{s+1}\setminus Q_{s+1}$, we have
	\begin{align}\label{exps+1}
		|T^{-1}_{\widetilde{\Omega}_{k}^{s+1}}(x,y)|<e^{-\widetilde{\gamma}_{s}\|x-y\|_1} \  {\rm for } \ x\in \partial^-\widetilde{\Omega}_{k}^{s+1} \ {\rm and }\ \ y\in 2\Omega_{k}^{s+1},
	\end{align}
	where $\widetilde{\gamma}_s=\gamma_s(1-N_{s+1}^{\frac{1}{c}-1})$.
\end{lem}
\begin{proof}[Proof of Lemma \ref{lems+1}]
	Notice first  that $${\rm dist}( \partial^-\widetilde{\Omega}_{k}^{s+1} ,  2\Omega_{k}^{s+1})\gtrsim\operatorname{diam}\widetilde{\Omega}_{k}^{s+1}>N_{s+1}\gg N_s^{c^3}.$$	Since   $\widetilde{\Omega}_{k}^{s+1}\setminus A_k^{s+1}$ is $s$-\textbf{good},
	we have by \eqref{exp}
	$$|T^{-1}_{\widetilde{\Omega}_{k}^{s+1}\setminus A_k^{s+1}}(x,w)|<e^{-\gamma_s\|x-w\|_1} \quad \text{for }  x\in \partial^-\widetilde{\Omega}_{k}^{s+1} , \ \ w\in  (\widetilde{\Omega}_{k}^{s+1}\setminus A_k^{s+1})\cap2\Omega_{k}^{s+1} .$$
From  \eqref{222} and  $k\notin Q_{s+1}$, we obtain $$	\left\|T_{\widetilde{\Omega}_{k}^{s+1}}^{-1}\right\|<\delta_{s}^{-2}\delta_{s+1}^{-2}< \delta_{s+1}^{-3}.$$
	Using  resolvent identity implies (since $x\in \partial^-\widetilde{\Omega}_{k}^{s+1}$)
	\begin{align*}
		|T^{-1}_{\widetilde{\Omega}_{k}^{s+1}}(x,y)|&=\left|T^{-1}_{\widetilde{\Omega}_{k}^{s+1}\setminus A_k^{s+1}}(x,y)\chi_{\widetilde{\Omega}_{k}^{s+1}\setminus A_k^{s+1}}(y)-\sum_{(w', w)\in\partial A_k^{s+1}  } T^{-1}_{\widetilde{\Omega}_{k}^{s+1}\setminus A_k^{s+1}}(x,w)\Gamma(w,w')T^{-1}_{\widetilde{\Omega}_{k}^{s+1}}(w',y)\right|\\
		&< e^{-\gamma_s\|x-y\|_1}+  2d\cdot 2^{s+1} \sup_{w\in\partial^+ A_k^{s+1}}e^{-\gamma_s\|x-w\|_1}\|T^{-1}_{\widetilde{\Omega}_{k}^{s+1}}\|\\
		&<e^{-\gamma_s\|x-y\|_1}+ \sup_{w\in\partial^+ A_k^{s+1}}e^{-\gamma_s\left( \|x-y\|_1-\|y-w\|_1\right)+C|\log\delta_{s+1}|} \\\
		&<e^{-\gamma_s\|x-y\|_1}+e^{- \gamma_s\left(1-C\left( \|x-y\|_1^{\frac{1}{c}-1}+\frac{|\log\delta_{s+1}|}{\|x-y\|_1}\right) \right)\|x-y\|_1}\\
		&<e^{- \gamma_{s}\left(1-N_{s+1}^{\frac{1}{c}-1}\right)\|x-y\|_1}\\
		&=e^{-\widetilde{\gamma}_s\|x-y\|_1},
	\end{align*}
	since
	$$N_{s+1}^{c}\lesssim\operatorname{diam}\widetilde{\Omega}_k^{s+1}\sim\|x-y\|_1,  \  \|y-w\|_1\lesssim\operatorname{diam}\Omega_{k}^{s+1}\lesssim\left( \operatorname{diam}\widetilde{\Omega}_k^{s+1}\right) ^\frac{1}{c}$$
	and
	\begin{equation}\label{guji}
		|\log\delta_{s+1}|\sim|\log\delta_{s}|^{c^5}\sim N_{s+1}^{c^{10}\tau}<N_{s+1}^{\frac{1}{c}}.
	\end{equation}
This proves the lemma. 
\end{proof}
Next we consider the general case and will finish the proof of $({\bf e})_{s+1}$. Define 
$$ \Lambda'=\Lambda\backslash \bigcup_{k\in \widetilde{P}_{s+1}}\Omega^{s+1}_{k}.$$
We claim that  $\Lambda'$ is $s$-\textbf{good}. In fact, for $s'\leq s-1$, assume $\widetilde{\Omega}_{l'}^{s'}\subset\Lambda',\  \widetilde{\Omega}_{l'}^{s'}\subset \Omega_{l}^{s'+1}$ and $\widetilde{\Omega}_l^{s'+1} \cap \left(\bigcup_{k\in \widetilde{P}_{s+1}}\Omega^{s+1}_{k}\right)\neq \emptyset.$
Thus by \eqref{xsxs},	 we obtain
 $\widetilde{\Omega}_l^{s'+1} \subset \bigcup_{k\in \widetilde{P}_{s+1}}\Omega^{s+1}_{k},$
which  contradicts    $ \widetilde{\Omega}_{l'}^{s'}\subset\Lambda'.$
If there exists $k'$ such that $k'\in Q_s$ and  $\widetilde{\Omega}_{k'}^{s}\subset \Lambda'\subset \Lambda$,  then by \eqref{s+1g} there exists $k\in P_{s+1}$, such that
$$\widetilde{\Omega}_{k'}^{s}\subset \Omega_{k}^{s+1}\subset \Lambda.$$
Hence recalling \eqref{wp}, one has $k\in \widetilde{P}_{s+1}$  and 
$$\widetilde{\Omega}_{k'}^{s}\subset\bigcup_{k\in \widetilde{P}_{s+1}}  \Omega_{k}^{s+1}.$$ 
This contradicts  $	\widetilde{\Omega}_{k'}^{s}\subset \Lambda'$. We have proven the claim. As a result, 
the estimates \eqref{L2} and \eqref{exp} hold true with $\Lambda$ replaced by $\Lambda'$. We can now estimate $T_\Lambda^{-1}$. For this purpose, we have the following two cases. 
\begin{itemize}
	\item[(1).]
	Assume that $x\notin \bigcup_{k\in \widetilde{P}_{s+1}}2\Omega^{s+1}_{k}$.    Then  $N_{s}^{c^3}\ll N_{s+1}\leq  {\rm dist}(x,\partial_\Lambda^-\Lambda')$.
	For $y \in \Lambda$, using resolvent identity shows
	\begin{align*}
		T_{\Lambda}^{-1}(x,y)=T_{\Lambda'}^{-1}(x,y)\chi_{\Lambda'}(y)-\sum_{(w, w')\in\partial_\Lambda\Lambda'} T_{\Lambda'}^{-1}(x,w)\Gamma(w,w')T_{\Lambda}^{-1}(w',y).
	\end{align*}
	Since
	\begin{align*}
		\sum_{y\in \Lambda'}|T_{\Lambda'}^{-1}(x,y)\chi_{\Lambda'}(y)|&\leq \sum_{\|x-y\|\leq N_s^{c^3}}|T_{\Lambda'}^{-1}(x,y)|+ \sum_{\|x-y\|> N_s^{c^3}}|T_{\Lambda'}^{-1}(x,y)|\\
		&\leq N_s^{c^3}\cdot \|T^{-1}_{\Lambda'}\|+\sum_{\|x-y\|> N_s^{c^3}}e^{-\gamma_s\|x-y\|_1}\\
		&\leq 2N_s^{c^3}\delta_{s-1}^{-3}\delta_{s}^{-2}\\&<\frac{1}{2}\delta_{s}^{-3}
	\end{align*}
	and
	\begin{align*}
		\sum_{w\in\partial^-_\Lambda\Lambda'}|T_{\Lambda'}^{-1}(x,w)|&\leq\sum_{\|x-w\|_1\geq N_{s+1}}e^{-\gamma_s\|x-w\|_1}<e^{-\frac{1}{2}\gamma_sN_{s+1}},
	\end{align*}
	we get
	\begin{align*}
		\sum_{y\in \Lambda}|T_{\Lambda}^{-1}(x,y)|&\leq \sum_{y\in \Lambda'}|T_{\Lambda'}^{-1}(x,y)\chi_{\Lambda'}(y)|+\sum_{y\in \Lambda,(w, w')\in\partial_\Lambda\Lambda'}|T_{\Lambda'}^{-1}(x,w)\Gamma(w,w')T_{\Lambda}^{-1}(w',y)|  \nonumber \\
		&\leq\frac{1}{2}\delta_s^{-3}+2d\sum_{w\in\partial^-_\Lambda\Lambda'}|T_{\Lambda'}^{-1}(x,w)|\cdot \sup_{w'\in \Lambda}\sum_{y\in \Lambda}|T_{\Lambda}^{-1}(w',y)|\\
		&\leq\frac{1}{2}\delta_s^{-3}+\frac{1}{10}\sup_{w'\in \Lambda}\sum_{y\in \Lambda}|T_{\Lambda}^{-1}(w',y)|.
	\end{align*}
	\item[(2).]   Assume that $x\in 2\Omega^{s+1}_{k}$ for some $k\in \widetilde{P}_{s+1}$. Then by \eqref{s+1g}, we have $\widetilde{\Omega}_{k}^{s+1} \subset \Lambda$ and    $k\notin Q_{s+1}.$
	For $y\in \Lambda,$ t using  resolvent identity shows
	\begin{align*}
		T_{\Lambda}^{-1}(x,y)=T_{\widetilde{\Omega}_{k}^{s+1}}^{-1}(x,y)\chi_{\widetilde{\Omega}_{k}^{s+1}}(y)-\sum_{(w,w')\in\partial_\Lambda\widetilde{\Omega}_{k}^{s+1}}T_{\widetilde{\Omega}_{k}^{s+1}}^{-1}(x,w)\Gamma(w,w')T_{\Lambda}^{-1}(w',y).
	\end{align*}
By  \eqref{222},  \eqref{exps+1} and 
	$$ N_{s+1}<\operatorname{diam}\widetilde{\Omega}_{k}^{s+1}\lesssim {\rm dist}(x,\partial^-_\Lambda\widetilde{\Omega}_{k}^{s+1}),$$
	we have
	\begin{align*}
		\sum_{y\in \Lambda}|T_{\Lambda}^{-1}(x,y)|&\leq \sum_{y\in \Lambda}|T_{\widetilde{\Omega}_{k}^{s+1}}^{-1}(x,y)\chi_{\widetilde{\Omega}_{k}^{s+1}}(y)|+\sum_{y\in\Lambda,(w,w')\in\partial_\Lambda\widetilde{\Omega}_{k}^{s+1}}|T_{\widetilde{\Omega}_{k}^{s+1}}^{-1}(x,w)\Gamma(w,w')T_{\Lambda}^{-1}(w',y)|   \\
		&<\#\widetilde{\Omega}_{k}^{s+1}\cdot \|T_{\widetilde{\Omega}_{k}^{s+1}}^{-1}\|+CN_{s+1}^{c^2d}e^{-\widetilde{\gamma}_s N_{s+1}}\sup_{w'\in \Lambda}\sum_{y\in \Lambda}|T_{\Lambda}^{-1}(w',y)|\\
		&<CN_{s+1}^{c^2d}\delta_{s}^{-2}\left\|\theta+k \cdot \omega-\theta_{s+1}\right\|^{-1} \cdot\left\|\theta+k \cdot \omega+\theta_{s+1}\right\|^{-1}+\frac{1}{10}\sup_{w'\in \Lambda}\sum_{y\in \Lambda}|T_{\Lambda}^{-1}(w',y)|\\
		&<\frac{1}{2}\delta_{s}^{-3}\left\|\theta+k \cdot \omega-\theta_{s+1}\right\|^{-1} \cdot\left\|\theta+k \cdot \omega+\theta_{s+1}\right\|^{-1}+\frac{1}{10}\sup_{w'\in \Lambda}\sum_{y\in \Lambda}|T_{\Lambda}^{-1}(w',y)|.
	\end{align*}
\end{itemize}
Combining the above two cases, we obtain
\begin{align}\label{L22}
	\|T_{\Lambda}^{-1}\|&\leq\sup_{x\in \Lambda}\sum_{y\in \Lambda}|T_{\Lambda}^{-1}(x,y)| \nonumber \\
	&<\delta^{-3}_{s}\sup_{\left\{k\in P_{s+1} :\   \widetilde{\Omega}_{k}^{s+1}\subset \Lambda \right\}}\left\|\theta+k\cdot  \omega-\theta_{s+1}\right\|^{-1} \cdot\left\|\theta+k \cdot \omega+\theta_{s+1}\right\|^{-1}.
\end{align}
Finally, we turn to the off-diagonal decay estimates. From  \eqref{gai}, \eqref{s+1g} and  \eqref{wp}, it follows that for $k'\in \widetilde{P}_t\cap Q_t\ (1\leq t\leq s)$
there exists $k\in \widetilde{P}_{t+1}$ such that
$$\widetilde{\Omega}_{k'}^t\subset\Omega_k^{t+1}$$
and 
$$\widetilde{P}_{s+1}\cap Q_{s+1}=\emptyset.$$  Moreover,
$$\bigcup_{1\leq t\leq s+1}\bigcup_{k\in \widetilde{P}_t}\widetilde{\Omega}^t_{k}\subset \Lambda.$$
Hence for any $w \in \Lambda$, if
$$w\in \bigcup_{k\in \widetilde{P}_1}2\Omega^1_{k},  $$
then there exists $1\leq t\leq s+1$ such that
$$w\in \bigcup_{k\in \widetilde{P}_t\setminus Q_t}2\Omega^t_{k}.  $$
For every $w\in \Lambda,$ define its block in $\Lambda$
\[J_w=\left\{\begin{aligned}
	&\Lambda_{\frac{1}{2}N_1}(w)\cap\Lambda \quad \text{if }   w\notin \bigcup_{k\in \widetilde{P}_1}2\Omega^1_{k},\   &\textcircled{1} \\
	&\widetilde{\Omega}^{t}_k \quad \text{if }  w \in2\Omega_k^{t} \text{ for some }k\in  \widetilde{P}_t\setminus Q_t.\   &\textcircled{2}
\end{aligned}\right. \]
Then $\operatorname{diam}J_w\leq\operatorname{diam}\widetilde{\Omega}^{s+1}_k < 3N_{s+1}^{c^2}$. For \textcircled{1}, we have  $J_w\cap Q_0=\emptyset$ and ${\rm dist}(w,\partial^-_\Lambda J_w)\geq\frac{1}{2}N_1.$ Thus
$$| T^{-1}_{J_w}(w,w')|<e^{-\gamma_0\|w-w'\|_1} \  \text{\rm for}\  w'\in\partial^-_\Lambda J_w.  $$
For \textcircled{2}, by \eqref{exps+1}, we have
$$  |T_{J_w}(w,w')|<e^{-\widetilde{\gamma}_{t-1}\|w-w'\|_1} \  {\rm for}\  w'\in\partial^-_\Lambda J_w. $$
Let $\|x-y\|>N_{s+1}^{c^3}$.  The resolvent identity reads as 
\begin{align*}
	T_{\Lambda}^{-1}(x,y)=T_{J_x}^{-1}(x,y)\chi_{J_x}(y)-\sum _{(w,w')\in\partial_\Lambda J_x}T_{J_x}^{-1}(x,w)\Gamma(w,w')T_{\Lambda}^{-1}(w',y).
\end{align*}
The first term in the above identity is zero since $\|x-y\|>N_{s+1}^{c^3}>3 N_{s+1}^{c^2}$ (so that $y\notin J_x$). It follows that
\begin{align*}
	|T_{\Lambda}^{-1}(x,y)|&\leq C N_{s+1}^{c^2d}e^{-\min\limits_t\left( \gamma_0(1-2N_1^{-1}),\widetilde{\gamma}_{t-1}(1-N_{t}^{-1})\right)\|x-x_1\|_1 }|T_{\Lambda}^{-1}(x_1,y)| \\
	&\leq C N_{s+1}^{c^2d} e^{-\widetilde{\gamma}_s(1-N_{s+1}^{-1})\|x-x_1\|_1}|T_{\Lambda}^{-1}(x_1,y)|\\
	&<e^{-\widetilde{\gamma}_s(1-N_{s+1}^{-1}-\frac{C\log N_{s+1}}{N_{s+1}})\|x-x_1\|_1}|T_{\Lambda}^{-1}(x_1,y)|\\
	&<e^{-\gamma_s(1-N_{s+1}^{\frac{1}{c}-1})^2\|x-x_1\|_1}|T_{\Lambda}^{-1}(x_1,y)|\\
	&=e^{-\gamma_s'\|x-x_1\|_1}|T_{\Lambda}^{-1}(x_1,y)|
\end{align*}	for some $x_1\in \partial_\Lambda^+ J_x $,
where $\gamma_s'=\gamma_s(1-N_{s+1}^{\frac{1}{c}-1})^2.$
Iterate the above procedure   and stop it if  for some $L$, $\|x_{L}-y\|<3N_{s+1}^{c^2}$. Recalling \eqref{guji} and \eqref{L22}, we get
\begin{align*}
	|T_{\Lambda}^{-1}(x,y)|&\leq e^{-\gamma_s'\|x-x_1\|_1}\cdots e^{-\gamma_s'\|x_{L-1}-x_L\|_1}|T_{\Lambda}^{-1}(x_L,y)| \\
	&\leq e^{-\gamma_s'(\|x-y\|_1-3N_{s+1}^{c^2})}\|T_\Lambda^{-1}\|<e^{-\gamma_s'(1-3N_{s+1}^{c^2-c^3})\|x-y\|_1}\delta_{s+1}^{-3}\\
	&<e^{-\gamma_s'(1-3N_{s+1}^{c^2-c^3}-3\frac{|\log\delta_{s+1}|}{N_{s+1}^{c^3}})\|x-y\|_1}\\
	&<e^{-\gamma_s'(1-N_{s+1}^{\frac{1}{c}-1})\|x-y\|_1}\\
	&=e^{-\gamma_{s+1}\|x-y\|_1}.
\end{align*}
This gives the off-diagonal decay estimates. 

We have completed the proof of Theorem \ref{Inthm}.
\end{proof}
\section{Arithmetic Anderson localization}
As an application of Green's function estimates of previous section, we prove the arithmetic version of Anderson localization below. 
\begin{proof}[Proof of Theorem \ref{thm1}]
	Recall first 
	 \begin{align}\label{res}
		\nonumber&\Theta_{\tau_1}=\left\{(\theta,\omega) \in \mathbb{T}\times \mathcal{R}_{\tau,\gamma}:\   {\rm the\  relation}\  \left\| 2 \theta+n\cdot \omega \right\|\leq e^{-\|n\|^{\tau_1}} \text{ holds for finitely many $n\in \mathbb{Z}^d$} \right\},
	\end{align}
	where $0<\tau_1<\tau$.
	
	
	We prove  for $0<\varepsilon\leq\varepsilon_0$, $\omega\in\mathcal{R}_{\tau,\gamma}$ and $(\theta,\omega) \in \Theta_{\tau_1},$ $H(\theta)$ has only pure point spectrum with exponentially decaying eigenfunctions.
	Let $\varepsilon_0$ be given by Theorem \ref{Inthm}. Fix $\omega$ and $\theta$ so that $\omega\in\mathcal{R}_{\tau,\gamma}$ and $(\theta,\omega) \in \Theta_{\tau_1}.$  Let $E\in [-2,2]$ be a generalized eigenvalue of $H(\theta)$ and $u=\{u(n)\}_{n\in\Z^d}\neq0$ be the corresponding generalized eigenfunction satisfying $|u(n)|\leq(1+\|n\|)^d.$ From Schnol's theorem, it suffices to show $u$  decays exponentially. For this purpose, note first  there exists (since $(\theta,\omega) \in \Theta_{\tau_1}$) some $\widetilde{s}\in\N$ such that
	\begin{equation}\label{xiajie}
		\left\| 2 \theta+n\cdot \omega \right\|>e^{-\|n\|^{\tau_1}}\ {\rm for}\ {\rm all}\   n\ {\rm satisfying}\ \|n\| \geq N_{\widetilde{s}}.	
		\end{equation}
	We claim  that there exists $s_0>0$ such that,  for $s\geq s_0$,
	\begin{equation}\label{mao}
		\Lambda_{2N_s^{c^4}}\cap\left(\bigcup_{k\in Q_s}\widetilde{\Omega}_k^s\right)\neq \emptyset.
	\end{equation}
	For otherwise, then there exist a subsequence $s_i\to+\infty$ (as $i\to\infty$) such that
	\begin{equation}\label{bumao}
		\Lambda_{2N_{s_i}^{c^4}}\cap\left(\bigcup_{k\in Q_{s_i}}\widetilde{\Omega}_k^{s_i}\right)= \emptyset.
	\end{equation}
	Then we can enlarge $\Lambda_{N_{s_i}^{c^4}}$ 
	to $\widetilde{\Lambda}_{i}$
	satisfying 
	$$
	\Lambda_{N_{s_i}^{c^4}}\subset \widetilde{\Lambda}_{i}\subset \Lambda_{N_{s_i}^{c^4}+50N_{s_i}^{c^2}}, 	$$
	and 
	$$ \widetilde{\Lambda}_{i}\cap\widetilde{\Omega}_k^{s'}\neq\emptyset \Rightarrow \widetilde{\Omega}_k^{s'}\subset \widetilde{\Lambda}_{i}\ {\rm for}\  s'\leq s\ {\rm and}\  k\in P_{s'}.$$
	From \eqref{bumao}, we have
	$$ \widetilde{\Lambda}_{i}\cap\left(\bigcup_{k\in Q_{s_i}}\widetilde{\Omega}_k^{s_i}\right)= \emptyset,$$
	which shows  $\widetilde{\Lambda}_{i}$ is $s_i$-\textbf{good}. As a result, for $n\in \Lambda_{N_{s_i}}$,  since ${\rm dist}(n,\partial^- \widetilde{\Lambda}_{N_{s_i}^{c^4}})\geq\frac{1}{2} N_{s_i}^{c^4}>N_{s_i}^{c^3}$, we have
	\begin{align*}
		|u(n)|&\leq \sum_{(n',n'')\in \partial \widetilde{\Lambda}_{i}} | T^{-1}_{\widetilde{\Lambda}_{N_{s_i}^{c^4}}}(n, n') u(n'')|\\
		&\leq2d\sum_{n'\in \partial^- \widetilde{\Lambda}_{i}} | T^{-1}_{\widetilde{\Lambda}_{i}}(n, n')|\cdot \sup_{n''\in  \partial^+ \widetilde{\Lambda}_{i}}| u(n'')|\\
		&\leq C N_{s_i}^{2c^4d}\cdot e^{-\frac{1}{2}\gamma_{\infty}N_{s_i}^{c^4}}\rightarrow0.
	\end{align*}
	From $N_{s_i}\to +\infty$, it follows that $u(n)=0\ {\rm for}\ \forall\  n\in \mathbb{Z}^d$. This contradicts  $u\neq0$, and the claim is proved. 
	
	Next, define
	$$U_{s} =\Lambda_{8 N_{s+1}^{c^4}} \backslash \Lambda_{4 N_{s}^{c^4}},\ U_{s}^{*} = \Lambda_{10 N_{s+1}^{c^4}} \backslash \Lambda_{3 N_{s}^{c^4}}. $$
	We can also enlarge  $U^*_s$  to $\widetilde{U}^*_s$ so that
	$$ U^*_s\subset \widetilde{U}^*_s\subset \Lambda_{50N_s^{c^2}}(	U^*_s),$$
	and 
	$$\widetilde{U}^*_s\cap\widetilde{\Omega}_k^{s'}\neq\emptyset \Rightarrow \widetilde{\Omega}_k^{s'}\subset \widetilde{U}^*_s \ {\rm for}\  s'\leq s \ {\rm and}\  k\in P_{s'} .$$
Let $n$ satisfy $\|n\|>\max(4 N_{\widetilde{s}}^{c^4},4 N_{s_0}^{c^4})$. Then  there exists some $s\geq\max(\widetilde{s},s_0)$ such that
	\begin{equation}\label{n}
		n \in 	U_{s}.
	\end{equation}
By \eqref{mao},  without loss of generality, we may assume
	$$	\Lambda_{2N_s^{c^4}}\cap\widetilde{\Omega}_k^s\neq \emptyset$$
	for some $k\in Q_s^+$.
	Then for $k\neq k'\in Q_s^+$, we have
	$$\|k-k'\|>\left|\log\frac{\gamma}{2\delta_s}\right|^\frac{1}{\tau}\gtrsim N_{s+1}^{c^5}\gg\operatorname{diam}\widetilde{U}^*_s.$$
	Thus $$\widetilde{U}^*_s\cap\left(\bigcup_{l\in Q_s^+}\widetilde{\Omega}_l^s\right)= \emptyset.$$
	Now, if there exists $l\in Q_s^-$ such that
	$$\widetilde{U}^*_s\cap\widetilde{\Omega}_l^s\neq  \emptyset,$$
	then
	$$N_{s}<N_{s}^{c^4}-100N_s^{c^2} \leq\|l\|-\|k\|\leq\|l+k\|\leq\|l\|+\|k\|<11N_{s+1}^{c^4}.$$
	Recalling $$ Q_s\subset P_s \subset \mathbb{Z}^d+\frac{1}{2}\sum_{i=0}^{s-1}l_i, $$
	we have $l+k\in \mathbb{Z}^d.$ Hence  by \eqref{xiajie}, 
	\begin{align*}
	e^{-(11N_{s+1}^{c^4})^{\tau_1}}&<\|2\theta+ (l+k)\cdot \omega\|\\
	&\leq\|\theta+l\cdot \omega-\theta_s\|+\|\theta+k\cdot \omega+\theta_s\|<2\delta_s.
	 \end{align*}
	This contradicts 	$$|\log\delta_{s}|\sim N_{s+1}^{c^5\tau}\gg N_{s+1}^{c^4\tau_1}.$$
	We thus have shown
	$$\widetilde{U}^*_s\cap\left(\bigcup_{l\in Q_s}\widetilde{\Omega}_l^s\right)= \emptyset.$$
	This implies $\widetilde{U}^*_s$ is $s$-\textbf{good}.
	
	Finally, recalling \eqref{n},  we have  $${\rm dist}(n,\partial^-\widetilde{U}^*_s)\geq\min\left({10N_{s+1}^{c^4}-|n|,|n|-3N_s^{c^4}}\right)-1\geq\frac{1}{5}\|n\|> N_s^{c^3}.$$
	Then
	\begin{align*}
		|u(n)|&\leq \sum_{(n',n'')\in \partial \widetilde{U}^*_s} | T^{-1}_{\widetilde{U}^*_s}(n, n') u(n'')|\\
		&\leq2d\sum_{n'\in \partial^- \widetilde{U}^*_s} | T^{-1}_{\widetilde{U}^*_s}(n, n')|\cdot \sup_{n''\in  \partial^+ \widetilde{U}^*_s}| u(n'')|\\
		&\leq CN_{s+1}^{2c^4d}\cdot e^{-\frac{1}{5}\gamma_{\infty}\|n\|}\\
		&\leq C\|n\|^{2c^5d}	\cdot e^{-\frac{1}{5}\gamma_{\infty}\|n\|}\\
		&<e^{-\frac{1}{6}\gamma_{\infty}\|n\|},
	\end{align*}
	which yields  the exponential decay $u$.	

We complete the proof of Theorem \ref{thm1}. 
\end{proof}
\begin{rem}
	 Assume that for some $E\in[-2,2]$, the inductive  process  stops at a finite stage (i.e.,  $Q_s=\emptyset$ for some $s<\infty$). Then for  $N>N_s^{c^5}$, we can enlarge ${\Lambda}_N$  to $\widetilde{\Lambda}_N$  with
	$${\Lambda}_N\subset \widetilde{\Lambda}_{N}\subset \Lambda_{N+50N_{s}^{c^2}},$$
	and
	$$\widetilde{\Lambda}_N\cap\widetilde{\Omega}_k^{s'}\neq\emptyset \Rightarrow \widetilde{\Omega}_k^{s'}\subset \widetilde{\Lambda}_N\ {\rm for}\  s'\leq s \ {\rm and}\  k\in P_{s'}. $$
	Thus $\widetilde{\Lambda}_N$ is $s$-\textbf{good}.
	For $n\in {\Lambda}_{N^\frac{1}{2}}$, since  ${\rm dist}(n,\partial^-\widetilde{\Lambda}_N)>N_s^{c^3}$, we have
	\begin{align*}
		|u(n)|&\leq \sum_{(n',n'')\in \partial \widetilde{\Lambda}_{N}} | T^{-1}_{\widetilde{\Lambda}_{N}}(n, n') u(n'')|\\
		&\leq2d\sum_{n'\in \partial^- \widetilde{\Lambda}_{N}} | T^{-1}_{\widetilde{\Lambda}_{N}}(n, n')|\cdot \sup_{n''\in  \partial^+ \widetilde{\Lambda}_{N}}| u(n'')|\\
		&\leq C N^{2d}\cdot e^{-\frac{1}{2}\gamma_{\infty}N}\rightarrow0.
	\end{align*}
	Hence  such $E$ is not a generalized eigenvalue of $H(\theta)$.
\end{rem}

\section{$(\frac12-)$-H\"older continuity of the IDS}
In this section,  we apply our estimates to obtain $(\frac12-)$-H\"older continuity of the IDS. 
\begin{proof}[Proof of Theorem \ref{thm2}]
	
	Let $T$ be given by \eqref{T}.  Fix $\mu>0$, $\theta\in \mathbb{T}$ and $ E\in [-2,2]$. Let $\varepsilon_0$ be defined in Theorem \ref{Inthm} and assume $0<\varepsilon\leq \varepsilon_0.$
	Fix\begin{equation}\label{zh}
		0<\eta<\eta_0=\min\left( e^{-\left( \frac{4}{\mu}\right) ^{\frac{c}{c-1}}},e^{-|\log\delta_0|^c}\right).
	\end{equation} 
	
	Denote by  $\left\{\xi_{r} :\  r=1, \ldots, R\right\} \subset\operatorname{span}\left(\delta_{n} :\  n \in \Lambda_{N}\right)$   the $\ell^{2}$-orthonormal eigenvectors of  $T_{\Lambda_N}$ with eigenvalues belonging to $[-\eta, \eta]$. We aim to prove that for sufficiently large $N$ (depending on $\eta$), 
	$$R\leq(\#\Lambda_N)\eta^{\frac{1}{2}-\mu}.$$
	From \eqref{zh}, we can choose $ s\geq1$ such that
	$$
	|\log \delta_{s-1}|^{c}\leq|\log\eta|<|\log \delta_{s}|^{c}.
	$$
	Enlarge  $\Lambda_{N}$ to $\widetilde{\Lambda}_N $ so that
	$$	\Lambda_{N}\subset\widetilde{\Lambda}_N\subset \Lambda_{N+50N_s^{c^2}}$$
	and $$
	\quad \widetilde{\Lambda}_N\cap\widetilde{\Omega}_k^{s'}\neq\emptyset \Rightarrow \widetilde{\Omega}_k^{s'}\subset\widetilde{\Lambda}_N\ {\rm for}\  s'\leq s \  {\rm and}\  k \in P_{s'}.$$
	Define further
	$$
	\mathcal{K}=\left\{k \in P_{s} :\  \widetilde{\Omega}_k^{s}\subset\widetilde{\Lambda}_N , \min_{\sigma=\pm1} (\|\theta +k\cdot  \omega+\sigma \theta_s\|)<\eta^{\frac{1}{2}-\frac{\mu}{2} }\right\}
	$$
and 
	$$
	\widetilde{\Lambda}_N'=\widetilde{\Lambda}_N \setminus \bigcup_{k \in \mathcal{K}} \Omega_{k}^{s}.
	$$
	Thus by \eqref{haha}, we obtain
	$$k'\in Q_{s'},\widetilde{\Omega}_{k'}^{s'}\subset\widetilde{\Lambda}_N',\widetilde{\Omega}_{k'}^{s'}\subset \Omega_{k}^{s'+1} \Rightarrow \widetilde{\Omega}_k^{s'+1} \subset\widetilde{\Lambda}_N'\ \text{ for }s'<s.$$
	Since
	$$|\log\eta|<|\log \delta_{s}|^{c}\sim|\log \delta_{s-1}|^{c^6}\sim N_{s}^{c^{11}\tau}<N_{s}^\frac{1}{c},$$
	we get from the resolvent identity
	\begin{align}\label{eL2}
		\begin{split}
			\|T_{\widetilde{\Lambda}_N'}^{-1}\|&<\delta^{-3}_{s-1}\sup_{\left\{k\in P_s :\   \widetilde{\Omega}_{k}^s\subset \widetilde{\Lambda}_N' \right\}}\|\theta+k\cdot\omega-\theta_{s}\|^{-1}\cdot \|\theta+k\cdot\omega+\theta_{s}\|^{-1} \\
			&<\delta^{-3}_{s-1}\eta^{\mu-1 }
			<\frac{1}{2}\eta^{-1},
		\end{split}	
	\end{align}
	where  the last inequality follows from \eqref{zh}.

	By the uniform distribution of $\{n\cdot \omega\}_{n\in \mathbb{Z}^d}$ in $\mathbb{T}$, we have
	\begin{align*}
		\#(\widetilde{\Lambda}_N\setminus\widetilde{\Lambda}_N')&\leq \# \Omega_k^s\cdot \#\mathcal{K}\\
		&\leq CN_s^{cd}\cdot \#\left\{k \in \mathbb{Z}+\sum_{i=0}^{s-1}l_i :\ \|k\|\leq N+50N_s^{c^2},  \min_{\sigma=\pm1} (\|\theta +k\cdot  \omega+\sigma \theta_s\|)<\eta^{\frac{1}{2}-\frac{\mu}{2} }\right\}\\
		&\leq CN_s^{cd}\cdot \eta^{\frac{1}{2}-\frac{\mu}{2} }(N+50N_s^{c^2})^d\\
		&\leq CN_s^{cd}\cdot \eta^{\frac{1}{2}-\frac{\mu}{2} }\#\Lambda_N
	\end{align*}
	for sufficiently large $N$.
	
	For a vector  $\xi\in\C^{\Lambda}$ with $\Lambda\subset\Z^d$, we define $\|\xi\|$ to be the $\ell^2$-norm.  Assume  $\xi \in\left\{\xi_{r} :\  r \leq R\right\}$ be an eigenvector of $T_{\Lambda_N}$.  Then
	$$\|T_{\Lambda_N}\xi\|=\|R_{\Lambda_N}T\xi\|\leq\eta.$$
	Hence
	\begin{align}\label{fuza}
			\eta\geq\|R_{\widetilde{\Lambda}_N'}T_{\Lambda_N}\xi\|
			&=\|R_{\widetilde{\Lambda}_N'}TR_{\widetilde{\Lambda}_N'}\xi+R_{\widetilde{\Lambda}_N'}TR_{\Lambda_N\setminus\widetilde{\Lambda}_N'}\xi- R_{\widetilde{\Lambda}_N'\setminus\Lambda_N}T\xi  \|
	\end{align}
	Applying $T_{\widetilde{\Lambda}_N'}^{-1}$ to \eqref{fuza} and \eqref{eL2}  implies
	\begin{align}\label{zafu}
		\left\|R_{\widetilde{\Lambda}_N'}\xi+T_{\widetilde{\Lambda}_N'}^{-1}\left( R_{\widetilde{\Lambda}_N'}TR_{\Lambda_N\setminus\widetilde{\Lambda}_N'}\xi- R_{\widetilde{\Lambda}_N'\setminus\Lambda_N}T\xi \right)  \right\|<\frac{1}{2}.
	\end{align}
	Denote
	$$H=\operatorname{Range}T_{\widetilde{\Lambda}_N'}^{-1}\left( R_{\widetilde{\Lambda}_N'}TR_{\Lambda_N\setminus\widetilde{\Lambda}_N'}- R_{\widetilde{\Lambda}_N'\setminus\Lambda_N}T \right).$$
	Then
	\begin{align*}
		\operatorname{dim}H&\leq\operatorname{Rank}T_{\widetilde{\Lambda}_N'}^{-1}\left( R_{\widetilde{\Lambda}_N'}TR_{\Lambda_N\setminus\widetilde{\Lambda}_N'}- R_{\widetilde{\Lambda}_N'\setminus\Lambda_N}T \right)\\
		&\leq\#(\widetilde{\Lambda}_N\setminus\widetilde{\Lambda}_N') +\#(\widetilde{\Lambda}_N\setminus\Lambda_N) \\
		&\leq CN_s^{cd}\cdot \eta^{\frac{1}{2}-\frac{\mu}{2} }\#\Lambda_N+CN_s^{c^2d}N^{d-1}\\
		&\leq CN_s^{cd}\cdot \eta^{\frac{1}{2}-\frac{\mu}{2} }\#\Lambda_N.
	\end{align*}
	Denote by $P_H$ the orthogonal projection to $H$. Applying $I-P_H$ to \eqref{zafu}, we get
	\begin{equation*}\label{xishou}
		\|R_{\widetilde{\Lambda}_N'}\xi-P_HR_{\widetilde{\Lambda}_N'}\xi\|^2=\|R_{\widetilde{\Lambda}_N'}\xi\|^2-\|P_HR_{\widetilde{\Lambda}_N'}\xi\|^2\leq\frac{1}{4}.
	\end{equation*}
	Before concluding the proof, we need a useful lemma. 
	\begin{lem}\label{fx}
		Let $H$ be a Hilbert space and  let $H_1$, $H_2$ be its subspaces. Let $\left\{\xi_{r} :\  r=1, \ldots, R\right\}$ be a set of  orthonormal vectors.  Then we have
		$$\sum_{r=1}^R\|P_{H_1}P_{H_2}\xi_{r}\|^2\leq\operatorname{dim}H_1.$$
	\end{lem}
	
	\begin{proof}[Proof of Lemma \ref{fx}]
		Denote by  $\left\langle \cdot, \cdot \right\rangle$ the inner product on $H.$ Let  $\left\{\phi_{i}\right\}$ be the orthonormal basis of $H_1$. By Parseval's equality and  Bessel's inequality, we have
		\begin{align*}
			\sum_{r=1}^R\|P_{H_1}P_{H_2}\xi_{r}\|^2&=\sum_{r=1}^R\sum_i|\left\langle\phi_{i},P_{H_2}\xi_r \right\rangle|^2 \\
			&=\sum_i\sum_{r=1}^R|\left\langle P_{H_2}\phi_{i},\xi_r \right\rangle|^2 \\
			&\leq \sum_i\|P_{H_2}\phi_{i}\|^2\\
			&\leq \sum_i\|\phi_{i}\|^2\leq\operatorname{dim}H_1.
		\end{align*}	
	\end{proof}
	Finally, it follows from Lemma \ref{fx} that
	\begin{align*}
		R=\sum_{r=1}^R\|\xi_{r}\|^2&=\sum_{r=1}^R	\|R_{\widetilde{\Lambda}_N'}\xi_r\|^2+\sum_{r=1}^R\|R_{\Lambda_N\setminus\widetilde{\Lambda}_N'}\xi_r\|^2\\
		&\leq\frac{1}{4}R+\sum_{r=1}^R\left(\|P_HR_{\widetilde{\Lambda}_N'}\xi_r\|^2+\|R_{\Lambda_N\setminus\widetilde{\Lambda}_N'}\xi_r\|^2\right)\\
		&\leq\frac{1}{4}R+\operatorname{dim}H+\#(\Lambda_N\setminus\widetilde{\Lambda}_N')\\
		&\leq\frac{1}{4}R+ CN_s^{cd}\cdot \eta^{\frac{1}{2}-\frac{\mu}{2} }\#\Lambda_N.
	\end{align*}
	Hence we get
	$$	R\leq CN_s^{cd}\cdot \eta^{\frac{1}{2}-\frac{\mu}{2} }\#\Lambda_N\leq\eta^{\frac{1}{2}-\mu} \#\Lambda_N.$$
We finish the proof of Theorem \ref{thm2}. 
\end{proof}	
\begin{rem}
	In the above  proof,   if the inductive process stops at a finite stage (i.e., $Q_s=\emptyset$ for some $s$) and $|\log\delta_s|^c \leq|\log\eta|$. Then $\widetilde{\Lambda}_N$ is $s$-\textbf{good} and 
	$$\|T_{\widetilde{\Lambda}_N}^{-1}\|
	<\delta^{-3}_{s-1}\delta_s^{-2}
	<\frac{1}{2}\eta^{-1},$$
which implies  $$R\leq \frac{4}{3}\#(\widetilde{\Lambda}_N\setminus\Lambda_N)\leq C N_s^{c^2d} N^{-1}\#\Lambda_N.$$ Letting  $N\to \infty$, we get $	\mathcal{N}(E+\eta)-\mathcal{N}(E-\eta)=0$, which means  $E\notin \sigma(H(\theta)).$
\end{rem}
\section*{Acknowledgments}
   Y. Shi  is partially supported by National Key R\&D Program under Grant 2021YFA1001600  and  NSF of China under Grant 12271380.  Z. Zhang is partially supported by NSF of China under Grant 12171010.
\appendix{}
\section{}\label{appa}
 \begin{proof}[Proof of Remark \ref{r0}]
	Let $i\in Q_{0}^{+}$ and  $j\in\widetilde{Q}_{0}^{-}$ satisfy  \[\left\|\theta+i\cdot  \omega+\theta_{0}\right\|<\delta_{0}, \ \left\|\theta+j\cdot  \omega-\theta_{0}\right\|<\delta_{0}^{\frac{1}{100}}.\]
Then  \eqref{bruno} implies $1,\omega_1,\dots,\omega_d$ are rational independent and $\{k\cdot \omega\}_{k \in \mathbb{Z}^d}$ is dense in $\mathbb{T}$. Thus there exists $k\in \mathbb{Z}^d$ such that $\|2\theta+k\cdot \omega\|$ is sufficiently small  with
	\begin{align*}
	\left\|\theta+(k-j)\cdot  \omega+\theta_{0}\right\|&\leq\|2\theta+k\cdot \omega\|+\left\|\theta+j\cdot  \omega-\theta_{0}\right\|<\delta_{0}^{\frac{1}{100}},\\
	\left\|\theta+(k-i)\cdot  \omega-\theta_{0}\right\|&\leq\|2\theta+k\cdot \omega\|+\left\|\theta+i \cdot \omega+\theta_{0}\right\|<\delta_{0}.
	\end{align*}
We obtain then  $k-j\in\widetilde{Q}_{0}^{+}$ and $k-i\in Q_{0}^{-},$ which implies $${\rm dist}\left(\widetilde{Q}_{0}^{+}, Q_{0}^{-}\right)\leq {\rm dist}\left(\widetilde{Q}_{0}^{-}, Q_{0}^{+}\right).$$
	The similar argument shows  $${\rm dist}\left(\widetilde{Q}_{0}^{+}, Q_{0}^{-}\right)\geq {\rm dist}\left(\widetilde{Q}_{0}^{-}, Q_{0}^{+}\right).$$
	We have shown  $${\rm dist}\left(\widetilde{Q}_{0}^{+}, Q_{0}^{-}\right)={\rm dist}\left(\widetilde{Q}_{0}^{-}, Q_{0}^{+}\right).$$
\end{proof}
 \section{}
 \begin{lem}[Schur Complement Lemma]\label{Su}
	Let $A\in\mathbb{C}^{d_1\times d_1}, D\in\mathbb{C}^{d_2\times d_2}, B\in\mathbb{C}^{d_1\times d_2}, D\in\mathbb{C}^{d_2\times d_1}$ be matrices and
	$$	M=
	\begin{pmatrix}
		A&B\\
		C&D
	\end{pmatrix}.$$
	Assume further that $A$ is invertible and $\|B\|,\|C\|\leq1$. Then we have
	\begin{itemize}
		\item[(1).]
		$$	\det M=\det A\cdot\det S,$$
		where
		$$
		S=D-CA^{-1}B
		$$
		is called the Schur complement of $A$.
		\item[(2).] $M$  is invertible iff $S$ is invertible, and
		\begin{equation}\label{kay}
			\|S^{-1}\|\leq\|M^{-1}\|<4\left( 1+\|A^{-1}\|\right)^2 \left( 1+\|S^{-1}\|\right) .
		\end{equation}
	\end{itemize}
\end{lem}
\begin{proof}[Proof of Lemma \ref{Su}]
	Direct computation shows
	$$M^{-1}=
	\begin{pmatrix}
		A^{-1}+A^{-1}BS^{-1}CA^{-1}&-A^{-1}BS^{-1}\\
		-S^{-1}CA^{-1}&S^{-1}
	\end{pmatrix},$$
	which implies \eqref{kay}.
\end{proof}
\section{}
\begin{lem}\label{even}
	Let $l\in\frac{1}{2}\mathbb{Z}^d$ and let $\Lambda\subset\mathbb{Z}^d+l$ be a finite set which is symmetrical about the origin (i.e.,  $n\in \Lambda\Leftrightarrow-n \in \Lambda$). Then $$\operatorname{det}T(z)_{\Lambda}=\operatorname{det}\left((\cos 2 \pi(z+n\cdot  \omega){\delta_{n,n'}}-E+\varepsilon \Delta\right)_{n\in\Lambda} $$ is an even function of $z$.
\end{lem}
\begin{proof}[Proof of Lemma \ref{even}]
	Define the unitary map
	\begin{align*}
		U_\Lambda:\ell^2(\Lambda)\longrightarrow \ell^2(\Lambda)\ {\rm with}\ (U\phi)(n)=\phi(-n).
	\end{align*}
	Then  $$U_\Lambda^{-1}T(z)_{\Lambda}U_\Lambda=\left( (\cos 2 \pi(z-n\cdot  \omega)\delta_{n,n'}-E+\varepsilon \Delta\right)_{n\in\Lambda} =T(-z)_{\Lambda},$$
	which implies $$\operatorname{det}T(z)_{\Lambda}=\operatorname{det}T(-z)_{\Lambda}.$$
\end{proof}

\bibliographystyle{alpha}

 \end{document}